\documentclass[11pt]{article}
\usepackage[a4paper, total={6in, 8in}]{geometry}
\RequirePackage{amsthm,amsmath,amsfonts,amssymb}
\RequirePackage[authoryear]{natbib}
\RequirePackage[colorlinks,citecolor=blue,urlcolor=blue]{hyperref}
\RequirePackage{graphicx}
\usepackage{subfigure,epsfig,psfrag}
\usepackage{verbatim,float,url,enumerate}
\usepackage{bm,dsfont,color,appendix}
\usepackage{adjustbox}
\usepackage{mathtools}
\usepackage{bbm}
\usepackage{amsthm}
 \usepackage{booktabs}
\usepackage{amsmath}
\usepackage{latexsym,graphicx}
\usepackage{tikz}
\usetikzlibrary{arrows,positioning}
\usepackage{rotating}
\usepackage[linesnumbered,ruled]{algorithm2e}
\usepackage{makecell}
\usepackage{tabularx}
\usepackage{enumitem}
\usepackage[T1]{fontenc}
\usepackage{cfr-lm}
\usepackage{lmodern}
\usepackage[makeroom]{cancel}

\usepackage{xr}
\def\bx{{\bf x}}

\def\real{{\mathbb R}}

\def\b1{{\boldsymbol 1}}

\def\Var{\mathrm{Var}}

\def\sign{{\mathrm{sign}}}

\def\diag{\mathrm{diag}}

\def\bP{\bold P}

\def\bX{\boldsymbol{X}}

\def\bE{\mathbb{E}}

\def\Id{{\rm {\mathbf Id}}}

\def\cov{\rm{cov}}
\def\var{\rm{var}}

\def\mO{\mathcal{O}}
\def\bomega{\bar{\omega}}

\newtheorem{theorem}{Theorem}[section]
\newtheorem{lemma}[theorem]{Lemma}
\newtheorem{proposition}[theorem]{Proposition}
\newtheorem{assumption}[theorem]{Assumption}

\def\mA{\mathcal{A}}
\def\1{\mathbbm{1}}

\usepackage[outline]{contour}
\usepackage{xcolor}

\usepackage{xr-hyper} 
\usepackage{hyperref}
\numberwithin{equation}{section}

\title{$\ell_1$-norm constrained multi-block sparse canonical correlation analysis via proximal gradient descent}
\author{Leying Guan\\ leying.guan@yale.edu}
\date{}                                           

\begin{document}
\maketitle
\begin{abstract}
Multi-block CCA constructs linear relationships explaining coherent variations across multiple blocks of data. We view the multi-block CCA problem as finding leading generalized eigenvectors and propose to solve it via a proximal gradient descent algorithm with $\ell_1$ constraint for high dimensional data. In particular, we use a decaying sequence of constraints over proximal iterations, and show that the resulting estimate is rate-optimal under suitable assumptions. Although several previous works have demonstrated such optimality for the $\ell_0$ constrained problem using iterative approaches, the same level of theoretical understanding for the $\ell_1$ constrained formulation is still lacking. We also describe an easy-to-implement deflation procedure to estimate multiple eigenvectors sequentially. We compare our proposals to several existing methods whose implementations are available on R CRAN, and the proposed methods show competitive performances in both simulations and a real data example.
\end{abstract}

\section{Introduction}
Multi-block canonical correlation analysis (mCCA) generalizes canonical correlation analysis (CCA) to $D$ data blocks for $D> 2$ \citep{kettenring1971canonical}.  There are different types of generalizations \citep{kettenring1971canonical, nielsen2002multiset}, we consider the case where we construct the leading direction to maximize the total of cross-block covariance relative to the total of within block variance, also referred to as the sum of covariance formulation (SUMCOR).

Variants of mCCA have been applied to various applications for dimension reduction and exploration with high dimensional data, including joint blind source separation \citep{li2009joint}, multi-omics data integration \citep{subramanian2020multi,rodosthenous2020integrating},  neuroimaging \citep{sui2012review} and others. Researchers seek leading mCCA directions that capture coherent variations across blocks in these applications. Due to the high-dimensional nature of the datasets, proper regularization of the mCCA loadings $\beta$ is needed. In the two-block setting, researchers have developed different sparse CCA methods and provided theoretical analyses for these procedures \citep{witten2009penalized, hardoon2011sparse,  chen2013sparse, suo2017sparse, gao2017sparse}. Different mCCA type procedures have been proposed for the same goal. Theoretical aspects are less studied for these methods. In a recent independent work \citep{gao2021sparse}, the authors cast the mCCA problem as a generalized eigenvalue problem with $\ell_0$ penalty and investigated its statistical convergence.

We similarly  view mCCA problem as a special case of the generalized eigenvalue problem and study the SUMCOR generalization with a constraint on the total single-block variance. For block $d$, we let $X_{[d]}$ be its feature vector, and let $\beta_{[d]}$ its feature loadings, we are interested finding $\beta_{[d]}$ such that
\begin{equation}
\label{eq:formulation1}
\max_{\beta}\sum_{d_1=1}^D\sum_{d_2=1}^D \cov(X_{[d_1]}\beta_{[d_1]},X_{[d_2]}\beta_{[d_2]}), \quad s.t. \; \sum_{d=1}^D\var(X_{[d]}\beta_{[d]})=1.
\end{equation}
Compared with another  popular constraint where  people require $\var(X_{[d]}\beta_{[d]})=1$ for all $d=1,\ldots, D$, (\ref{eq:formulation1}) impose less restriction on the per-block contribution and estimate it adaptively from the data, which makes the procedure more robust to the existence of uninformative blocks.

Let $\hat\Sigma$ be the empirical covariance matrix using the aggregated features $X$ from all blocks, and $\hat\Lambda$ be a block diagonal matrix with the $d^{th}$ block being the empirical covariance matrix for block $d$. In high dimensions, a regularized empirical estimate can be acquired (up to a scaling factor) by maximizing the empirical Rayleigh quotient under norm constraints:
\begin{equation}
\label{eq:formulation2}
\max_{\beta} \frac{\beta^\top\hat\Sigma\beta}{\beta^\top \hat\Lambda \beta},\quad s.t.\; \|\beta\|_2=1, \;\|\beta\|_{q}\leq L,
\end{equation}
where $\|.\|_q$ is the $\ell_q$ norm of $\beta$. Such formulations have been used for general sparse generalized eigenvalue problem. Several previous work have  considered it at $q = 0$ and proposed iterative updating rules  constraint and studied their convergence \citep{sriperumbudur2011majorization, tan2018sparse, cai2020inverse}.   For example, in \cite{tan2018sparse}, the authors proposed rifle for estimating leading eigenvector with a $\ell_0$-norm constraint by a gradient descent procedure, and proved that the procedure is rate optimal under suitable assumptions.  In \cite{cai2020inverse} and \cite{gao2021sparse}, the authors considered different truncated procedures to solve the $\ell_0$-norm constrained\slash penalized problem with multiple directions by estimating the space spanned by them  jointly.   Researchers have also investigated the performance of the such a problem with $\ell_1$ constraint \citep{gaynanova2017penalized, jung2019penalized},  however,  its statistical properties  are not as thoroughly studied.  In practice, $\ell_1$ constraints are perhaps still more frequently used by researchers who are working on multi-omics studies and considering the mCCA-type dimension reductions \citep{witten2009extensions, tenenhaus2011regularized, meng2014multivariate, tenenhaus2017regularized, kanatsoulis2018structured}.

In this paper, we consider solving the non-convex problem (\ref{eq:formulation2}) at $q = 1$ via proximal gradient descent, where we can find the optimal solution conveniently at each proximal iteration. More specifically, we propose a novel updating scheme and solve a sequence of proximal problems with decaying  $\ell_1$ constraints. We summarize three main contributions of this paper as below:
\begin{enumerate}
\item  We propose solving the $\ell_1$-constrained Rayleigh quotient problem with a novel iterative procedure using proximal gradient descent with properly designed decaying bounds. 
\item We show that the proposed procedure produces a sequence of mCCA direction estimates containing rate-optimal solutions under similar assumptions for the optimality of $\ell_0$ constrained/penalized procedures. We can identify a rate-optimal estimation from this sequence by considering a penalized objective.
\end{enumerate}

We compare our proposal to several available R packages for solving mCCA problems. Our proposal, especially the $\ell_1$-norm bounded procedure, demonstrates competitive performance in our empirical studies. We observe the estimation from our proposal to outperform that from the $\ell_0$ constrained formulation when the problem becomes less sparse or the signal-to-noise ratio decreases (see Section \ref{sec:sim} for detailed comparisons). We then apply different methods to the TCGA cancer data set, and our proposals again outperform the competing methods for extracting coherent cross-block information with its leading directions. We provide an R package {\it msCCA} on R CRAN for using the proposed procedures.

The article is organized as follows. In Section \ref{sec:method}, we give details of our method and the main algorithm. We provide theoretical guarantees of our proposals in section \ref{sec:convergence} and extend our proposals to the estimation of multiple mCCA directions through a sequential deflation procedure in Section  \ref{sec:deflation}. Finally, we compare different methods for estimating mCCA directions using simulated data in Section \ref{sec:sim} and apply these methods to TCGA data in Section \ref{sec:tcga}.

\section{Multi-block Sparse CCA  via proximal gradient descent}
\label{sec:method}

\subsection{Multi-block Sparse Canonical Analysis}
\label{subsec:msCCA}
Let $X\in \real^p$ be a concatenation of variables from $D$ different blocks with $p=\sum_{d=1}^D p_d$ and $p_{d}$ being the number of features from the $d^{th}$ block. Let $[d] =\{j: p_1+\ldots+p_{d-1} +1\leq j \leq p_1+\ldots+p_{d-1} +p_d\}$ be the subset of index for features from block $d$ and $X_{[d]}\in \real^{p_d}$ be features from block $d$. We are interested in finding the aggregated direction $\beta\in \real^p$, such that $Z = X\beta$ captures the largest amount of common variation shared across different blocks:
\begin{equation}
\label{eq:mCCA1}
\max_{\beta} \; \var(Z),\; \; s.t.\; \sum_{d=1}^D\var(X_{[d]}\beta_{[d]})=1.
\end{equation}
When $D=2$, the solution of (\ref{eq:mCCA1}) is the concatenated leading canonical direction.  (\ref{eq:mCCA1}) is one type of SUMCOR generalization, and another widely used generalization is to let $\Var(X_{[d]}\beta_{[d]})=1$ for each $d = 1,\ldots, D$:
\begin{equation}
\label{eq:mCCA2}
\max_{\beta} \; \var(Z),\; \; s.t.\; \var(X_{[d]}\beta_{[d]})=1,\;\mbox{for all }d=1,\ldots, D.
\end{equation}
In the case where the source of variation is not shared by all blocks, (\ref{eq:mCCA2}) can be more influenced by uninformative blocks compared to (\ref{eq:mCCA1}), because (\ref{eq:mCCA2}) requires different blocks to have the same variation after projecting on the estimated direction.

We call $\beta$ the aggregated mCCA direction.   Let $\bX\in \real^{n\times p}$ be our observed data for $n$ samples, and $\bX_{[d]}\in \real^{n\times p_d}$ be the observations for block $d$. Suppose all columns in $\bX$ are demeaned for convenience. The direction $\beta$ can be estimated empirically, e.g., 
\begin{equation}
\label{eq:mCCA3}
\max_{\beta} \; \beta^\top \hat\Sigma\beta,\; \; s.t.\; \sum_{d=1}^D \beta_{[d]}^\top\hat\Lambda_{[d]}\beta_{[d]}=1,
\end{equation}
where $\hat\Sigma = \frac{\bX^\top \bX}{n}$ is the empirical covariance matrix for $X$ and $\hat\Lambda_{[d]} \coloneqq \hat\Sigma_{[d][d]}$ is the empirical covariance for $X_{[d]}$.  Equivalently, we can find the mCCA direction maximizing the Rayleigh quotient:
\begin{equation}
\label{eq:mCCA4}
\max_{\beta} f(\beta)\coloneqq \frac{\beta^\top\hat\Sigma\beta}{\beta^\top \hat\Lambda \beta},\quad s.t.\; \|\beta\|_2=1.
\end{equation}
When  $p$ is large,  proper forms of regularization are often needed for generalizable solutions.  For example, we can introduce sparsity regularization on $\beta$ to better estimate the mCCA direction, resulting in the multi-block sparse canonical correlation analysis (msCCA).  Here, we consider estimating (\ref{eq:mCCA4}) with $\ell_1$ constraint. For instance, let $L_*$ be the $\ell_1$ norm of $\xi_1$, the corresponding population mCCA solution for (\ref{eq:mCCA4}), and we may want to consider (\ref{eq:mCCA5}) below.
\begin{equation}
\label{eq:mCCA5}
\max_{\beta} f(\beta),\quad s.t.\; \|\beta\|_2=1, \;\|\beta\|_{1}\leq L_*.
\end{equation}
$L_*$ is unknown. This is not a problem, since we will iteratively update our estimate of the leading mCCA direction via a proximal gradient descent and with a suitably designed decaying sequence of $ \ell_1$  bounds $\{L_1,\ldots, L_{\infty}\}$ for some large $L_1$ and small $L_{\infty}$. The intuition is that when this sequence is refined enough and contains values close enough to $L_* = \|\xi_1\|_1$ whose corresponding estimates could be of high quality. We will give more details about choices of $\{L_1,\ldots, L_{\infty}\}$ in Section \ref{sec:convergence}.

\subsection{msCCA via proximal gradient descent}
\label{subsec:proximal}
We estimate mCCA directions via proximal gradient descent.  At iteration $(t+1)$, let $\beta_t$ be the our current estimate, and $g(\beta_{t})$ be the gradient of $f(\beta)$ evaluated at $\beta_{t}$, which takes the form
\begin{equation}
\label{eq:gradient1}
g(\beta_t)=\frac{2}{\beta_t^\top\hat\Lambda \beta_t }\left(\hat\Sigma\beta_t - f(\beta_t)  \hat\Lambda\beta_t\right).
\end{equation}
Let $\eta$ be some user-specified small step-size, we consider the following proximal problem for updating our estimate:
\begin{equation}
\label{eq:proximal1}
\min_{\|\beta\|_2=1,\|\beta\|_1 \leq L_{t+1}}\left\{\frac{f(\beta_t)}{\eta \beta_t^\top\hat\Lambda \beta_t}\|\beta-\beta_t\|_2^2-g(\beta_t)^\top \beta\right\}.
\end{equation}
Let $\beta_{t+1}$ the solution to (\ref{eq:proximal1}) at iteration $(t+1)$. In (\ref{eq:proximal1}), the coefficient in front of the quadratic term $\|\beta-\beta_t\|_2^2$ is designed to avoid overshooting along the gradient direction.

Set $\theta= \beta_t+\frac{\eta}{f(\beta_t)}(\hat\Sigma - f(\beta_t)\hat\Lambda)\beta_t$ as the proximal target, which is the minimizer of (\ref{eq:proximal1}) without the norm constraint on $\beta$. $\beta_{t+1}$ is then the projection of the proximal target $\theta$ onto the space $\{\beta:\|\beta\|_2=1,\; \|\beta\|_1\leq L_{t+1}\}$.  Algorithm \ref{alg:proximal} provides details of  our proposal for estimating the leading mCCA direction.
\begin{algorithm}
\label{alg:proximal}
\caption{msCCA via proximal gradient descent}
\KwData{data matrix $\bX$, initial estimate $\beta_{0}$, proximal descent step size $\eta$,  and a decaying sequence of norm constraints $\sqrt{p}\geq L_1\geq L_2\geq \ldots \geq L_{\infty}=L\geq 1$.}

\KwResult{$\{(\hat r, \hat\beta, Z)\}$: estimated mCCA coefficient, direction, and projection onto the CCA direction for different blocks.}

Set $t = 0$.

\While{Not converge}{

Calculate the proximal response:
\begin{equation}
\label{eq:theta}
\theta = \beta_{t}+\frac{\eta}{f(\beta_t)}\left(\hat \Sigma -f(\beta_t) \hat \Lambda\right)\beta_t
\end{equation}

Update the aggregated CCA directions:
\begin{equation}
\label{eq:update}
\beta_{t+1} =\arg\min_{\|\beta\|_2^2=1, \|\beta\|_1\leq L_{t+1}}\|\beta-\theta\|_2^2
\end{equation}

$t = t+1$;
}
Let $t^*$ be the selected iteration number, e.g., selected via cross-validation, and  $\hat\beta$ represent our estimate $\beta_{t^*}$  from iteration $t^*$.  We set $\hat r = \frac{\hat\beta^\top \hat\Sigma \hat\beta}{\hat\beta^\top \hat\Lambda\hat\beta}$, and $Z\in \real^{n\times D}$ with $Z_d = \frac{1}{\sqrt{n}}\bX_{[d]} \hat\beta_{[d]}$ for $d = 1,\ldots, D$.
\end{algorithm}
We can find  $\beta_{t+1}$ in (\ref{eq:update}) with a convenient numerical subroutine.  Let $|\theta|_{(1)}\geq |\theta|_{(2)}\geq \ldots \geq |\theta|_{(p)}$ be the ordered version of $\{|\theta_j|\}$. We consider two cases:
\begin{itemize}
\item  When $|\theta|_{(1)} > |\theta|_{(\lceil L^2_{t+1}\rceil )}$: As a result of the first part of  Proposition \ref{prop:proximal_l1_update0}, we can solve (\ref{eq:update})  by soft-thresholding $\theta_{t+1}$ at a proper $c$ and rescaling, with $c$ found numerically via binary search as a result of Proposition \ref{prop:proximal_l1_update}.
\item  When $|\theta|_{(1)} = |\theta|_{(\lceil L^2_{t+1}\rceil )}$: we can not use the aforementioned soft-thresholding strategy directly. However,  we can take $\beta_{t+1}$ as any unit vector $\beta$ satisfying the second part of Proposition \ref{prop:proximal_l1_update0}, and it will be an optimal solution to (\ref{eq:update}).
\end{itemize} 

\begin{proposition}
\label{prop:proximal_l1_update0}
(1) When $|\theta|_{(1)} > |\theta|_{(\lceil L^2_{t+1}\rceil )}$, $\beta_{t+1} = \frac{\tilde\beta}{\|\tilde\beta\|_2}$ is the solution to (\ref{eq:update}) where $\tilde\beta = \sign(\theta)\cdot [|\theta| - c]_{+}$ is the soft-thresholded $\theta$  at $c$, and $c$ is the smallest non-negative value such that $\frac{\|\tilde\beta\|_1}{\|\tilde\beta\|_2}\leq L_{t+1}$. (2) When $|\theta|_{(1)} = |\theta|_{(\lceil L^2_{t+1}\rceil )}$, set $\beta_{t+1}=\sign(\theta)\cdot \tilde\beta$  for any non-negative unit  vector $\tilde\beta\geq 0$ with $\|\tilde\beta\|_1 = L_{t+1}$ and $\tilde\beta_j = 0$ if $|\theta_j|< |\theta|_{(\lceil L^2_{t+1}\rceil )}$. Then, $\beta_{t+1}$ is an optimal solution to (\ref{eq:update}).
\end{proposition} 
\begin{proposition}
\label{prop:proximal_l1_update}
The quantity $ \frac{\|[|\theta|-c]_+\|_1}{\|[|\theta|-c]_+\|_2}$ is continuous and non-increasing  in $c$ for all  $0\leq c< |\theta|_{(1)}$.
\end{proposition}


\section{Statistical Convergence}
\label{sec:convergence}
In this section, we study the theoretical guarantee of Algorithm \ref{alg:proximal} in the regime where $n, p\rightarrow \infty$. We define $\xi_j$ as the $j^{th}$ population mCCA direction, normalized to have norm $\|\xi_j\|_2=1$ for $j = 1,\ldots, p$:
\[
\xi_j = \arg\max_{\|\beta\|_2^2=1}\frac{\beta^\top\Sigma\beta}{\beta^\top\Lambda\beta},\; s.t.\; \xi_j^\top \Lambda \xi_{j'} = 0, \;\mbox{for all } j' < j.
\]
For such a problem, it is known that the optimal rate  for mean squared error in estimating $\xi_1$ is $\frac{s\ln p}{n}$, assuming $s$-sparsity of $\xi_1$ and some regularity conditions \citep{cai2013sparse, gao2015minimax, tan2018sparse}. We show  that we can find such a rate-optimal $\hat\beta$ from Algorithm \ref{alg:proximal} given a reasonable  sequence $\{L_t\}$, an informative initial guess $\beta_0$ and under Assumptions \ref{ass:regularity}-\ref{ass:sparsity}.

Let $\rho_j = \frac{\xi_j^\top\Lambda \xi_j}{\xi_j^\top\Sigma\xi_j}$ be the population mCCA correlation coefficient for $\xi_j$ and $\delta_t = 1-|\beta^\top_t\xi_1|$ to measure the discrepancy  between our estimate at iteration $t$ from Algorithm \ref{alg:proximal} and the leading population mCCA direction $\xi_1$. Features are standardized to have mean 0 and variance 1. We also let $\lambda_H^{\max}$, $\lambda_H^{\min}$ represent the largest and smallest eigenvalues for some matrix $H$.

\begin{assumption}
\label{ass:regularity}
The aggregated data $X$ is  multivariate Gaussian with $X\sim \mathcal{N}(0, \Sigma)$ and $\Sigma_{jj}=1$ for all $j=1,\ldots, p$. The covariance for each single block  has bounded smallest and largest eigenvalues: Defining $\Lambda_{[d]}\coloneqq \Lambda_{[d][d]}= \Sigma_{[d][d]}$, then $\frac{1}{M}\leq \min_{d}\lambda_{\Lambda_{[d]}}^{\min}\leq \max_{d}\lambda_{\Lambda_{[d]}}^{\max} \leq M$ for some  constant $M$.
\end{assumption}

\begin{assumption}
\label{ass:eigen_gap}
The leading mCCA correlation coefficient $\rho_1$ is bounded away from 1 with $\rho_1 \geq 1+C$, and the gap between $\rho_1$ and $\rho_2$ is sufficiently large with $\rho_1 - \rho_2 \geq \gamma \rho_1$ for some constants $C,\; \gamma>0$.
\end{assumption}

\begin{assumption}
\label{ass:sparsity}
The leading mCCA direction is sparse: let $s = \|\xi_1\|_0$, we have $\frac{s\ln p}{n}\rightarrow 0$.
\end{assumption}

Theorem \ref{thm:thm1} gives the statistical guarantee  on the estimation errors over iterations using the proposed procedure with some large $L_0$ and small $L_{\infty}\leq \|\xi_1\|_1$, e.g., $L_{\infty} = 1$.
\begin{theorem}
\label{thm:thm1}
Let $1>c_0>0$ be any small constant,  $0<\eta \leq \frac{1-c_0}{2M[M+3]}$ be a user-specified constant step size.  Let $c_{B_1}, \;c_{B2}$ be any positive constants and $\{L_0,\ldots, L_{\infty}\}$ be any decaying sequence with $L_0 = (1+c_{B_1})\sqrt{s}$, and $L_{t} = L_{\infty}+(1-c_{B_2}\eta\sqrt{\frac{s\ln p}{n}})^{t}(L_{0}-L_{\infty})$. Set $\nu = \frac{(1-c_0)\gamma^2}{M}$ and $B_t = L_t -\|\xi_1\|_1$. Define 
\[
T^*= \max\{t: B_t\geq  \frac{4c_{B_2}(1+c_{B_1})}{\nu}\sqrt{\frac{s^2\ln p}{n}}\}.
\]
Under Assumptions \ref{ass:regularity}-\ref{ass:sparsity},  there exists some sufficiently small and large  constants $\psi_1$ and $\psi_2$  such that  if the initial guess satisfies $\|\beta_0\|_1\leq L_0$ and $\delta_0 \leq \psi_1$,  we can upper bound the estimation error $\delta_t$ for all $t\leq T^*$ with probability approaching 1 as $n\rightarrow\infty$:
\begin{equation}
\label{eq:estimation_main1}
\delta_t\leq \psi_2\left(\delta_0(1-\eta \nu)^t+\frac{s \ln p}{n}+\sqrt{\frac{s\ln p}{n}}\left(\frac{B_t}{\sqrt{s}}+\frac{B_t^2}{s}\right)\right),\; \mbox{for all\;} t\leq T^*.
\end{equation}
\end{theorem}
As a direct application of Theorem \ref{thm:thm1}, $\beta_{T^*}$ is rate-optimal.
\begin{lemma}
\label{lem:mainLemma}
Consider the same set-up as in Theorem \ref{thm:thm1}. Then, let $\psi_2$ be a sufficiently large constant, we have
\[
\lim_{n\rightarrow\infty} \bP(\delta_{T^*}\leq \psi_2\frac{s\ln p}{n})=1.
\]
\end{lemma}
We have shown that the existence of an optimal estimate in $\{\beta_t, 1\leq t\leq T^*\}$. Since $T^*$ is  unknown, can we identify an rate-optimal estimate in the produced sequence?  It is straightforward to achieve when we have an independent validation set since we can evaluate $\beta_t$ fairly with the validation set and pick a best one. Theoretically, one can also identify such a rate optimal solution by considering the penalized objective (\ref{eq:penalized_obj}):
\begin{equation}
\label{eq:penalized_obj}
\underline{f}_{\tau}(\beta)=f(\beta)-\tau \rho_1\sqrt{\frac{\ln p}{n}}(\|\beta\|_1+\frac{c_2\|\beta\|_1^2}{ L_0}),
\end{equation}
for any positive constant $c_2$ and a sufficiently large constant $\tau$. The estimate $\beta_{t}$ is guaranteed to be rate optimal with high probability if $\underline{f}_{\tau}(\beta_t)$ is maximized at iteration $t$.
\begin{theorem}
\label{thm:thm2}
Consider the same set-up as in Theorem \ref{thm:thm1}.  Let $t^* =\arg\max_{t} \underline{f}_{\tau}(\beta_t)$ be the iteration achieving the largest penalized objective. Then, when $\tau$ is a sufficiently large constant,  $\lim_{n\rightarrow\infty}\bP(\delta_{t^*}\leq \psi_2\frac{s\ln p}{n}) = 1$ for a sufficiently large constant $\psi_2$.
\end{theorem}
From Theorem \ref{thm:thm2}, we can pick the $\beta_t$ with the largest penalized objective $\underline{f}_{\tau}(\beta_t)$, it is rate-optimal with high probability as $n$ becomes large. In practice, however, we still recommend using cross-validation for selecting iteration number $t$ because it is unclear what is a good value for $\tau$ or $c_2$ when we care about finite sample performance.

\subsection{Initialization with theoretical guarantee}
\label{subsec:init}
Like other iterative updating methods for the generalized eigenvalue problem, msCCA requires an initial guess. One can initialize the estimate by solving  some relaxed convex problems. The attractive aspect of such methods is that they provide initial guess with  statistical guarantees as $n$ becomes large and under suitable assumptions.  Let $\|.\|_{*}$ and $\|.\|_{op}$ be the nuclear norm and the operator norm respectively. We consider the type of formulation used in in \cite{tan2018sparse} and \cite{gao2021sparse} for a rank 1 model, where the authors initialize the problem using the largest eigenvalue of $\hat{P}$, and $\hat{P}$ is the solution to the problem below ($S_{+}^{p\times p}$ denote the space of all $p\times p$ symmetric and semi-positive definite matrix):
\begin{equation}
\label{eq:convex1}
\min_{P\in S_+^{p\times p}} -tr(\hat\Sigma P)+\tau \sum_{i,j}|P_{ij}|, \; s.t. \; \|\hat\Lambda^{\frac{1}{2}}P\hat\Lambda^{\frac{1}{2}} \|_{*}=1, \;\|\hat\Lambda^{\frac{1}{2}}P\hat\Lambda^{\frac{1}{2}}\|_{op}\leq 1,
\end{equation}

 Let $\hat P$ be the solution from (\ref{eq:convex1}) and let $\hat\beta$ be its leading eigenvector. We can initialize our estimate as $\beta_0=\frac{\tilde \beta}{\|\tilde\beta\|_2}$ where $\tilde\beta$ is the truncated version of $\hat\beta$ that keeps only  entries in $\hat\beta$ with $ks$ largest magnitude values, for any integer $k \geq  1$. Then, $\beta_0$ satisfies requirements in Theorem \ref{thm:thm1} as an initial guess with high probability for large $n$, which can be shown using  Lemma 12 from \cite{yuan2013truncated} and arguments for Theorem 4.3 from  \cite{gao2021sparse}. 
\begin{lemma}
\label{lem:init}
Suppose that Assumptions \ref{ass:regularity} -\ref{ass:sparsity} hold, and further, $\sqrt{\frac{s^2\ln p}{n}}\rightarrow 0$ as $n\rightarrow\infty$. When $\tau \geq CM \rho_1\sqrt{\frac{\ln p}{n}}$ for a sufficiently large universal constant $C$,  we have
\[
1-|\xi_1^\top\beta_0|\leq \psi_1,\; \|\beta_0\|_1\leq (1+c_{B_1})\sqrt{s}
\]
happens with probability approaching 1 as $n\rightarrow \infty$ for any positive constants  $\psi_1$ and $c_{B_1}$.
\end{lemma}

\section{Estimation of multiple CCA directions}
\label{sec:deflation}
We can estimate multiple directions sequentially.   To motivate our procedure, we first consider the low dimensional setting where we adopt the empirical mCCA estimations without sparsity constraint. Let $\hat\beta^\ell$ be the estimated $\ell^{th}$ direction for $\ell = 1,\ldots, p$.  Given the first $k$ mCCA estimated directions, we estimate the $(k+1)^{th}$ direction under the constraint that its associated projection $Z_{k+1} = X\hat\beta$ is uncorrelated with precedent projections $Z_{\ell} = X\hat\beta_{\ell}$ for $\ell < k+1$:
\begin{equation}
\label{eq:deflate1}
\max \beta^\top \hat\Sigma\beta, \; s.t.\;\beta^\top\hat\Lambda\beta = 1,\;  \beta^\top \hat\Sigma \hat\beta^{\ell} = 0, \;\mbox{for all}\; \ell < k+1.
\end{equation}
We may  drop  these orthogonality constraints and consider the following problem:
\begin{equation}
\label{eq:deflate2}
\max \beta^\top \frac{\tilde X_{k+1}^\top \tilde X_{k+1}}{n}\beta, \; s.t.\; \beta^\top \hat\Lambda\beta = 1.
\end{equation}
Here, $\tilde X_{k+1} $ is some deflated version of $X$, and is constructed sequentially as below:
\begin{equation}
\tilde X_{\ell+1} = (\Id - \frac{\tilde Z_\ell \tilde Z_\ell^\top}{\|\tilde Z_\ell\|_2^2})\tilde X_\ell, \;\mbox{for all }1\leq \ell \leq k,
\end{equation}
with $\tilde X_\ell$ the deflated data matrix for component $\ell$ and $\tilde Z_\ell= \tilde X_\ell \hat\beta_\ell$. $\tilde X_1$ is the same as the original data matrix $X$.  The two formulations (\ref{eq:deflate1}) and (\ref{eq:deflate2}) are equivalent.
\begin{proposition}
\label{prop:deflate1}
Let $\hat\beta_1,\ldots, \hat\beta_p$ be the eigenvectors to the original problem with eigenvalues $\hat\rho_1\geq \hat\rho_2 \geq \ldots \geq \hat\rho_p$.  Then, at step $(k+1)$, the deflated problem has generalized leading eigenvector and eigenvalue pair $(\tilde \beta_1, \tilde \rho_1)$ with $\tilde\beta_1 = \hat\beta_{k+1}$, and $\tilde\rho_1 =\hat\rho_{k+1}$. 
\end{proposition}
This leads to the deflation procedure in Algorithm \ref{alg:deflation} where we sequentially estimate multiple mCCA directions at line 6 and perform deflation at line 4.

\begin{algorithm}[H]
\caption{subsequent direction estimation via deflation}
\label{alg:deflation}

Initialize $\tilde\bX = \bX$ be as the adjusted feature matrix.

\For{$k = 1,\ldots, K$}{

\If{$k > 1$}{
	$\tilde\bX  = (\Id - \frac{\tilde{Z}_{k-1} \tilde{Z}_{k-1}^\top}{\|\tilde{Z}_{k-1}\|_2^2})\tilde\bX$
}

Apply Algorithm \ref{alg:proximal} with $\hat\Sigma$ replaced by $\frac{\tilde \bX^\top \tilde\bX}{n}$, and obtain the $k^{th}$ estimated mCCA direction $\hat\beta_k$.

Construct $\tilde Z_{k}$ be the aggregated mCCA score after adjusting for previous directions and 
\[
\tilde Z_k = \tilde X_{k}\hat\beta_{k}.
\]
}
\end{algorithm}

The construction of $\tilde X$ and (\ref{eq:deflate2}) is equivalent to the Schur complement deflation procedure proposed in \cite{mackey2008deflation} for the sparse PCA problem.  The Schur complement deflates the covariance $\Sigma$ as below. Let $\tilde\Sigma_{k+1}$ be the deflated covariance for estimating the $(k+1)^{th}$ component. The Schur complement deflation constructs $\tilde\Sigma_{k+1}$ sequentially based on (\ref{eq:deflate3}):
\begin{equation}
\label{eq:deflate3}
\tilde\Sigma_{\ell+1} = \tilde\Sigma_\ell - \frac{\tilde\Sigma_\ell\hat\beta_\ell\hat\beta_\ell^\top \tilde\Sigma_\ell}{\hat\beta_\ell^\top\tilde\Sigma_\ell\hat\beta_\ell}, \;\mbox{for all }1\leq \ell \leq k,
\end{equation}
where $\tilde\Sigma_\ell$ is the deflated covariance matrix for $\ell^{th}$ eigenvector.  The reason why we deflates $X$ instead of $\Sigma$ is because the former reduces the computation from $\mathcal{O}(p^2)$ to $\mathcal{O}(np)$ when $p \gg n$ and does not require saving a large covariance matrix.
\begin{proposition}
\label{prop:deflate2}
For any $\hat\beta_1,\ldots, \hat\beta_{k+1}$, we have $\frac{\tilde X_{k+1}^\top\tilde X_{k+1}}{n} = \tilde \Sigma_{k+1}$, with $\tilde X_{k+1}$ formed based on (\ref{eq:deflate2}) and $\tilde\Sigma_{k+1}$ based on (\ref{eq:deflate3}).  
\end{proposition}
In high dimensions and with an additional constraint on the $\ell_1$ norm, (\ref{eq:deflate1}) and (\ref{eq:deflate2}) are no longer equivalent to each other, and we are not likely to end up with exact orthogonal projections. Despite that, \cite{mackey2008deflation} compared different deflation procedures for sparse PCA and concluded that methods like Schur complement deflation preserve some desirable properties compared to the naive Hotelling deflation. The Schur complement deflation does guarantee that  (1) $\tilde\Sigma_{k+1}$ is semi-positive definite, (2) $\hat\beta_{\ell}^\top\tilde\Sigma_{k+1} = 0$ for all $\ell < k+1$. 
\section{Simulation studies}
\label{sec:sim}
%
%
%
%
%
%
%
%

In this section, we compare five different methods for multi-block CCA estimations in simulations:
\begin{itemize}
\item msCCA1: the proposed msCCA estimation with $\ell_1$ constraint, combined with the proposed deflation procedure  for estimating multiple directions.
\item rifle(seq) or rifle:  rifle \citep{tan2018sparse} combined with the proposed deflation procedure  for estimating multiple directions.
\item pma: multi-block generalization of penalized matrix analysis \citep{witten2009penalized}.
\item rgcca: regularized generalized CCA \citep{tenenhaus2011regularized}.
\item  sgcca: sparse generalized CCA \citep{tenenhaus2014variable}.
\end{itemize}
In Section \ref{subsec:init}, we show that an initial guess from solving a relaxed convex problem (after thresholding) is a sufficiently good initializer in an asymptotic sense. Both rifle and msCCA1 offer rate optimal estimations with such initializers. However, we do not use it here in our numerical experiments because  it is computationally expensive. Instead, we initialize both msCCA1 and rifle with a non-sparse mCCA estimation using a subset of selected features that exhibit high across block correlations. More details can be found in Appendix \ref{app:algorithm}.

We consider the simulation setup where we have $D = 4$ blocks with the single block dimension fixed at $p_d = 500$ for $d=1,\ldots, D$,  and $K = 3$ mCCA components with $\rho_j > 1$. For all estimation methods, we estimate only two mCCA directions. We consider two scenarios described below.

\noindent\textbf{Scenario A: mCCA direction estimation with non-informative blocks.} In this scenario, only the first two blocks are correlated and contribute to the population leading mCCA directions. That is: for $d\neq d'$,
\[
\Sigma_{[d][d']} = 0,\; \mbox{if }  d\notin\{1,2\}\mbox{ or }d'\notin\{1,2\}.
\]
For $d,d'\in \{1,2\}$, we construct their covariance as $\Sigma_{[d][d']} = \Lambda_{[d]}U_{[d]} \Gamma U_{[d']}^{\top}\Lambda_{[d']}$, where $U_{[d]}\in \real^{p_d\times K}$ satisfies $U_{[d]}^\top\Lambda_{[d]} U_{[d]} = \Id$ for $d=1,2$, and  $\Gamma\in \real^{K\times K}$ is the diagonal matrix with diagonal entries $\tilde\rho_k\in (0 ,1)$. We fix $\tilde\rho_k = 0.9 - (k-1)/5$ for $k=1,\ldots,K$. Under this set-up, the leading mCCA direction $\xi_k$ is going to be proportional to $U_k$ with $\rho_k = \tilde\rho_k+1$ for $k=1,\ldots, K$. 

\noindent\textbf{Scenario B: mCCA direction estimation without non-informative blocks.} In this scenario, all four blocks are correlated and contribute to the leading mCCA directions.   We generate $U_{[d]}$ the same way as in scenario A for all $d=1,\ldots, 4$, and let $\Sigma_{[d][d']} = \Lambda_{[d]}U_{[d]} \Gamma U_{[d']}^{\top}\Lambda_{[d']}$ for all $d\neq d'$, and $\Gamma \in \real^{K\times K}$ is a diagonal matrix with diagonal elements $\tilde\rho_{k} = 0.9 - (k-1)/5$, $k=1,\ldots, K$. Under this set-up, the leading mCCA direction $\xi_k$ is also proportional to $U_k$ with $\rho_k = 3\tilde\rho_k+1$ for $k=1,\ldots, K$. 

\noindent Inside both scenarios,  different types of $\Lambda_{[d]}$ and sparsity levels in $U_{[d]}$ are considered:
\begin{itemize}
\item We consider three different types of $\Lambda_{[d]}$: (1) identity matrix, (2) spiked covariance matrix $\Lambda_{[d]} = \sum_{k=1}^3 \lambda_k u_{k}u_{k}^T+\Id$ for $\lambda_{k} = 5$ and $k=1, 2, 3$.  (3) Toeplitz with $\Lambda_{[d]}(i,j) = 0.3^{|i-j|}$ for entry $(i,j)$ in the $d^{th}$ block $\Lambda_{[d]}$.  All $\Lambda_{[d]}$ are normalized to make the diagonal entries be 1.

\item  We set the sample size $n\in \{300, 1000\}$, and the underlying per-block sparsity $s\in \{1, 5, 15\}$. The non-zero entries in $U_{[d]}$ randomly generated from $\mathcal{N}(0,1)$ and then normalized with respect to $\Lambda_{[d]}$ to make $U_{[d]}^\top \Lambda_{[d]} U_{[d]} = \Id$. Since $\xi_k$ corresponds to columns in $U$ for $k=1,\ldots, K$, sparsity in $U$ also indicates sparsity in the leading mCCA directions.
\end{itemize}
We evaluate the estimation quality looking at two aspects:
\begin{itemize}
\item  The achieved multi-block canonical correlations using an independent test data with 2000 samples: Since different estimated mCCA directions are not necessarily orthogonal to each other with respect to test observations, we deflate the second direction estimated from different methods as described in section  \ref{sec:deflation}. The larger the achieved deflated  multi-block canonical correlation is, the better the estimation approach is. 
\item Accuracy of the aggregated projection: Another evaluation we can look at is the accuracy of the aggregated projection. Let $Z_{\ell} = X\xi_{\ell}$ and $\hat Z_{\ell} = X\hat\beta$, we measure the quality by the remaining variance of $Z_{\ell}$ after regressing out the estimated $\hat Z$. The smaller the residual variance is, the better the estimation approach is.
\end{itemize}
Table \ref{tab:sim1}-\ref{tab:sim3} show achieved mCCA correlation using different methods for the first two leading directions, with  different within-block covariance structures for scenarios A and B.  Table \ref{tab:sim4}-\ref{tab:sim5} show the remaining variance of the true aggregated projections after regressing out the estimated ones. All simulation results are averaged over 20 random repetitions. For each entry in the tables, it gives the mean correlations or residual variances, with their standard deviations given in the parenthesis. The top two procedures are colored black, with the best procedure in bold for different simulation settings. Other procedures are in gray.

For both scenarios, msCCA1 and rifle are much better than pma, rgcca and sgcca in our simulations.  Scenario A is harder to estimate than scenario B, and the other three methods can extract very little useful information even when $n = 1000$.  Both msCCA1 and rifle have deteriorated performance  as we (1) increase the the number of non-zero features,  (2)decrease the sample size, or (3) decrease the signal-to-noise ratio, e.g., comparing the estimations of the first direction, and the second direction and comparing scenario A to scenario B. Compared to rifle,  msCCA1  has comparable performance to rifle when the problem is easier, but is on average better when the number of non-zero entries in mCCA directions $\xi_k$ increases and when the signal is weaker.

\begin{table}
\centering
\caption{Achieved mCCA correlations after deflation with identity within-block covariance.}
\label{tab:sim1}
\begin{adjustbox}{width=.75\textwidth}
\begin{tabular}{|ll | lllll|}
  \hline
identity & direction1 & msCCA1 & rifle(seq) & pma & sgcca & rgcca \\   \hline
  A & (300,1) & \textbf{1.87(0.008)} & 1.74(0.073) & \color{gray}{1.05(0.024)} & \color{gray}{1.05(0.024)} & \color{gray}{1.01(0.007)} \\ 
& (300,5) & \textbf{1.34(0.076)} & 1.18(0.051) & \color{gray}{1.04(0.023)} & \color{gray}{1.05(0.025)} & \color{gray}{1(0.006)} \\ 
    (n, s)  & (300,15) & \textbf{1.06(0.039)} & 1.03(0.017) & \color{gray}{1(0.006)} & \color{gray}{1.01(0.004)} & \color{gray}{1(0.005)} \\ 
   & (1000,1) & \textbf{1.89(0.002)} & \textbf{1.89(0.01)} & \color{gray}{1.22(0.04)} & \color{gray}{1.16(0.038)} & \color{gray}{1(0.007)} \\ 
   & (1000,5) & \textbf{1.85(0.004)} & 1.75(0.014) & \color{gray}{1.09(0.032)} & \color{gray}{1.13(0.033)} & \color{gray}{1.01(0.005)} \\ 
   & (1000,15) & \textbf{1.74(0.015)} & 1.38(0.035) & \color{gray}{1.06(0.029)} & \color{gray}{1.1(0.029)} & \color{gray}{1(0.007)} \\   \hline
  B & (300,1) & \textbf{3.7(0.002)} & \textbf{3.7(0.002)} & \color{gray}{2.21(0.221)} & \color{gray}{2.26(0.204)} & \color{gray}{1.16(0.02)} \\ 
   & (300,5) & \textbf{3.67(0.003)} & 3.62(0.004) & \color{gray}{2.26(0.25)} & \color{gray}{2.32(0.22)} & \color{gray}{1.2(0.03)} \\ 
   (n, s)   & (300,15) & \textbf{3.6(0.008)} & \textbf{3.6(0.008)} & \color{gray}{1.64(0.189)} & \color{gray}{2.21(0.191)} & \color{gray}{1.26(0.034)} \\ 
   & (1000,1) & \textbf{3.7(0.003)} & 3.69(0.003) & \color{gray}{3.03(0.15)} & \color{gray}{2.85(0.202)} & \color{gray}{2.4(0.009)} \\ 
   & (1000,5) & \textbf{3.69(0.003)} & 3.65(0.005) & \color{gray}{3.2(0.089)} & \color{gray}{3.17(0.137)} & \color{gray}{2.38(0.012)} \\ 
   & (1000,15) & \textbf{3.68(0.003)} & 3.65(0.004) & \color{gray}{2.98(0.126)} & \color{gray}{3.17(0.087)} & \color{gray}{2.4(0.019)} \\  \hline 
  identity & direction2 &  &  &  &  &  \\   \hline
  A & (300,1) & \textbf{1.37(0.077)} & 1.21(0.072) & \color{gray}{1.02(0.009)} & \color{gray}{1.06(0.025)} & \color{gray}{1(0.006)} \\ 
   & (300,5) & \textbf{1.11(0.051)} & 1.04(0.032) & \color{gray}{1.01(0.007)} & \color{gray}{1.01(0.009)} & \color{gray}{1(0.006)} \\ 
  (n, s)    & (300,15) & \textbf{1.02(0.022)} & \color{gray}{1(0.006)} & 1.01(0.006) & \color{gray}{1(0.006)} & 1.01(0.005) \\ 
   & (1000,1) & 1.69(0.004) & \textbf{1.7(0.014)} & \color{gray}{1.01(0.012)} & \color{gray}{1.15(0.034)} & \color{gray}{1.02(0.005)} \\ 
   & (1000,5) & \textbf{1.53(0.054)} & 1.38(0.044) & \color{gray}{1.05(0.025)} & \color{gray}{1.09(0.03)} & \color{gray}{1.01(0.008)} \\ 
   & (1000,15) & \textbf{1.17(0.046)} & 1.07(0.028) & \color{gray}{1.03(0.018)} & \color{gray}{1.04(0.024)} & \color{gray}{1.01(0.006)} \\   \hline
  B & (300,1) & \textbf{3.1(0.005)} & \textbf{3.1(0.005)} & \color{gray}{1.54(0.174)} & \color{gray}{1.75(0.186)} & \color{gray}{1.12(0.02)} \\ 
   & (300,5) & \textbf{3.01(0.006)} & 2.84(0.012) & \color{gray}{1.32(0.123)} & \color{gray}{1.7(0.192)} & \color{gray}{1.09(0.015)} \\ 
  (n, s)    & (300,15) & \textbf{2.28(0.148)} & 1.94(0.163) & \color{gray}{1.44(0.145)} & \color{gray}{1.47(0.165)} & \color{gray}{1.09(0.015)} \\ 
   & (1000,1) & \textbf{3.1(0.004)} & 3.07(0.007) & \color{gray}{2(0.202)} & \color{gray}{2.62(0.13)} & \color{gray}{1.8(0.01)} \\ 
   & (1000,5) & \textbf{3.07(0.007)} & 2.96(0.013) & \color{gray}{2.48(0.173)} & \color{gray}{2.36(0.178)} & \color{gray}{1.78(0.019)} \\ 
   & (1000,15) & \textbf{3.03(0.006)} & \color{gray}{2.88(0.008)} & \color{gray}{2.28(0.23)} & 2.98(0.083) & \color{gray}{1.79(0.014)} \\ 
   \hline
\end{tabular}
\end{adjustbox}
\end{table}

\begin{table}
\centering
\caption{Achieved mCCA correlations after deflation with toplitz within-block covariance.}
\label{tab:sim2}
\begin{adjustbox}{width=.75\textwidth}
\begin{tabular}{| ll | lllll |}
  \hline
toplitz & direction1 & msCCA1 & rifle(seq) & pma & sgcca & rgcca \\   \hline
  A & (300,1) & \textbf{1.86(0.018)} & 1.83(0.046) & \color{gray}{1.07(0.029)} & \color{gray}{1.02(0.018)} & \color{gray}{1(0.007)} \\ 
   & (300,5) & \textbf{1.47(0.073)} & 1.23(0.06) & \color{gray}{1.04(0.025)} & \color{gray}{1.04(0.018)} & \color{gray}{1.01(0.007)} \\ 
  (n, s)    & (300,15) & \textbf{1.05(0.02)} & 1.01(0.01) & \color{gray}{1(0.009)} & 1.01(0.007) & 1.01(0.008) \\ 
   & (1000,1) & \textbf{1.89(0.009)} & 1.88(0.014) & \color{gray}{1.16(0.038)} & \color{gray}{1.17(0.04)} & \color{gray}{1.02(0.01)} \\ 
   & (1000,5) & \textbf{1.85(0.005)} & 1.71(0.026) & \color{gray}{1.16(0.046)} & \color{gray}{1.13(0.037)} & \color{gray}{1.02(0.007)} \\ 
   & (1000,15) & \textbf{1.71(0.039)} & 1.39(0.045) & \color{gray}{1.04(0.023)} & \color{gray}{1.05(0.024)} & \color{gray}{1.02(0.006)} \\   \hline
  B & (300,1) & \textbf{3.7(0.002)} & \textbf{3.7(0.002)} & \color{gray}{2.4(0.193)} & \color{gray}{2.23(0.173)} & \color{gray}{1.34(0.027)} \\ 
   & (300,5) & \textbf{3.68(0.002)} & 3.62(0.005) & \color{gray}{2.24(0.192)} & \color{gray}{2.15(0.178)} & \color{gray}{1.33(0.028)} \\ 
  (n, s)    & (300,15) & \textbf{3.6(0.005)} & \textbf{3.6(0.008)} & \color{gray}{2.2(0.17)} & \color{gray}{2.32(0.15)} & \color{gray}{1.4(0.032)} \\ 
   & (1000,1) & \textbf{3.7(0.003)} & 3.69(0.003) & \color{gray}{2.92(0.122)} & \color{gray}{2.75(0.178)} & \color{gray}{2.39(0.01)} \\ 
   & (1000,5) & \textbf{3.69(0.002)} & 3.66(0.003) & \color{gray}{3.2(0.063)} & \color{gray}{2.91(0.17)} & \color{gray}{2.38(0.011)} \\ 
   & (1000,15) & \textbf{3.68(0.003)} & 3.65(0.003) & \color{gray}{2.86(0.101)} & \color{gray}{2.93(0.099)} & \color{gray}{2.41(0.017)} \\   \hline
  toplitz & direction2 &  &  &  &  &  \\   \hline
  A & (300,1) & \textbf{1.48(0.065)} & 1.36(0.075) & \color{gray}{1.05(0.026)} & \color{gray}{1.08(0.03)} & \color{gray}{1.01(0.006)} \\ 
   & (300,5) & \textbf{1.13(0.046)} & 1.09(0.03) & \color{gray}{1.01(0.012)} & \color{gray}{1.04(0.019)} & \color{gray}{1.02(0.007)} \\ 
  (n, s)    & (300,15) & 1.02(0.019) & \textbf{1.03(0.011)} & \color{gray}{0.99(0.008)} & \color{gray}{1(0.006)} & \color{gray}{1(0.007)} \\ 
   & (1000,1) & \textbf{1.7(0.011)} & 1.67(0.041) & \color{gray}{1.09(0.036)} & \color{gray}{1.17(0.039)} & \color{gray}{1.02(0.006)} \\ 
   & (1000,5) & \textbf{1.55(0.045)} & 1.33(0.054) & \color{gray}{1.11(0.04)} & \color{gray}{1.11(0.039)} & \color{gray}{1(0.006)} \\ 
   & (1000,15) & \textbf{1.18(0.054)} & \color{gray}{1.05(0.021)} & \color{gray}{1.05(0.024)} & 1.11(0.032) & \color{gray}{1.01(0.007)} \\   \hline
  B & (300,1) & \textbf{3.09(0.006)} & \textbf{3.09(0.007)} & \color{gray}{1.88(0.218)} & \color{gray}{1.86(0.188)} & \color{gray}{1.1(0.015)} \\ 
   & (300,5) & \textbf{3.01(0.009)} & 2.84(0.014) & \color{gray}{1.52(0.156)} & \color{gray}{1.81(0.199)} & \color{gray}{1.1(0.014)} \\ 
  (n, s)    & (300,15) & \textbf{2.49(0.144)} & 2.14(0.171) & \color{gray}{1.26(0.112)} & \color{gray}{1.6(0.187)} & \color{gray}{1.11(0.016)} \\ 
   & (1000,1) & \textbf{3.1(0.005)} & \textbf{3.1(0.005)} & \color{gray}{2.5(0.17)} & \color{gray}{2.63(0.188)} & \color{gray}{1.84(0.014)} \\ 
   & (1000,5) & \textbf{3.08(0.005)} & 2.99(0.012) & \color{gray}{1.9(0.199)} & \color{gray}{2.57(0.169)} & \color{gray}{1.83(0.016)} \\ 
   & (1000,15) & \textbf{3.03(0.006)} & 2.93(0.007) & \color{gray}{2.28(0.19)} & \color{gray}{2.69(0.189)} & \color{gray}{1.83(0.022)} \\ 
   \hline
\end{tabular}
\end{adjustbox}
\end{table}

\begin{table}
\centering
\caption{Achieved mCCA correlations after deflation with spiked within-block covariance.}
\label{tab:sim3}
\begin{adjustbox}{width=.75\textwidth}
\begin{tabular}{| ll | lllll |}
  \hline
spiked & direction1 & msCCA1 & rifle(seq) & pma & sgcca & rgcca \\   \hline
  A & (300,1) & \textbf{1.87(0.015)} & 1.73(0.053) & \color{gray}{1.3(0.025)} & \color{gray}{1.34(0.032)} & \color{gray}{1.17(0.019)} \\ 
   & (300,5) & \textbf{1.58(0.051)} & 1.33(0.042) & \color{gray}{1.21(0.022)} & \color{gray}{1.2(0.031)} & \color{gray}{1.16(0.02)} \\ 
  (n, s)    & (300,15) & \textbf{1.29(0.045)} & 1.2(0.032) & \color{gray}{1.16(0.02)} & \color{gray}{1.15(0.021)} & 1.2(0.018) \\ 
   & (1000,1) & \textbf{1.89(0.003)} & 1.87(0.011) & \color{gray}{1.37(0.021)} & \color{gray}{1.4(0.02)} & \color{gray}{1.25(0.02)} \\ 
   & (1000,5) & \textbf{1.85(0.003)} & 1.71(0.017) & \color{gray}{1.29(0.024)} & \color{gray}{1.29(0.032)} & \color{gray}{1.25(0.026)} \\ 
   & (1000,15) & \textbf{1.69(0.024)} & 1.42(0.023) & \color{gray}{1.21(0.022)} & \color{gray}{1.21(0.023)} & \color{gray}{1.24(0.02)} \\  \hline 
  B & (300,1) & \textbf{3.7(0.003)} & \textbf{3.7(0.003)} & \color{gray}{2.77(0.07)} & \color{gray}{3.31(0.109)} & \color{gray}{1.9(0.041)} \\ 
   & (300,5) & \textbf{3.67(0.002)} & 3.63(0.004) & \color{gray}{2.64(0.062)} & \color{gray}{2.87(0.054)} & \color{gray}{1.87(0.057)} \\ 
   (n, s)   & (300,15) & \textbf{3.52(0.046)} & 3.5(0.056) & \color{gray}{2.24(0.065)} & \color{gray}{2.31(0.058)} & \color{gray}{1.97(0.05)} \\ 
   & (1000,1) & \textbf{3.69(0.002)} & \textbf{3.69(0.002)} & \color{gray}{3.07(0.027)} & \color{gray}{3.61(0.041)} & \color{gray}{2.12(0.048)} \\ 
   & (1000,5) & \textbf{3.69(0.003)} & 3.66(0.003) & \color{gray}{2.78(0.069)} & \color{gray}{2.93(0.074)} & \color{gray}{2.04(0.05)} \\ 
   & (1000,15) & 3.63(0.032) & \textbf{3.65(0.003)} & \color{gray}{2.36(0.058)} & \color{gray}{2.42(0.045)} & \color{gray}{2.02(0.038)} \\   \hline
  spiked & direction2 &  &  &  &  &  \\   \hline
  A & (300,1) & \textbf{1.61(0.043)} & 1.33(0.065) & \color{gray}{1.1(0.033)} & \color{gray}{1.11(0.033)} & \color{gray}{1.07(0.011)} \\ 
   & (300,5) & \textbf{1.22(0.057)} & 1.1(0.035) & \color{gray}{1.08(0.02)} & 1.1(0.031) & 1.1(0.022) \\ 
  (n, s)    & (300,15) & 1.08(0.029) & \color{gray}{1.03(0.009)} & \textbf{1.1(0.014)} & 1.08(0.017) & \color{gray}{1.06(0.015)} \\ 
   & (1000,1) & \textbf{1.67(0.012)} & 1.52(0.051) & \color{gray}{1.13(0.038)} & \color{gray}{1.3(0.03)} & \color{gray}{1.1(0.015)} \\ 
   & (1000,5) & \textbf{1.53(0.036)} & 1.27(0.044) & \color{gray}{1.1(0.021)} & \color{gray}{1.11(0.024)} & \color{gray}{1.1(0.02)} \\ 
   & (1000,15) & \textbf{1.38(0.057)} & 1.17(0.025) & \color{gray}{1.08(0.012)} & \color{gray}{1.08(0.013)} & \color{gray}{1.08(0.013)} \\   \hline
  B & (300,1) & 3.08(0.007) & \color{gray}{2.99(0.014)} & \color{gray}{2.38(0.139)} & \textbf{3.09(0.131)} & \color{gray}{1.41(0.03)} \\ 
   & (300,5) & \textbf{2.9(0.05)} & 2.76(0.057) & \color{gray}{1.81(0.119)} & \color{gray}{2.21(0.11)} & \color{gray}{1.5(0.033)} \\ 
   (n, s)   & (300,15) & \textbf{2.73(0.044)} & 2.54(0.069) & \color{gray}{1.49(0.074)} & \color{gray}{1.78(0.073)} & \color{gray}{1.52(0.04)} \\ 
   & (1000,1) & \textbf{3.09(0.005)} & 3.04(0.007) & \color{gray}{2.59(0.071)} & \color{gray}{2.86(0.083)} & \color{gray}{1.59(0.024)} \\ 
   & (1000,5) & \textbf{3.06(0.006)} & 2.99(0.008) & \color{gray}{2.01(0.162)} & \color{gray}{2.48(0.087)} & \color{gray}{1.6(0.023)} \\ 
   & (1000,15) & \textbf{3.01(0.03)} & 2.92(0.008) & \color{gray}{1.44(0.07)} & \color{gray}{1.87(0.038)} & \color{gray}{1.54(0.034)} \\ 
   \hline
\end{tabular}
\end{adjustbox}
\end{table}

\begin{table}
\centering
\caption{Achieved aggregated projection residuals with identity within-block covariance.}
\label{tab:sim4}
\begin{adjustbox}{width=.95\textwidth}
\begin{tabular}{| ll | lllll |}
  \hline
identity & projection1 & msCCA1 & rifle(seq) & pma & sgcca & rgcca \\   \hline
  A & (300,1) & \textbf{1.6e-02(5.0e-03)} & 1.2e-01(4.8e-02) & \color{gray}{9.6e-01(1.9e-02)} & \color{gray}{9.3e-01(3.6e-02)} & \color{gray}{9.9e-01(2.1e-03)} \\ 
   & (300,5) & \textbf{5.2e-01(9.8e-02)} & 6.7e-01(7.2e-02) & \color{gray}{9.3e-01(3.7e-02)} & \color{gray}{9.2e-01(4.0e-02)} & \color{gray}{9.9e-01(2.0e-03)} \\ 
   & (300,15) & \textbf{8.6e-01(5.9e-02)} & 9.4e-01(3.0e-02) & \color{gray}{9.9e-01(4.8e-03)} & \color{gray}{9.9e-01(2.8e-03)} & \color{gray}{9.9e-01(2.3e-03)} \\ 
   & (1000,1) & 3.5e-03(7.1e-04) & \textbf{9.2e-06(1.8e-06)} & \color{gray}{6.8e-01(6.6e-02)} & \color{gray}{6.4e-01(5.5e-02)} & \color{gray}{9.8e-01(3.7e-03)} \\ 
   & (1000,5) & \textbf{3.3e-02(2.5e-03)} & 1.0e-01(9.6e-03) & \color{gray}{8.7e-01(5.1e-02)} & \color{gray}{7.9e-01(5.7e-02)} & \color{gray}{9.7e-01(4.2e-03)} \\ 
   & (1000,15) & \textbf{9.9e-02(1.2e-02)} & 4.1e-01(3.5e-02) & \color{gray}{8.7e-01(4.9e-02)} & \color{gray}{8.0e-01(5.2e-02)} & \color{gray}{9.8e-01(5.0e-03)} \\   \hline
  B & (300,1) & 7.1e-05(3.7e-05) & \textbf{9.5e-06(5.4e-07)} & \color{gray}{5.2e-01(9.8e-02)} & \color{gray}{4.2e-01(9.8e-02)} & \color{gray}{7.6e-01(2.0e-02)} \\ 
   & (300,5) & \textbf{2.7e-03(1.5e-04)} & 7.9e-03(3.3e-04) & \color{gray}{4.4e-01(9.6e-02)} & \color{gray}{4.1e-01(1.0e-01)} & \color{gray}{7.6e-01(2.3e-02)} \\ 
   & (300,15) & \textbf{1.0e-02(6.7e-04)} & 1.1e-02(7.8e-04) & \color{gray}{5.4e-01(9.0e-02)} & \color{gray}{2.6e-01(7.2e-02)} & \color{gray}{7.0e-01(2.5e-02)} \\ 
   & (1000,1) & \textbf{6.0e-06(4.2e-06)} & 4.0e-04(4.8e-05) & \color{gray}{2.4e-01(8.8e-02)} & \color{gray}{1.4e-01(6.6e-02)} & \color{gray}{2.0e-01(2.2e-03)} \\ 
   & (1000,5) & \textbf{6.6e-04(4.0e-05)} & 4.5e-03(4.5e-04) & \color{gray}{2.3e-01(8.9e-02)} & \color{gray}{8.0e-02(4.8e-02)} & \color{gray}{2.0e-01(3.5e-03)} \\ 
   & (1000,15) & \textbf{1.9e-03(5.4e-05)} & 5.0e-03(2.0e-04) & \color{gray}{1.5e-01(6.6e-02)} & \color{gray}{7.8e-02(4.8e-02)} & \color{gray}{2.0e-01(4.0e-03)} \\   \hline
  identity & projection2 &  &  &  &  &  \\   \hline
  A & (300,1) & \textbf{4.6e-01(1.1e-01)} & 6.1e-01(1.1e-01) & \color{gray}{9.4e-01(3.7e-02)} & \color{gray}{9.2e-01(4.1e-02)} & \color{gray}{9.9e-01(1.2e-03)} \\ 
   & (300,5) & \textbf{9.0e-01(5.3e-02)} & 9.8e-01(1.3e-02) & \color{gray}{9.9e-01(7.7e-03)} & \color{gray}{1.0e+00(7.5e-04)} & \color{gray}{1.0e+00(1.0e-03)} \\ 
   & (300,15) & \color{gray}{1.0e+00(5.9e-04)} & \color{gray}{1.0e+00(9.4e-04)} & \color{gray}{1.0e+00(9.9e-04)} & \textbf{9.9e-01(3.2e-03)} & \textbf{9.9e-01(1.5e-03)} \\ 
   & (1000,1) & \textbf{6.9e-03(1.8e-03)} & 5.0e-02(5.0e-02) & \color{gray}{9.8e-01(2.0e-02)} & \color{gray}{8.8e-01(4.4e-02)} & \color{gray}{9.9e-01(2.2e-03)} \\ 
   & (1000,5) & \textbf{2.5e-01(8.6e-02)} & 3.9e-01(7.4e-02) & \color{gray}{9.2e-01(3.7e-02)} & \color{gray}{8.8e-01(4.4e-02)} & \color{gray}{9.8e-01(3.3e-03)} \\ 
   & (1000,15) & \textbf{7.4e-01(7.6e-02)} & 8.8e-01(4.3e-02) & \color{gray}{9.9e-01(2.4e-03)} & \color{gray}{9.9e-01(3.3e-03)} & \color{gray}{9.9e-01(3.2e-03)} \\   \hline
  B & (300,1) & 7.8e-04(3.0e-04) & \textbf{3.0e-05(3.6e-06)} & \color{gray}{7.4e-01(8.7e-02)} & \color{gray}{6.9e-01(9.4e-02)} & \color{gray}{9.0e-01(1.8e-02)} \\ 
   & (300,5) & \textbf{1.3e-02(7.1e-04)} & 4.2e-02(1.9e-03) & \color{gray}{8.3e-01(7.1e-02)} & \color{gray}{6.7e-01(9.6e-02)} & \color{gray}{9.4e-01(1.3e-02)} \\ 
   & (300,15) & \textbf{3.5e-01(9.5e-02)} & 4.5e-01(8.8e-02) & \color{gray}{8.3e-01(6.5e-02)} & \color{gray}{8.7e-01(6.5e-02)} & \color{gray}{9.5e-01(9.2e-03)} \\ 
   & (1000,1) & \textbf{6.4e-05(2.2e-05)} & 3.7e-03(1.1e-03) & \color{gray}{5.8e-01(1.1e-01)} & \color{gray}{4.9e-01(1.0e-01)} & \color{gray}{3.5e-01(4.0e-03)} \\ 
   & (1000,5) & \textbf{3.6e-03(1.8e-04)} & 2.2e-02(2.0e-03) & \color{gray}{3.1e-01(9.5e-02)} & \color{gray}{5.7e-01(1.1e-01)} & \color{gray}{3.8e-01(9.9e-03)} \\ 
   & (1000,15) & \textbf{1.0e-02(3.5e-04)} & 3.5e-02(1.3e-03) & \color{gray}{4.3e-01(1.1e-01)} & \color{gray}{1.4e-01(6.6e-02)} & \color{gray}{3.6e-01(8.6e-03)} \\ 
   \hline
\end{tabular}
\end{adjustbox}
\end{table}

\begin{table}
\centering
\caption{Achieved aggregated projection residuals with  toplitz within-block covariance.}
\label{tab:sim5}
\begin{adjustbox}{width=.95\textwidth}
\begin{tabular}{| ll | lllll |}
  \hline
toplitz & projection1 & msCCA1 & rifle(seq) & pma & sgcca & rgcca \\   \hline
  A & (300,1) & \textbf{2.3e-02(1.0e-02)} & 3.2e-02(2.3e-02) & \color{gray}{8.8e-01(5.3e-02)} & \color{gray}{9.0e-01(4.5e-02)} & \color{gray}{9.8e-01(3.7e-03)} \\ 
   & (300,5) & \textbf{4.5e-01(9.3e-02)} & 6.4e-01(6.9e-02) & \color{gray}{9.5e-01(3.3e-02)} & \color{gray}{9.4e-01(3.1e-02)} & \color{gray}{9.9e-01(2.2e-03)} \\ 
   & (300,15) & \textbf{9.0e-01(3.6e-02)} & 9.2e-01(1.9e-02) & \color{gray}{9.9e-01(2.5e-03)} & \color{gray}{9.9e-01(7.9e-03)} & \color{gray}{9.9e-01(3.1e-03)} \\ 
   & (1000,1) & 3.0e-03(7.0e-04) & \textbf{1.9e-05(1.1e-06)} & \color{gray}{7.1e-01(6.0e-02)} & \color{gray}{6.8e-01(6.4e-02)} & \color{gray}{9.5e-01(8.0e-03)} \\ 
   & (1000,5) & \textbf{3.0e-02(2.4e-03)} & 1.5e-01(4.6e-02) & \color{gray}{6.2e-01(6.4e-02)} & \color{gray}{6.9e-01(6.4e-02)} & \color{gray}{9.6e-01(6.4e-03)} \\ 
   & (1000,15) & \textbf{1.4e-01(4.6e-02)} & 4.4e-01(5.4e-02) & \color{gray}{8.6e-01(4.6e-02)} & \color{gray}{7.4e-01(5.1e-02)} & \color{gray}{9.6e-01(6.4e-03)} \\   \hline
  B & (300,1) & 7.1e-05(3.1e-05) & \textbf{9.9e-06(5.8e-07)} & \color{gray}{4.2e-01(1.0e-01)} & \color{gray}{4.7e-01(9.9e-02)} & \color{gray}{6.9e-01(2.1e-02)} \\ 
   & (300,5) & \textbf{2.4e-03(1.6e-04)} & 7.9e-03(3.6e-04) & \color{gray}{4.0e-01(8.7e-02)} & \color{gray}{3.4e-01(8.8e-02)} & \color{gray}{6.8e-01(2.6e-02)} \\ 
   & (300,15) & \textbf{9.7e-03(3.2e-04)} & 1.0e-02(8.6e-04) & \color{gray}{5.3e-01(9.2e-02)} & \color{gray}{3.0e-01(8.1e-02)} & \color{gray}{6.1e-01(2.2e-02)} \\ 
   & (1000,1) & \textbf{5.2e-05(2.3e-05)} & 2.9e-04(2.0e-05) & \color{gray}{2.8e-01(9.5e-02)} & \color{gray}{9.8e-02(4.8e-02)} & \color{gray}{2.1e-01(1.8e-03)} \\ 
   & (1000,5) & \textbf{7.8e-04(4.5e-05)} & 4.1e-03(2.5e-04) & \color{gray}{2.4e-01(8.8e-02)} & \color{gray}{1.8e-01(7.9e-02)} & \color{gray}{2.1e-01(3.0e-03)} \\ 
   & (1000,15) & \textbf{2.0e-03(8.6e-05)} & 4.9e-03(1.3e-04) & \color{gray}{2.5e-01(8.6e-02)} & \color{gray}{1.3e-01(6.6e-02)} & \color{gray}{2.0e-01(3.2e-03)} \\   \hline
  toplitz & projection2 &  &  &  &  &  \\   \hline
  A & (300,1) & \textbf{3.3e-01(1.0e-01)} & 4.4e-01(1.0e-01) & \color{gray}{9.7e-01(2.3e-02)} & \color{gray}{9.5e-01(3.2e-02)} & \color{gray}{9.9e-01(2.0e-03)} \\ 
   & (300,5) & \textbf{7.3e-01(8.5e-02)} & 8.5e-01(6.5e-02) & \color{gray}{9.5e-01(3.3e-02)} & \color{gray}{9.7e-01(2.6e-02)} & \color{gray}{9.9e-01(2.0e-03)} \\ 
   (n, s)   & (300,15) & \textbf{9.5e-01(2.8e-02)} & 9.7e-01(1.5e-02) & \color{gray}{9.9e-01(2.0e-03)} & \color{gray}{1.0e+00(1.3e-03)} & \color{gray}{9.9e-01(1.6e-03)} \\ 
   & (1000,1) & \textbf{8.4e-03(2.1e-03)} & 1.0e-01(6.9e-02) & \color{gray}{8.7e-01(4.7e-02)} & \color{gray}{8.0e-01(5.8e-02)} & \color{gray}{9.8e-01(4.4e-03)} \\ 
   & (1000,5) & \textbf{2.1e-01(7.6e-02)} & 4.5e-01(8.3e-02) & \color{gray}{9.5e-01(2.7e-02)} & \color{gray}{9.3e-01(3.7e-02)} & \color{gray}{9.9e-01(1.6e-03)} \\ 
   & (1000,15) & \textbf{6.9e-01(8.7e-02)} & 8.7e-01(3.7e-02) & \color{gray}{9.8e-01(8.0e-03)} & \color{gray}{9.7e-01(1.2e-02)} & \color{gray}{9.7e-01(4.3e-03)} \\   \hline
  B & (300,1) & 5.5e-04(2.3e-04) & \textbf{8.7e-05(3.6e-05)} & \color{gray}{5.5e-01(1.0e-01)} & \color{gray}{5.6e-01(1.0e-01)} & \color{gray}{9.0e-01(1.6e-02)} \\ 
   & (300,5) & \textbf{1.3e-02(8.9e-04)} & 4.5e-02(2.7e-03) & \color{gray}{7.0e-01(8.5e-02)} & \color{gray}{7.5e-01(8.8e-02)} & \color{gray}{8.9e-01(1.4e-02)} \\ 
    (n, s)  & (300,15) & \textbf{2.0e-01(7.7e-02)} & 3.4e-01(8.8e-02) & \color{gray}{6.7e-01(8.8e-02)} & \color{gray}{7.1e-01(8.8e-02)} & \color{gray}{8.9e-01(1.8e-02)} \\ 
   & (1000,1) & \textbf{1.4e-04(7.7e-05)} & 7.3e-04(8.5e-05) & \color{gray}{3.0e-01(9.3e-02)} & \color{gray}{4.8e-01(1.1e-01)} & \color{gray}{3.4e-01(7.5e-03)} \\ 
   & (1000,5) & \textbf{3.2e-03(1.8e-04)} & 1.8e-02(2.4e-03) & \color{gray}{4.8e-01(1.1e-01)} & \color{gray}{3.8e-01(1.0e-01)} & \color{gray}{3.4e-01(6.7e-03)} \\ 
   & (1000,15) & \textbf{9.8e-03(3.1e-04)} & 2.8e-02(9.3e-04) & \color{gray}{3.2e-01(9.0e-02)} & \color{gray}{3.4e-01(9.9e-02)} & \color{gray}{3.5e-01(1.1e-02)} \\
   \hline
\end{tabular}
\end{adjustbox}
\end{table}
\begin{table}
\centering
\caption{Achieved aggregated projection residuals with spiked within-block covariance.}
\label{tab:sim6}
\begin{adjustbox}{width=.95\textwidth}
\begin{tabular}{| ll | lllll |}
  \hline
spiked & projection1 & msCCA1 & rifle(seq) & pma & sgcca & rgcca \\   \hline
  A & (300,1) & \textbf{5.5e-02(5.0e-02)} & 8.7e-02(4.9e-02) & \color{gray}{5.1e-01(4.2e-02)} & \color{gray}{5.1e-01(5.7e-02)} & \color{gray}{7.0e-01(2.7e-02)} \\ 
   & (300,5) & \textbf{3.4e-01(8.1e-02)} & 5.7e-01(6.5e-02) & \color{gray}{7.0e-01(4.6e-02)} & \color{gray}{6.9e-01(5.0e-02)} & \color{gray}{7.7e-01(3.5e-02)} \\ 
  (n, s)    & (300,15) & \textbf{5.6e-01(5.7e-02)} & 6.8e-01(4.7e-02) & \color{gray}{7.1e-01(3.6e-02)} & \color{gray}{7.4e-01(2.8e-02)} & \color{gray}{7.3e-01(3.0e-02)} \\ 
   & (1000,1) & \textbf{4.2e-03(1.4e-03)} & 4.4e-02(2.5e-02) & \color{gray}{4.9e-01(5.5e-02)} & \color{gray}{4.5e-01(5.4e-02)} & \color{gray}{6.7e-01(4.0e-02)} \\ 
   & (1000,5) & \textbf{4.4e-02(2.8e-03)} & 1.6e-01(2.9e-02) & \color{gray}{5.5e-01(3.9e-02)} & \color{gray}{5.6e-01(4.3e-02)} & \color{gray}{6.2e-01(3.0e-02)} \\ 
   & (1000,15) & \textbf{1.7e-01(4.4e-02)} & 4.3e-01(4.7e-02) & \color{gray}{7.0e-01(3.9e-02)} & \color{gray}{6.9e-01(3.9e-02)} & \color{gray}{6.7e-01(3.6e-02)} \\   \hline
  B & (300,1) & 4.4e-04(7.8e-05) & \textbf{2.1e-04(4.4e-05)} & \color{gray}{1.7e-01(6.0e-02)} & \color{gray}{6.9e-03(2.1e-03)} & \color{gray}{4.4e-01(4.0e-02)} \\ 
   & (300,5) & \textbf{6.0e-03(5.7e-04)} & 9.7e-03(6.1e-04) & \color{gray}{1.9e-01(3.7e-02)} & \color{gray}{1.1e-01(8.2e-03)} & \color{gray}{4.1e-01(3.0e-02)} \\ 
   (n, s)   & (300,15) & \textbf{2.0e-02(1.8e-03)} & \textbf{2.0e-02(1.7e-03)} & \color{gray}{2.8e-01(3.8e-02)} & \color{gray}{2.7e-01(3.6e-02)} & \color{gray}{3.8e-01(2.9e-02)} \\ 
   & (1000,1) & \textbf{3.6e-04(5.2e-05)} & 5.9e-04(6.2e-05) & \color{gray}{6.5e-02(3.0e-03)} & \color{gray}{6.2e-03(2.7e-03)} & \color{gray}{3.0e-01(1.8e-02)} \\ 
   & (1000,5) & \textbf{3.7e-03(5.1e-04)} & 6.4e-03(5.4e-04) & \color{gray}{1.4e-01(4.4e-02)} & \color{gray}{1.2e-01(4.5e-02)} & \color{gray}{3.3e-01(2.9e-02)} \\ 
   & (1000,15) & \textbf{8.1e-03(5.1e-04)} & 9.2e-03(3.8e-04) & \color{gray}{2.8e-01(3.9e-02)} & \color{gray}{3.0e-01(5.5e-02)} & \color{gray}{4.2e-01(3.7e-02)} \\   \hline
  spiked & projection2 &  &  &  &  &  \\   \hline
  A & (300,1) & \textbf{1.6e-01(7.9e-02)} & 5.0e-01(7.6e-02) & \color{gray}{8.2e-01(4.5e-02)} & \color{gray}{7.7e-01(5.5e-02)} & \color{gray}{8.4e-01(1.9e-02)} \\ 
   & (300,5) & \textbf{6.2e-01(8.2e-02)} & 7.5e-01(5.9e-02) & \color{gray}{8.1e-01(4.0e-02)} & \color{gray}{7.9e-01(4.7e-02)} & \color{gray}{7.9e-01(3.4e-02)} \\ 
  (n, s)    & (300,15) & \color{gray}{8.7e-01(3.5e-02)} & \color{gray}{9.1e-01(2.1e-02)} & 8.6e-01(2.5e-02) & \color{gray}{8.7e-01(2.1e-02)} & \textbf{8.2e-01(2.7e-02)} \\ 
   & (1000,1) & \textbf{6.0e-02(4.9e-02)} & 1.6e-01(6.7e-02) & \color{gray}{7.5e-01(5.7e-02)} & \color{gray}{5.8e-01(5.9e-02)} & \color{gray}{7.6e-01(2.3e-02)} \\ 
   & (1000,5) & \textbf{2.9e-01(8.5e-02)} & 5.4e-01(8.1e-02) & \color{gray}{8.8e-01(3.4e-02)} & \color{gray}{8.4e-01(4.3e-02)} & \color{gray}{8.6e-01(3.1e-02)} \\ 
   & (1000,15) & \textbf{4.1e-01(7.3e-02)} & 6.6e-01(4.3e-02) & \color{gray}{8.2e-01(3.2e-02)} & \color{gray}{8.0e-01(3.2e-02)} & \color{gray}{7.7e-01(3.1e-02)} \\   \hline
  B & (300,1) & \textbf{2.4e-03(4.4e-04)} & 1.7e-02(5.1e-03) & \color{gray}{3.3e-01(8.4e-02)} & \color{gray}{2.6e-01(9.8e-02)} & \color{gray}{5.9e-01(4.6e-02)} \\ 
   & (300,5) & \textbf{1.2e-01(6.7e-02)} & 1.6e-01(6.5e-02) & \color{gray}{5.3e-01(8.2e-02)} & \color{gray}{4.8e-01(9.4e-02)} & \color{gray}{6.7e-01(4.9e-02)} \\ 
   (n, s)   & (300,15) & \textbf{7.3e-02(3.9e-03)} & 1.4e-01(2.4e-02) & \color{gray}{6.2e-01(5.7e-02)} & \color{gray}{4.4e-01(5.8e-02)} & \color{gray}{5.6e-01(3.5e-02)} \\ 
   & (1000,1) & \textbf{1.4e-03(2.0e-04)} & 9.5e-03(1.2e-03) & \color{gray}{1.1e-01(3.1e-02)} & \color{gray}{3.5e-01(1.1e-01)} & \color{gray}{5.0e-01(3.8e-02)} \\ 
   & (1000,5) & \textbf{8.9e-03(5.4e-04)} & 2.2e-02(1.6e-03) & \color{gray}{4.9e-01(8.2e-02)} & \color{gray}{3.2e-01(8.4e-02)} & \color{gray}{5.6e-01(4.2e-02)} \\ 
   & (1000,15) & \textbf{2.2e-02(1.2e-03)} & 3.6e-02(1.1e-03) & \color{gray}{6.0e-01(5.6e-02)} & \color{gray}{4.3e-01(5.8e-02)} & \color{gray}{5.8e-01(4.0e-02)} \\ \hline
   \end{tabular}
\end{adjustbox}
\end{table}

\section{Application to TCGA breast cancer data sets}
\label{sec:tcga}
When it comes to real data applications, we can no longer compare our estimation with the underlying signals regarding the aggregated projection $Z_{k}$. However, we can still compare the achieved deflated mCCA correlations using independent test samples. 

We compare the estimation quality of mCCA methods using a processed TCGA breast cancer data set (http://mixomics.org/mixdiablo/case-study-tcga/).  In this dataset, methylation (2000 features), mirna (184 features), and  mrna (2000 features) data are available for 989 patients. We add to the TCGA data set two non-informative pseudo blocks methylation.pseudo  and mirna.pseudo created by  permuting samples in the original methylation and mirna data sets independently. The two pseudo blocks have strong within-block correlation patterns but are not correlated with the three true blocks.

We split the processed TCGA data into roughly two equal-sized sets, and considered one the training  and the other the test. To test out the deflation procedure, we estimate 10 components for each method. In the left panel of figure \ref{fig:tcga1}, we show the achieved deflated mCCA correlations for five methods on the test set, averaged over the two random splits.   In this experiment, msCCA1 achieves the best performance in terms of the achieved deflated mCCA correlations, followed by rifle-seq, rgcca, sgcca and pma.  The estimated directions from pma become ``trivial'' and provide little additional information when the rank is greater than 5. As we can see later, the estimated directions from pma tend to correlate with each other, which makes additional gain hard as the number of component becomes large.

The 989 breast cancer participants are diagnosed with different breast cancer subtypes (Basal:  178; Her2: 78; LumA: 534; LumB: 199).  Although there is no guarantee that mCCA directions are relevant for separating different cancer subtypes, people in general consider mCCA directions to be potentially useful and interpretable for characterizing biological differences across samples -- such directions capture data variability supported by data from different sources and are likely to contain true biological signals. In our numerical experiments, the attempts to extract coherent information across blocks can help us alleviate the influence of the two pseudo blocks.   In the middle and right panels of Figure \ref{fig:tcga1}, we show the distribution of deviance loss and misclassification rates on the test data after performing roughly equal-sized random train-test splitting 50 times. We apply multinomial regression using the R package {\it glmnet} \citep{friedman2010regularization} for predicting cancer subtypes,  with features being ten derived mCCA directions from the five methods in comparison. The mCCA directions from msCCA1 achieve highest correlations on the test data, and are more informative for predicting cancer subtypes. If we further look into the derived mCCA direction $\beta$, msCCA1 and rifle are the only two methods that are robust to non-informative blocks in our experiment. In Figure \ref{fig:tcga2}, we show the coefficients $\beta$ for the leading direction and the ninth direction.  The three original blocks methylation,  mirna  and  mrna are colored black, red and blue and the two pseudo blocks are colored gray (referred to as noise in Figure \ref{fig:tcga2}).
\begin{figure}[H]
\caption{Left panel: Plot of mCCA correlations on the test data after deflation. Middle and right panel: missclassification rates  and deviance losses on test data using different methods. }
\label{fig:tcga1}
\includegraphics[width = 1\textwidth, height = .4\textwidth]{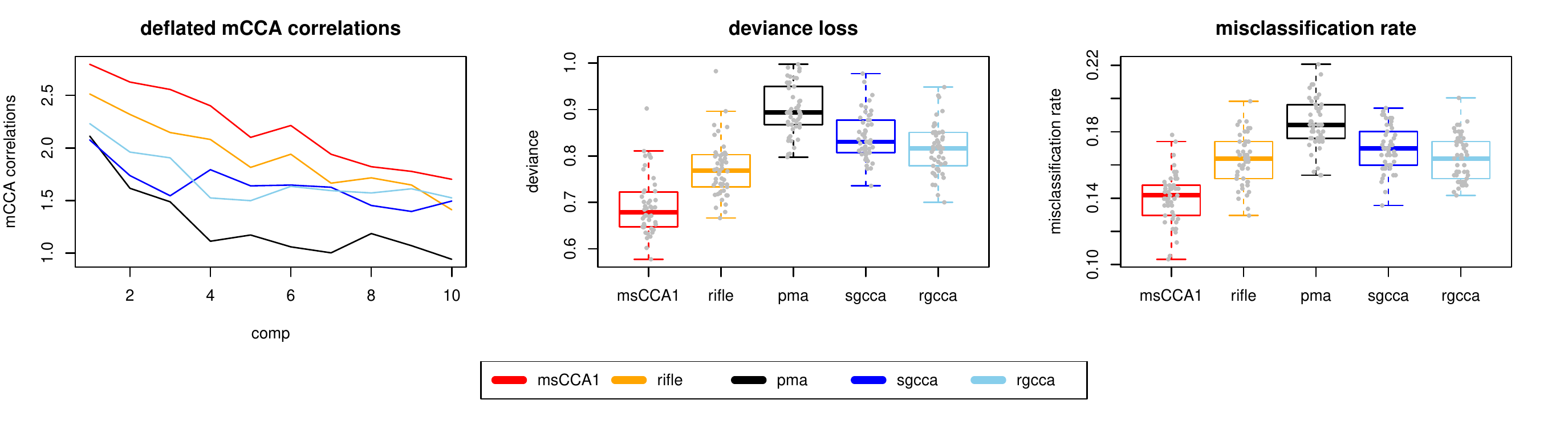}
\end{figure}

\begin{figure}
\caption{mCCA direction coefficient plots for first direction and ninth direction. The three original blocks methylation,  mirna  and  mrna are colored black, red and blue and the two pseudo blocks are colored gray and referred to as noise.}
\label{fig:tcga2}
\includegraphics[width = 1\textwidth, height = .6\textwidth]{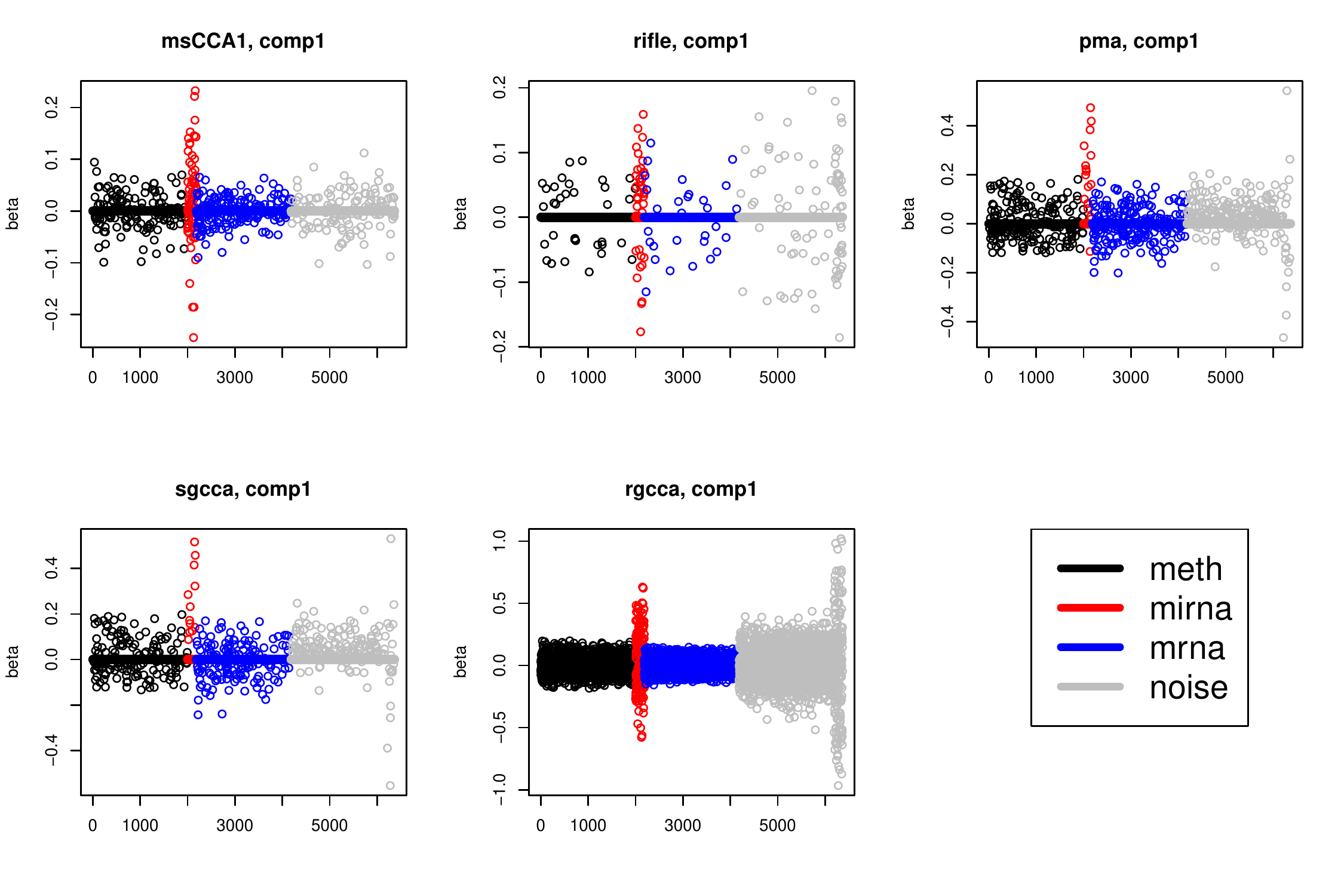}\\
\includegraphics[width = 1\textwidth, height = .6\textwidth]{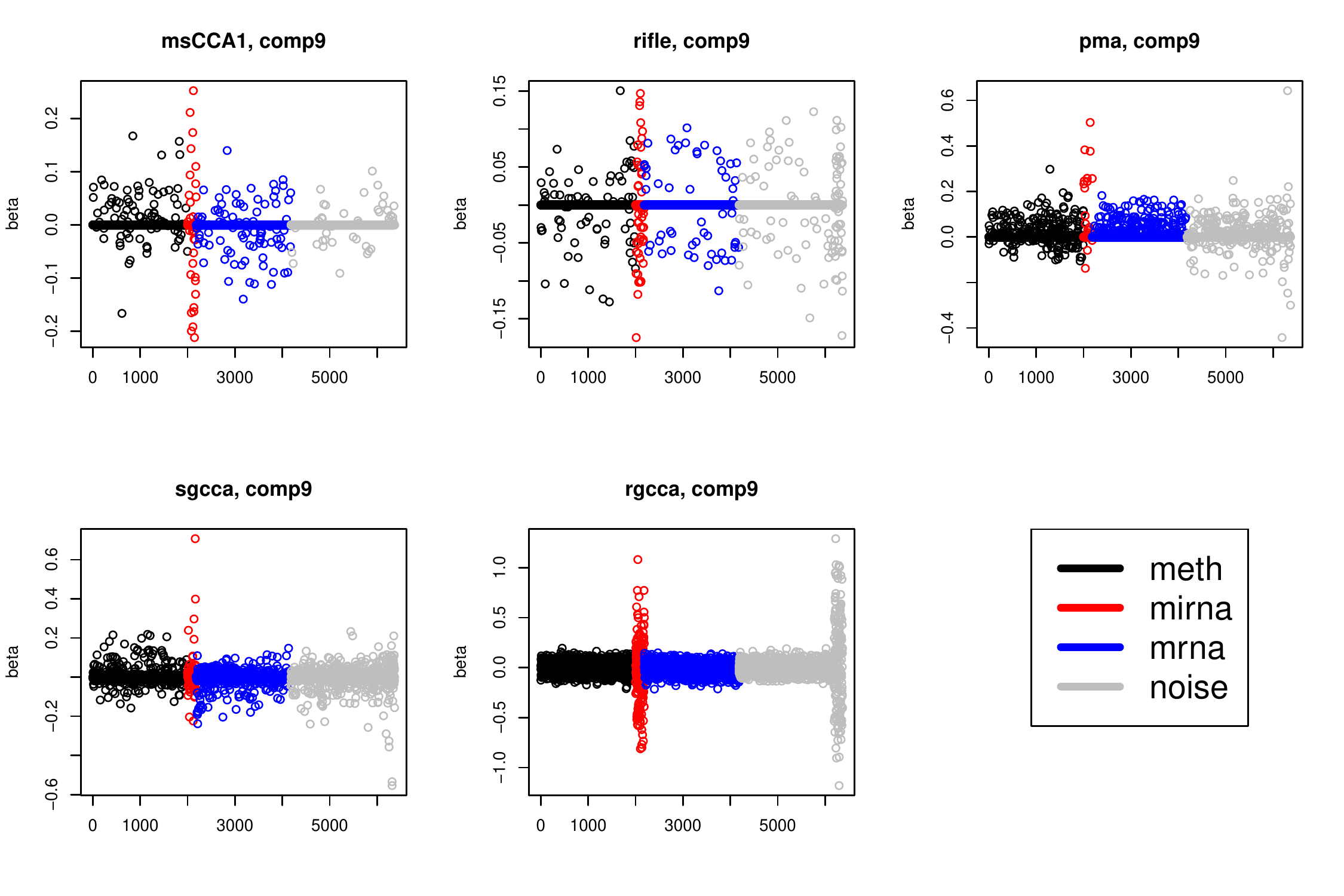}
\end{figure}

One reason that pma fails to provide additional information with higher rank components could be that it fails to deflate different components properly. Figure \ref{fig:tcga3} shows the correlation plots using the aggregated msCCA projections $Z$ from different methods. The projections are arranged based on methods. The diagonal blocks show the correlations between projections using a given method and the off-diagonal blocks show the correlations between projections estimated using different methods. The estimated $Z$s tend to be highly correlated using pma for different ranks, the other four methods have done a reasonable job at orthogonalizing the aggregated projections. The proposed deflation method works well for both msCCA1 and rifle, whose overall correlation patterns are slightly weaker and thus better decorrelated compared to rgcca and sgcca. The top mCCA directions are similar for msCCA1 and rifle, which also show some similarity compared to those estimated using rgcca and sgcca.

\begin{figure}
\begin{center}
\caption{Correlation plots for aggregated projections $Z$ for different ranks and different methods. Columns/rows are arranged by grouping estimates from a given method together and then ordered based on its rank order. The diagonal block shows the correlations between mCCA directions using a given method and the off-diagonal blocks show the correlations between mCCA directions estimated using different methods. }
\label{fig:tcga3}
\includegraphics[width = 1\textwidth, height =.9\textwidth]{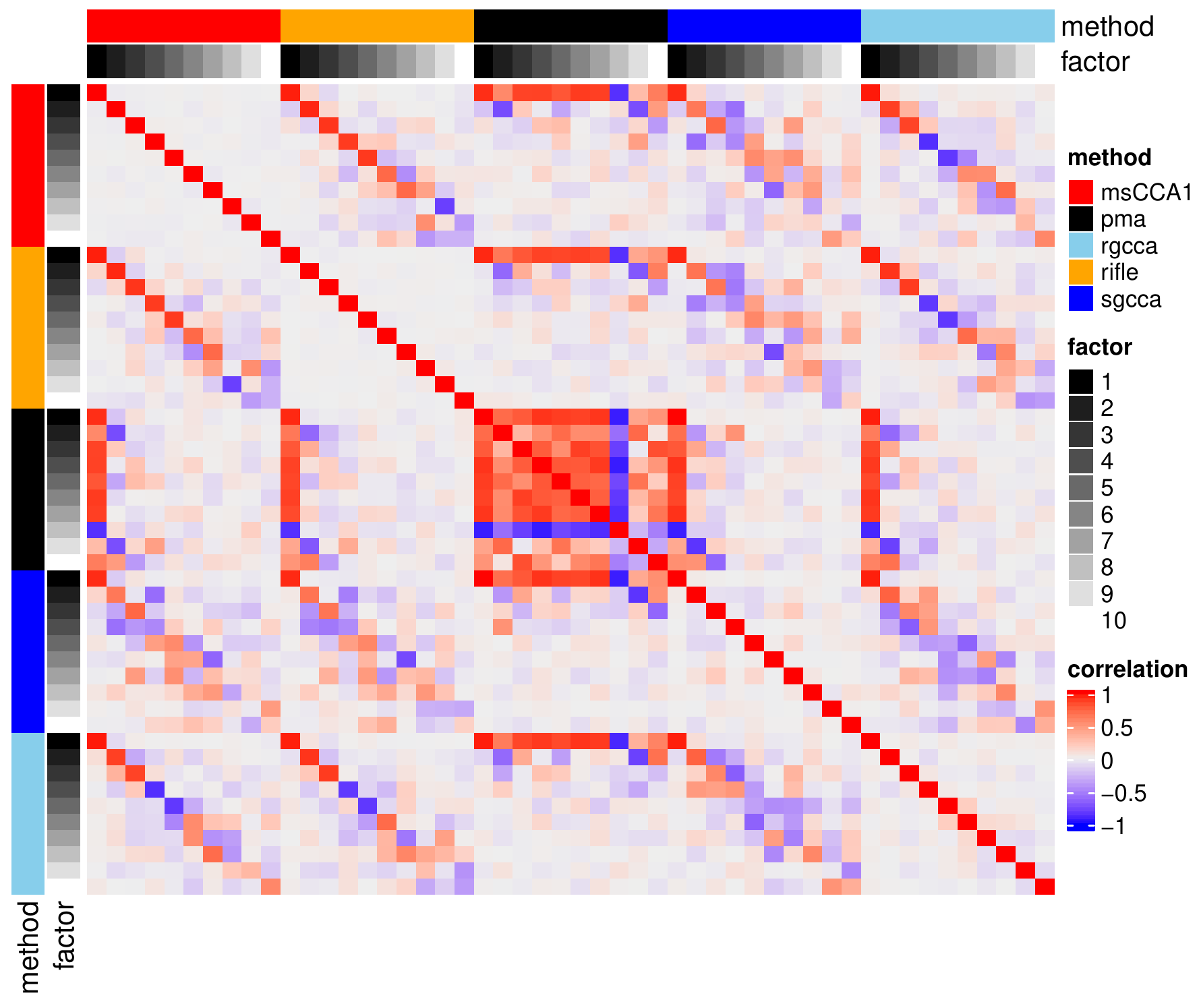}
\end{center}
\end{figure}

\section{Discussions}
We consider the problem of multi-block CCA estimation with high dimensional data and propose estimating the leading mCCA direction using proximal gradient descent (msCCA1) with decaying $\ell_1$ bounds. We show that the proposed procedure can lead to a  rate-optimal estimate for the leading mCCA direction under suitable model assumptions, and demonstrate its good empirical performance with intensive numerical examples. We also describe an easy-to-implement deflation procedure. When combined with methods like msCCA1 and rifle, we can estimate multiple directions sequentially. The sequential procedure allows easy adaptation to different sparsity levels underlying for different directions and allows users to add new directions when necessary easily.


\bibliographystyle{apalike}
\bibliography{mCCA}

\begin{thebibliography}{}

\bibitem[Cai et~al., 2013]{cai2013sparse}
Cai, T.~T., Ma, Z., and Wu, Y. (2013).
\newblock Sparse pca: Optimal rates and adaptive estimation.
\newblock {\em The Annals of Statistics}, 41(6):3074--3110.

\bibitem[Cai and Li, 2020]{cai2020inverse}
Cai, Y. and Li, P. (2020).
\newblock An inverse-free truncated rayleigh-ritz method for sparse generalized
  eigenvalue problem.
\newblock In {\em International Conference on Artificial Intelligence and
  Statistics}, pages 3460--3470. PMLR.

\bibitem[Chen et~al., 2013]{chen2013sparse}
Chen, M., Gao, C., Ren, Z., and Zhou, H.~H. (2013).
\newblock Sparse cca via precision adjusted iterative thresholding.
\newblock {\em arXiv preprint arXiv:1311.6186}.

\bibitem[Friedman et~al., 2010]{friedman2010regularization}
Friedman, J., Hastie, T., and Tibshirani, R. (2010).
\newblock Regularization paths for generalized linear models via coordinate
  descent.
\newblock {\em Journal of statistical software}, 33(1):1.

\bibitem[Gao et~al., 2015]{gao2015minimax}
Gao, C., Ma, Z., Ren, Z., and Zhou, H.~H. (2015).
\newblock Minimax estimation in sparse canonical correlation analysis.
\newblock {\em The Annals of Statistics}, 43(5):2168--2197.

\bibitem[Gao et~al., 2017]{gao2017sparse}
Gao, C., Ma, Z., and Zhou, H.~H. (2017).
\newblock Sparse cca: Adaptive estimation and computational barriers.
\newblock {\em The Annals of Statistics}, 45(5):2074--2101.

\bibitem[Gao and Ma, 2021]{gao2021sparse}
Gao, S. and Ma, Z. (2021).
\newblock Sparse gca and thresholded gradient descent.
\newblock {\em arXiv preprint arXiv:2107.00371}.

\bibitem[Gaynanova et~al., 2017]{gaynanova2017penalized}
Gaynanova, I., Booth, J.~G., and Wells, M.~T. (2017).
\newblock Penalized versus constrained generalized eigenvalue problems.
\newblock {\em Journal of Computational and Graphical Statistics},
  26(2):379--387.

\bibitem[Hardoon and Shawe-Taylor, 2011]{hardoon2011sparse}
Hardoon, D.~R. and Shawe-Taylor, J. (2011).
\newblock Sparse canonical correlation analysis.
\newblock {\em Machine Learning}, 83(3):331--353.

\bibitem[Jung et~al., 2019]{jung2019penalized}
Jung, S., Ahn, J., and Jeon, Y. (2019).
\newblock Penalized orthogonal iteration for sparse estimation of generalized
  eigenvalue problem.
\newblock {\em Journal of Computational and Graphical Statistics},
  28(3):710--721.

\bibitem[Kanatsoulis et~al., 2018]{kanatsoulis2018structured}
Kanatsoulis, C.~I., Fu, X., Sidiropoulos, N.~D., and Hong, M. (2018).
\newblock Structured sumcor multiview canonical correlation analysis for
  large-scale data.
\newblock {\em IEEE Transactions on Signal Processing}, 67(2):306--319.

\bibitem[Kettenring, 1971]{kettenring1971canonical}
Kettenring, J.~R. (1971).
\newblock Canonical analysis of several sets of variables.
\newblock {\em Biometrika}, 58(3):433--451.

\bibitem[Li et~al., 2009]{li2009joint}
Li, Y.-O., Adali, T., Wang, W., and Calhoun, V.~D. (2009).
\newblock Joint blind source separation by multiset canonical correlation
  analysis.
\newblock {\em IEEE Transactions on Signal Processing}, 57(10):3918--3929.

\bibitem[Ma, 2013]{ma2013sparse}
Ma, Z. (2013).
\newblock Sparse principal component analysis and iterative thresholding.
\newblock {\em The Annals of Statistics}, 41(2):772--801.

\bibitem[Mackey, 2008]{mackey2008deflation}
Mackey, L.~W. (2008).
\newblock Deflation methods for sparse pca.
\newblock In {\em NIPS}, volume~21, pages 1017--1024.

\bibitem[Meng et~al., 2014]{meng2014multivariate}
Meng, C., Kuster, B., Culhane, A.~C., and Gholami, A.~M. (2014).
\newblock A multivariate approach to the integration of multi-omics datasets.
\newblock {\em BMC bioinformatics}, 15(1):1--13.

\bibitem[Nielsen, 2002]{nielsen2002multiset}
Nielsen, A.~A. (2002).
\newblock Multiset canonical correlations analysis and multispectral, truly
  multitemporal remote sensing data.
\newblock {\em IEEE transactions on image processing}, 11(3):293--305.

\bibitem[Rodosthenous et~al., 2020]{rodosthenous2020integrating}
Rodosthenous, T., Shahrezaei, V., and Evangelou, M. (2020).
\newblock Integrating multi-omics data through sparse canonical correlation
  analysis for the prediction of complex traits: a comparison study.
\newblock {\em Bioinformatics}, 36(17):4616--4625.

\bibitem[Sch{\"a}fer and Strimmer, 2005]{schafer2005shrinkage}
Sch{\"a}fer, J. and Strimmer, K. (2005).
\newblock A shrinkage approach to large-scale covariance matrix estimation and
  implications for functional genomics.
\newblock {\em Statistical applications in genetics and molecular biology},
  4(1).

\bibitem[Sriperumbudur et~al., 2011]{sriperumbudur2011majorization}
Sriperumbudur, B.~K., Torres, D.~A., and Lanckriet, G.~R. (2011).
\newblock A majorization-minimization approach to the sparse generalized
  eigenvalue problem.
\newblock {\em Machine learning}, 85(1-2):3--39.

\bibitem[Subramanian et~al., 2020]{subramanian2020multi}
Subramanian, I., Verma, S., Kumar, S., Jere, A., and Anamika, K. (2020).
\newblock Multi-omics data integration, interpretation, and its application.
\newblock {\em Bioinformatics and biology insights}, 14:1177932219899051.

\bibitem[Sui et~al., 2012]{sui2012review}
Sui, J., Adali, T., Yu, Q., Chen, J., and Calhoun, V.~D. (2012).
\newblock A review of multivariate methods for multimodal fusion of brain
  imaging data.
\newblock {\em Journal of neuroscience methods}, 204(1):68--81.

\bibitem[Suo et~al., 2017]{suo2017sparse}
Suo, X., Minden, V., Nelson, B., Tibshirani, R., and Saunders, M. (2017).
\newblock Sparse canonical correlation analysis.
\newblock {\em arXiv preprint arXiv:1705.10865}.

\bibitem[Tan et~al., 2018]{tan2018sparse}
Tan, K.~M., Wang, Z., Liu, H., and Zhang, T. (2018).
\newblock Sparse generalized eigenvalue problem: Optimal statistical rates via
  truncated rayleigh flow.
\newblock {\em Journal of the Royal Statistical Society: Series B (Statistical
  Methodology)}, 80(5):1057--1086.

\bibitem[Tenenhaus et~al., 2014]{tenenhaus2014variable}
Tenenhaus, A., Philippe, C., Guillemot, V., Le~Cao, K.-A., Grill, J., and
  Frouin, V. (2014).
\newblock Variable selection for generalized canonical correlation analysis.
\newblock {\em Biostatistics}, 15(3):569--583.

\bibitem[Tenenhaus and Tenenhaus, 2011]{tenenhaus2011regularized}
Tenenhaus, A. and Tenenhaus, M. (2011).
\newblock Regularized generalized canonical correlation analysis.
\newblock {\em Psychometrika}, 76(2):257.

\bibitem[Tenenhaus et~al., 2017]{tenenhaus2017regularized}
Tenenhaus, M., Tenenhaus, A., and Groenen, P.~J. (2017).
\newblock Regularized generalized canonical correlation analysis: a framework
  for sequential multiblock component methods.
\newblock {\em Psychometrika}, 82(3):737--777.

\bibitem[Vu and Lei, 2013]{vu2013minimax}
Vu, V.~Q. and Lei, J. (2013).
\newblock Minimax sparse principal subspace estimation in high dimensions.
\newblock {\em The Annals of Statistics}, 41(6):2905--2947.

\bibitem[Witten et~al., 2009]{witten2009penalized}
Witten, D.~M., Tibshirani, R., and Hastie, T. (2009).
\newblock A penalized matrix decomposition, with applications to sparse
  principal components and canonical correlation analysis.
\newblock {\em Biostatistics}, 10(3):515--534.

\bibitem[Witten and Tibshirani, 2009]{witten2009extensions}
Witten, D.~M. and Tibshirani, R.~J. (2009).
\newblock Extensions of sparse canonical correlation analysis with applications
  to genomic data.
\newblock {\em Statistical applications in genetics and molecular biology},
  8(1).

\bibitem[Yuan and Zhang, 2013]{yuan2013truncated}
Yuan, X.-T. and Zhang, T. (2013).
\newblock Truncated power method for sparse eigenvalue problems.
\newblock {\em Journal of Machine Learning Research}, 14(4).

\end{thebibliography}
\newpage
\appendix
\begin{center}
{\textbf{\Large Supplement to ``$\ell_1$-norm constrained multi-block sparse canonical correlation analysis via proximal gradient descent"}}
\end{center}

Throughout our theoretical analysis, we always assume that Assumptions \ref{ass:regularity} - \ref{ass:sparsity} hold. In Appendix \ref{app:supporting}, we collect some basic Lemmas and Propositions, with proofs deferred to Appendix \ref{app:proof_supportings}. In Appendix \ref{app:lemmas}, we present several technical Lemmas that are essential to prove our main Theorems, with proofs deferred to Appendix \ref{app:proof_technical}.   In Appendix \ref{app:proof_main}, we present proofs to Theorems, Lemmas and Propositions in the main paper.  Finally, we include some left-out details of the initialization used in our empirical studies  in Appendix \ref{app:algorithm}.

\section{Supporting Propositions and Lemmas}
\label{app:supporting}
In this Appendix, we  collect  elementary results on the tail bounds and some characterization of the estimation quality. For any constant $k'\geq 1$, we define the following events to bound the estimation errors:
\begin{align*}
\mA_1(k') = \left\{\|\hat H_{J_1 J_2} -H_{J_1 J_2}\|_{op}\leq C\lambda_{H}^{\max}\sqrt{\frac{k's\ln p}{n}},  \forall |J_1|,|J_2|\leq k's, H \in \{\Lambda,\Sigma\}\right\}.
\end{align*}
as well as 
\begin{align*}
&\mA_2 = \left\{|\xi_1^\top (\hat H -H)\xi_1|\leq C\frac{\lambda^{\max}_{H}}{M^3}\sqrt{\frac{\ln p}{n}}, H \in \{\Lambda,\Sigma\}\right\},\\
&\mA_3 = \left\{\|\xi_1^\top (\hat H-H)\|_{\infty}\leq C\lambda^{\max}_{H}\sqrt{\frac{\ln p}{n}},H \in \{\Lambda,\Sigma\}\right\}.
\end{align*}
We let $C>0$ be a universal constant that may change from instance to instance. Let $\mA(k') = \mA_1(k')\cap \mA_2\cap \mA_3$ be the intersection of events, then $\mA(k')$ happens with high probability.
\begin{proposition}
\label{prop:prop1}
For any given constant $k\geq 1$, there exist a sufficiently large $C >0$ such that $\bP(\mA(k'))\rightarrow 1$ happens with high probability as $n\rightarrow \infty$.
\end{proposition}
For any unit vector $\beta$, it can be written as $\beta = \sum_{j=1}^p \alpha_j \xi_j$. We define $\delta(\beta)\coloneqq 1-|\xi_1^\top\beta|$ to measure its discrepancy from $\xi_1$ and $\mu_j = \xi_j^\top\Lambda\xi_j$. Then, small $\delta(\beta)$ is equivalent to small distance between $\alpha$ and $(1,0,\ldots, 0)^\top$.  
\begin{proposition}
\label{prop:prop2}
For any unit vector $\beta= \sum_{j=1}^p \alpha_j\xi_j$, we  have 
\begin{align}
&\sum_{i=1}^p \alpha_i^2\geq \frac{1}{M},\label{eq:prop2eq1}\\
&\sum_{i\geq 2}\alpha_i^2\mu_i\geq \frac{2}{M}\delta(\beta)(1-\frac{\delta(\beta)}{2}),\label{eq:prop2eq2}\\
&\sum_{j\geq 2}\alpha_j^2\mu_j+(1-\alpha_1)^2\mu_1 \leq 2M\delta(\beta).\label{eq:prop2eq3}
\end{align}
\end{proposition}
We define $m_1(\beta) \coloneqq \frac{\sum_{j\geq 2}\alpha_j^2\mu_j}{\delta(\beta)}\in [\frac{1}{M}, 2M]$ as a direct result from Proposition \ref{prop:prop2} (\ref{eq:prop2eq2})-(\ref{eq:prop2eq3}), and also $m_{2}(\beta) = \frac{\beta^\top\Sigma\beta}{\rho_1}$. These two quantities will appear in our intermediate results characterizing the influence on various quantities due to the difference between $\beta$ and $\xi_1$. We can also bound $m_2(\beta)$.
\begin{proposition}
\label{prop:prop3}
(1) For all unit vector $\beta$, we have $m_{2}(\beta)\leq \beta^\top\Lambda\beta$. (2) When $\delta(\beta)\leq \frac{c^2}{8M^2}$ for  $0\leq c \leq \frac{1}{2}$, we have  
\[
(1-\frac{c^2}{2})\beta^\top\Lambda\beta\leq m_2(\beta)\leq \beta^\top\Lambda\beta\leq (1+\frac{c}{2})^2\mu_1,\; \min(\beta^\top\Lambda\beta, m_2(\beta))\geq (1-\frac{c}{2})^2\mu_1.
\]
\end{proposition}
Let $h(\beta)=\beta - \xi_1$ and $\omega(\beta) = \|\beta\|_1 - \|\xi_1\|_1$. We can bound $\|h(\beta)\|_1$ using $\omega(\beta)$.
\begin{proposition}
\label{prop:prop4}
For any unit vector $\beta$, we have  $\|h(\beta)\|_1\leq 2\sqrt{2s\delta(\beta)}+\omega(\beta)$.
\end{proposition}
Using the above Propositions, we can derive several basic tail bound results.  Let $\hat\rho_j\coloneqq f(\xi_1)$ denote the achieved empirical correlations with the population leading mCCA direction $\xi_1$, and $\mu_j \coloneqq \xi_j^\top\Lambda\xi_j$, $\hat\mu_j \coloneqq \xi_j^\top\hat\Lambda\xi_j$ for $j=1,\ldots, p$.  We also set $\bar{\omega}(\beta) = [\omega(\beta)]_+$ be the positive part of $\omega(\beta)$.
\begin{lemma}
\label{lem:basic_bound1}
Suppose that the event $\mA(k')$ holds for some constant $k'\geq 1$. There exists a sufficiently large  universal constant $C$,  such that for all $\zeta$ and  $H\in \{\Lambda, \Sigma\}$, we have
\begin{align}
&|\zeta^\top(\hat H - H)\zeta|\leq C\lambda_{H}^{\max}\sqrt{\frac{k's\ln p}{n}}\left[\|\zeta\|_2^2+\frac{\|\zeta\|_1^2}{k's}\right],\label{basic1main1}\\
& |\zeta^\top(\hat H - H)\xi_1|\leq C\lambda_{H}^{\max}\sqrt{\frac{s\ln p}{n}}(\|\zeta\|_2+\frac{\|\zeta\|_1}{\sqrt{s}}),\label{basic1main2}
\end{align}
When $\zeta = h(\beta)$, we have
\begin{align}
&|h(\beta)^\top(\hat H - H)h(\beta)|\leq C\lambda_{H}^{\max}\sqrt{\frac{k's\ln p}{n}}\left[\delta(\beta)+\frac{\bar{\omega}(\beta)^2}{k's}\right],\label{basic1main3}\\
& |h(\beta)^\top(\hat H - H)\xi_1|\leq C\lambda_{H}^{\max}\sqrt{\frac{s\ln p}{n}}(\sqrt{\delta(\beta)}+\frac{\bomega(\beta)}{\sqrt{s}}).\label{basic1main4}
\end{align}
\end{lemma}

\begin{lemma}
\label{lem:basic_bound2}
Suppose that the event $\mA(k')$ holds for some constant $k'\geq 1$. Define:
\begin{align}
&F_{\Lambda}(\beta) = \alpha_1^2\hat\mu_1+\sum_{j\geq 2}\alpha_j^2\mu_j,\label{basic2main01}\\
&F_{\Sigma}(\beta) = \alpha_1^2\hat\mu_1\hat\rho_1+\sum_{j\geq 2}\alpha_j^2\mu_j\rho_j\label{basic2main02}.
\end{align}
There exist a sufficiently large universal $C$, such that for  all unit vector $\beta$ and $H\in \{\Lambda, \Sigma\}$,  we have
\begin{align}
&|\xi^\top_1 \hat{H}(\beta-\alpha_1\xi_1)|\leq C\lambda^{\max}_H\sqrt{\frac{\ln p}{n}}(\sqrt{\delta(\beta)}+\frac{\bomega(\beta)}{\sqrt{s}}),\label{basic2main1}\\
&|\beta^\top\hat H\beta - F_H|\leq C\lambda^{\max}_H\sqrt{\frac{s\ln p}{n}}\left(\sqrt{\delta(\beta)}+\frac{\bomega(\beta)}{\sqrt{s}}+\sqrt{k'}\delta(\beta)+\frac{\bomega(\beta)^2}{\sqrt{k'}s}\right).\label{basic2main2}
\end{align}
\end{lemma}

\begin{lemma}
\label{lem:basic_bound3}
Suppose that  $\mA(1)$  holds, $|\beta|\leq (1+c_{B_1})\sqrt{s}$, and $\frac{\ln p}{n}\leq c^2$ for a constant $ 0<c < \frac{1}{2}$. There exists a sufficiently large universal constant $C$, such that for all such unit vector $\beta$, we have
\begin{align}
&|\hat\rho_1 - \rho_1|\leq \frac{C}{M}\rho_1\sqrt{\frac{\ln p}{n}}, \label{basic3main1}\\
&|\beta^\top(\hat H-H)\beta|\leq C(1+c_{B_1})^2\lambda_{H}^{\max}\sqrt{\frac{s\ln p}{n}},\;H\in\{\Lambda, \Sigma\}.\label{basic3main2}
\end{align}
\end{lemma}
Proofs of Propositions and Lemmas in this Section are deferred to Appendix \ref{app:proof_supportings}
\section{Technical Lemmas}
\label{app:lemmas}
We provide several technical Lemmas that will be used for proving  results in the main paper. For the convenience of notations,  we use the abbreviations $\hat r_t\coloneqq f(\beta_t)$, $\delta_t\coloneqq \delta(\beta_t)$, $m_{1t} = m_1(\beta_t)$, $m_{2t} = m_{2}(\beta_t)$, $\omega_t = \omega(\beta_t)$ and $\bomega_t = \bomega(\beta_t)$. We write $\beta_t$ as a linear combination of $\{\xi_j\}$ and $\beta_t = \sum_{j=1}^p \alpha_{jt}\xi_j$.   For any constant $0<c < \frac{1}{2}$ and $k > 2$, we define 
\begin{equation}
\label{eq:delta_bound1}
\delta^{upper}_t(c,k) = \min\{\frac{c^2}{8M^2}, \frac{cB_t}{\|\xi_1\|_1},\frac{c\gamma\eta}{M}(\frac{k-1}{k})^2(1-c)^2 \}, \;\mbox{for all }t\geq 0
\end{equation}
Since $B_0 = c_{B_1}\sqrt{s}$, when $t = 0$,
\[
\delta^{upper}_0(c,k) \geq \min\{\frac{c^2}{8M^2}, cc_{B_1},\frac{c\gamma\eta}{M}(\frac{k-1}{k})^2(1-c)^2 \}
\]
 is lower bounded by a positive constant.   When $\eta$, $c_{B2}$ and $\nu$ are in the assumed range in Theorem \ref{thm:thm1}, we can find constants $0<c \leq\frac{1}{2}, \;k \geq 2$ such that
\begin{align}
\label{condI}
\left\{\begin{array}{ll}& (\frac{k-1}{k})^2(1-c)^4\geq \frac{1}{3},\\
&  \frac{\gamma^2}{M}(\frac{k-1}{k})^5(1-c)^{8}\geq \nu,\\
&\frac{1}{2(1+c)^2M(M+3)}\geq \eta.
\end{array}\right.
\end{align}
We fix $(c, k)$ as positive constants to make (\ref{condI}) hold and achieves the smallest $\frac{k}{c}$. Note that the constants $(c, k)$ depend only on $(M, \gamma, \eta, c_0)$ (other choices of $(k,c )$ is allowed, we fix them as described to remove extra dependence on $(c, k)$ for constants appearing in the proofs). We also let $\iota>0$ be some positive constant that can depend on the true model parameters $(\gamma, M, \eta, c_0, c_{B_1}, c_{B_2})$, and whose values will vary in different Lemmas and proofs.

To make our statements and proofs more friendly to read, we define the following notation $\mO(.)$, and use $\mO(b_n)$ to represent $\psi b_n$ for some positive constant $\psi$ that may depend only on $(M, \gamma, \eta, c_0, c_{B_1}, c_{B_2})$. For example, if we state that $b_n'\leq \mO(b_n)$, it means that there exists a sufficiently large $(M, \gamma, \eta, c_0, c_{B_1}, c_{B_2})$-dependent constant $\psi$ such that  $b_n'\leq  \psi b_n$.

Below, Lemma \ref{lem:coef_bound} links $(\hat\rho_1 - \hat r_t)$ to  $\delta_t$. The former is the gap between optimal objective and objective achieved at iteration $t$, and the later is  the ``distance" between $\beta_t$ and $\xi_1$. Lemma (\ref{lem:proximal}) characterizes the lower bound of the proximal objective improvement at iteration $t\leq T^*$, and Lemma (\ref{lem:obj_lowerbound}) lower bounds the improvement of the actual objective using the  improvement of the proximal objective.  When combining these three Lemmas, we will be able to analysis the upper bounds of $\delta_t$ and $(\hat\rho_1 - \hat r_t)$ over the iterations $1\leq t\leq T^*$. The exact proofs are given after the Lemma statements.
\begin{lemma}
\label{lem:coef_bound}
Suppose that  $\mA(1)\cap \mA(k')$ holds for some $k' \geq 1$. There exists a sufficiently small constant $\iota$ such that when  $\sqrt{\frac{k's\ln p}{n}}\leq \iota$, for all $\beta_t$, we have,
\begin{align}
&[\hat\rho_1 - \hat r_t]_+\leq \rho_1 \left[\frac{k+1}{k}\frac{ m_{1t}\delta_t}{\beta_t^\top\Lambda\beta_t}+\mO\left(\frac{ks\ln p}{n}+\sqrt{\frac{s\ln p}{n}}\left(\frac{\bomega_t}{\sqrt{s}}+\frac{\bomega_t^2}{\sqrt{k'}s}\right)\right)\right],\label{eq:coef_main1}\\
&\hat\rho_1 - \hat r_t\geq \rho_1\left[ \frac{k-1}{k}\frac{\gamma m_{1t}\delta_t}{\beta_t^\top\Lambda\beta_t}-\mO\left(\frac{ks\ln p}{n}+\sqrt{\frac{s\ln p}{n}}\left(\frac{\bomega_t}{\sqrt{s}}+\frac{\bomega_t^2}{\sqrt{k'}s}\right)\right)\right].\label{eq:coef_main2}
\end{align}
For all $\beta_t$ with $\delta_t\leq \frac{c^2}{8M}$, 
\begin{align}
&\frac{[\hat\rho_1 - \hat r_t]_+}{\hat r_t}\leq  \frac{k+1}{k}\frac{ m_{1t}\delta_t}{m_{2t}}+\mO\left(\frac{ks\ln p}{n}+\sqrt{\frac{s\ln p}{n}}\left(\frac{\bomega_t}{\sqrt{s}}+\frac{\bomega_t^2}{\sqrt{k'}s}\right)\right),\label{eq:coef_main3}\\
&\frac{\hat\rho_1 - \hat r_t}{\hat r_t}\geq  \frac{k-1}{k}\frac{\gamma m_{1t}\delta_t}{m_{2t}}-\mO\left(\frac{ks\ln p}{n}+\sqrt{\frac{s\ln p}{n}}\left(\frac{\bomega_t}{\sqrt{s}}+\frac{\bomega_t^2}{\sqrt{k'}s}\right)\right).\label{eq:coef_main4}
\end{align}
\end{lemma}
\begin{lemma}
\label{lem:proximal}
Set $\Delta_{t,t+1} = \frac{1}{2}(\|\theta-\beta_t\|_2^2-\|\theta-\beta_{t+1}\|_2^2)$ to measure the proximal objective improvement at iteration $t$ where $\theta$ is defined as (\ref{eq:theta}).  Suppose that  $\mA(1)\cap \mA(k')$ and inequalities (\ref{eq:coef_main1}) - (\ref{eq:coef_main4}) hold for some $k' \geq 1$.  There exists a sufficiently small constant $\iota$ such that when  $\sqrt{\frac{k's\ln p}{n}}\leq \iota$,  for all $\beta_t$ with $\delta_t\leq \delta_t^{upper}(c,k)$ and $t\leq T^*$, we have 
\[
 \Delta_{t,t+1} \geq \frac{\eta^2\gamma^2m_{1t}^2}{4}(\frac{k-1}{k})^4(1-c)^6\delta_t-\eta^2\mO\left(\frac{ks\ln p}{n}+\sqrt{\frac{s\ln p}{n}}\left(\frac{\bomega_t}{\sqrt{s}}+\frac{\bomega_t^2}{s}\right)\right).
\]
\end{lemma}
\begin{lemma}
\label{lem:obj_lowerbound}
Suppose that  $\mA(1)\cap \mA(k')$  and  inequalities (\ref{eq:coef_main1}) - (\ref{eq:coef_main4}) hold  for some $k' \geq 1$. There exists a sufficiently small constant $\iota$ such that when  $\sqrt{\frac{k's\ln p}{n}}\leq \iota$, for all $\beta_t$ with $\beta_t\leq \frac{c^2}{8M}$, we have
\begin{equation}
\label{eq:obj_main1}
\hat r_{t+1} - \hat r_{t}\geq \frac{1}{\beta_t^\top\hat\Lambda\beta_t}\left[2\frac{\hat r_t}{\eta}\Delta_{t,t+1}-\mO\left(\hat r_t
\sqrt{\frac{s\ln p}{k' n}}(\delta_t+\delta_{t+1}+\frac{\bomega_t^2}{k's})\right)\right].
\end{equation}
\end{lemma}
Proofs of Lemma \ref{lem:coef_bound} - Lemma \ref{lem:obj_lowerbound} are deferred to Appendix \ref{app:proof_technical}.
\section{Proofs of main Theorems, Lemmas and Propositions}
\label{app:proof_main}
\subsection{Proof of Theorem \ref{thm:thm1}}
We can consider the case when $k' = 1$ and 
\begin{itemize}
\item   $\mA(1)$ holds and  $\sqrt{\frac{s\ln p}{n}}\leq \iota$ is sufficiently such that statements in Lemma \ref{lem:coef_bound} - Lemma \ref{lem:obj_lowerbound} hold, and  $\sqrt{\frac{s\ln p}{n}}\leq \iota$ for some small $\iota$ that we will specify later in our proof.
\item   The initial guess satisfies that
\begin{equation}
\label{eq:initial_requirement}
\delta_0 \leq \psi_1 \coloneqq \frac{(k-1)(1-c)^2\gamma }{2(k+1)M^4} \delta_0^{upper}(c,k).
\end{equation}
\end{itemize} 
We aim to show that when $\delta_t\leq \delta_t^{upper}(c,k)$  for all  $0\leq t \leq  T < T^*$, then, for all  $0\leq t \leq  T+1$, we have
\begin{align}
&\delta_{t}\leq \delta_{t}^{upper}(c,k),\label{eq:thm1eq0}\\
&(\hat\rho_1 - \hat r_{t})\leq \frac{(k+1)M}{k(1-c)\mu_1}(1-\nu\eta)^{t}\delta_0\rho_1+\rho_1\mO\left(\frac{ks\ln p}{n}+\sqrt{\frac{s\ln p}{n}}\left(\frac{B_t}{\sqrt{s}}+\frac{B_t^2}{s}\right)\right).\label{eq:thm1eq1}\\
&\delta_{t}\leq \frac{2(k+1)M^2}{(k-1)(1-c)\gamma}\frac{\beta_{t}^\top\Lambda\beta_{t}}{\mu_1}(1-\nu\eta)^{t}\delta_0+\mO\left(\frac{ks\ln p}{n}+\sqrt{\frac{s\ln p}{n}}\left(\frac{B_t}{\sqrt{s}}+\frac{B_t^2}{s}\right)\right).\label{eq:thm1eq2}
\end{align}
If so,  by induction, (\ref{eq:thm1eq0}) - (\ref{eq:thm1eq2}) hold for all $0\leq t\leq T^*$. From Proposition \ref{prop:prop3}, we have
\begin{equation}
\label{eq:thm1eq_prop3a}
(1-\frac{c}{2})^2\mu_1\leq \beta_t^\top\Lambda\beta_t\leq (1+\frac{c}{2})^2\mu_1.
\end{equation}
Hence, for all $1\leq t\leq T^*$ ,  since $k$ is also a constant depending on $(\eta, \gamma, M)$, we have
\begin{align}
&\delta_{t}\leq \frac{2(k+1)M^2}{(k-1)(1-c)^2\gamma}\delta_0(1-\nu\eta)^{t}+\mO\left(\frac{s\ln p}{n}+\sqrt{\frac{s\ln p}{n}}\left(\frac{B_t}{\sqrt{s}}+\frac{B_t^2}{s}\right)\right).\label{eq:thm1eq2_new}
\end{align}
To summarize,  we shall prove  that (\ref{eq:thm1eq0}), (\ref{eq:thm1eq1}) and (\ref{eq:thm1eq2}) hold for all $0\leq t\leq T+1$ when (1) $\delta_t\leq \delta^{upper}_{t}(c,k)$ for all $0\leq t\leq T < T^*$, (2)  $\mA(1)$ and (\ref{eq:initial_requirement}) hold, and (3)  $\sqrt{\frac{s\ln p}{n}}\leq \iota$ for a sufficiently small constant that does not depend on $T$.
 \subsubsection*{Proof of (\ref{eq:thm1eq1}) and (\ref{eq:thm1eq2})}
When $\delta_t\leq \delta_t^{upper}(c,k)$, by Lemma \ref{lem:proximal}, Lemma \ref{lem:obj_lowerbound} and $\bomega_t\leq B_t$, we have
{\small
\begin{align}
\label{eq:thm1eq3}
& \hat r_{t+1} - \hat r_t\geq \frac{\hat r_t}{\beta_t^\top\hat\Lambda\beta_t}\left[\frac{\eta\gamma^2m_{1t}^2}{2}(\frac{k-1}{k})^4(1-c)^6\delta_t-\mO\left(\frac{ks\ln p}{n}+\sqrt{\frac{s\ln p}{n}}(\frac{B_t}{\sqrt{s}}+\frac{B_t^2}{s})+\sqrt{\frac{s\ln p}{n}}(\delta_t+\delta_{t+1}+\frac{B_t^2}{s})\right)\right]\notag\\
&\xLeftrightarrow{\quad\quad}\notag\\
&(\hat \rho_1 - \hat r_{t+1}) - I_1\leq  (\hat \rho_1 - \hat r_{t})-I_2+I_3
\end{align}
}
where  we have used the fact that $\hat r_{t+1}-\hat r_t = -(\hat\rho_1 -\hat r_{t+1}) +(\hat\rho_1 -\hat r_{t})$ and set $I_1$, $I_2$, $I_3$ as
\begin{align*}
&I_1 =\frac{\hat r_t}{\beta_t^\top\hat\Lambda\beta_t}\mO(\sqrt{\frac{s\ln p}{n}})\delta_{t+1},\\
&I_2 = \frac{\hat r_t}{\beta_t^\top\hat\Lambda\beta_t}\times \frac{\eta\gamma^2m_{1t}^2}{2}(\frac{k-1}{k})^4(1-c)^6\left(1-\mO(\sqrt{\frac{s\ln p}{n}})\right)\delta_t,\\
&I_3 = \frac{\hat r_t}{\beta_t^\top\hat\Lambda\beta_t}\mO\left( \frac{ks\ln p}{n}+\sqrt{\frac{s\ln p}{n}}\left(\frac{B_t}{\sqrt{s}}+\frac{B_t^2}{s}\right)\right).
\end{align*}
By Lemma \ref{lem:basic_bound3} (\ref{basic3main2}), for all unit vector $\beta$ with $\|\beta\|_1\leq (1+c_{B_1})\sqrt{s}$:
\begin{align}
&\hat f(\beta) \leq \frac{\rho_1\beta^\top\Lambda\beta\left(1+\mO\left((1+c_{B_1})^2\sqrt{\frac{s\ln p}{n}}\right)\right)}{\beta^\top \Lambda \beta\left(1 - \mO\left((1+c_{B_1})^2\sqrt{\frac{s\ln p}{n}}\right)\right)} \overset{(a_1)}{\leq} \rho_1\frac{1+\frac{c}{16}}{(1-\frac{c}{16})}\leq \rho_1(1+\frac{c}{4}),\label{size_bound1}\\
&\frac{1}{\beta^\top\hat\Lambda\beta}\leq \frac{1}{\beta^\top\Lambda\beta\left(1 - \mO\left((1+c_{B_1})^2\sqrt{\frac{s\ln p}{n}}\right)\right)} \overset{(a_2)}{\leq}\frac{1}{\beta^\top\Lambda\beta(1-\frac{c}{16})}\leq \frac{1+\frac{c}{8}}{\beta^\top\Lambda\beta},\label{size_bound2}\\
&\frac{1}{\beta^\top\hat\Lambda\beta}\geq \frac{1}{\beta^\top\Lambda\beta\left(1 +\mO\left((1+c_{B_1})^2\sqrt{\frac{s\ln p}{n}}\right)\right)} \overset{(a_3)}{\leq}\frac{1}{\beta^\top\Lambda\beta(1+\frac{c}{16})}\leq \frac{1-\frac{c}{8}}{\beta^\top\Lambda\beta},\label{size_bound3}
\end{align}
where $(a_1)$, $(a_2)$ and $(a_3)$ hold when $\iota$ is small such that  $\mO((1+c_{B_1})^2\sqrt{\frac{s\ln p}{n}})\leq\frac{c}{16}$.

Combine (\ref{size_bound1}) and (\ref{size_bound2}) with the expression for $I_1$, when $\mO(\sqrt{\frac{s\ln p}{n}})\leq \frac{\gamma c'(k-1)}{2kM^3}$ with $c'$ a small positive constant that we will specify later, we have
\begin{equation}
\label{eq:thm1I1}
I_1\leq \frac{\gamma (k-1)c'}{k M^2}\rho_1\delta_{t+1}.
\end{equation}
Combine (\ref{size_bound1}) and (\ref{size_bound3}) with the expression for $I_2$, when  $\mO(\sqrt{\frac{s\ln p}{n}})\leq \frac{c}{4}$, we have
\begin{equation}
\label{eq:thm1I2}
I_2\geq \frac{\hat r_t}{\beta_t^\top\Lambda\beta_t}\times  \frac{\eta\gamma^2m_{1t}^2}{2}(\frac{k-1}{k})^4(1-c)^6\left(1-\frac{c}{2}\right)\delta_t.
\end{equation}
Combine (\ref{size_bound1}) and (\ref{size_bound2}) with the expression for $I_3$, we obtain
\begin{equation}
\label{eq:thm1I3}
I_3\leq \rho_1\mO(\frac{ks\ln p}{n}+\sqrt{\frac{s\ln p}{n}}(\frac{B_t}{\sqrt{s}}+\frac{B_t^2}{s})).
\end{equation}
By Proposition \ref{prop:prop2}:
\begin{equation}
\label{eq:thm1eq_prop2a}
\frac{1}{M}\leq \frac{2}{M}(1-\delta_t)\leq m_{1t}\leq 2M.
\end{equation}
We can upper bound $\delta_{t+1}$ in $I_1$ with Lemma \ref{lem:coef_bound} (\ref{eq:coef_main2}) and (\ref{eq:thm1eq_prop2a}), and that $\bomega_{t+1}\leq B_t, \bomega^2_{t+1}\leq B^2_t$, which results in
\begin{equation}
\label{eq:thm1I1b}
I_1\leq c'\left(\hat\rho_1 - \hat r_{t+1}\right) +\rho_1\mO\left(\frac{ks\ln p}{n}+\sqrt{\frac{s\ln p}{n}}\left(\frac{B_t}{\sqrt{s}}+\frac{B_{t}^2}{s}\right)\right).
\end{equation}
Similarly, we can lower bound $\delta_t$ in $I_2$ with Lemma \ref{lem:coef_bound} (\ref{eq:coef_main3})  and (\ref{eq:thm1eq_prop2a}), (\ref{size_bound1}), and that $\bomega_{t}\leq B_t, \bomega^2_{t}\leq B^2_t$, which results in
{\footnotesize
\begin{equation}
\label{eq:thm1I2b}
I_2\geq  \frac{\eta\gamma^2m_{1t}m_{2t}}{2\beta_t^\top\Lambda\beta_t}(\frac{k-1}{k})^5(1-c)^6\left(1-\frac{c}{2}\right)[\hat\rho_1 - \hat r_t]_+- \rho_1\mO\left(\frac{ks\ln p}{n}+\sqrt{\frac{s\ln p}{n}}\left(\frac{B_t}{\sqrt{s}}+\frac{B_{t}^2}{s}\right)\right).
\end{equation}
}
We define 

\begin{align}
\label{eq:thm1eq4}
R_{1t} &= \frac{\eta\gamma^2m_{1t}m_{2t}}{2\beta_t^\top\Lambda\beta_t}(\frac{k-1}{k})^5(1-c)^6\left(1-\frac{c}{2}\right)\notag\\
&\geq \frac{\eta \gamma^2(1-\frac{\delta_t}{2})(1-\frac{c^2}{2})}{M} (\frac{k-1}{k})^5(1-c)^6\left(1-\frac{c}{2}\right)\geq  \frac{\eta \gamma^2}{M} (\frac{k-1}{k})^5(1-c)^7.
\end{align}
At the last display, we have used (\ref{eq:thm1eq_prop2a}) to lower bound $m_{1t}$, $\frac{c^2}{8M^2}$ to upper bound $\delta_t$ and Proposition \ref{prop:prop3} to lower bound $\frac{m_{2t}}{\beta_t^\top\Lambda\beta_t}$, which says
\begin{equation}
\label{eq:thm1eq_prop3b}
(1-\frac{c^2}{2})\leq \frac{m_{2t}}{\beta_t^\top\Lambda\beta_t}\leq 1.
\end{equation}
Plug in the bounds of $I_1$, $I_2$, $I_3$ in (\ref{eq:thm1I1b}),  (\ref{eq:thm1I2b}), (\ref{eq:thm1I3}) and our definition of $R_{1t}$ into (\ref{eq:thm1eq3}), we obtain that
\begin{equation}
\label{eq:thm1eq5}
\hat\rho_1 - \hat r_{t+1}\leq \frac{1-R_{1t}}{1-c'}[\hat\rho_1 - \hat r_t]_+ +  \rho_1\mO\left(\frac{ks\ln p}{n}+\sqrt{\frac{s\ln p}{n}}\left(\frac{B_t}{\sqrt{s}}+\frac{B_{t}^2}{s}\right)\right).
\end{equation}
We now give the definition of $c'$. We define $c' = cR_{1t}\geq \frac{c\eta \gamma^2}{M} (\frac{k-1}{k})^5(1-c)^7$, which is lower bounded by a positive constant. We  lower and upper bound $\frac{1-R_{1t}}{1-c'}$. On the one hand, by (\ref{eq:thm1eq4}) and $\eta < \frac{1}{2M(M+3)}$, we know
\begin{equation}
\label{eq:thm1eq6}
R_{1t}\leq  \frac{1}{8}\Rightarrow \frac{1-R_{1t}}{1-c'}> 0.
\end{equation}
On the other hand, we have
\begin{equation}
\label{eq:thm1eq7}
\frac{1-R_{1t}}{1-c'} = \frac{1-R_{1t}}{1-cR_{1t}}\leq 1-(1-c)R_{1t}\leq 1- \frac{\eta \gamma^2}{M}(\frac{k-1}{k})^5(1-c)^8\leq 1-\nu\eta,
\end{equation}
where the last step in (\ref{eq:thm1eq7}) uses condition (\ref{condI}). Hence, we have
\begin{equation}
\label{eq:thm1eq8}
(\hat\rho_1 - \hat r_{t+1})\leq (1-\nu\eta)[\hat\rho_1 - \hat r_{t}]_{+}+\rho_1\mO\left(\frac{ks\ln p}{n}+\sqrt{\frac{s\ln p}{n}}\left(\frac{B_t}{\sqrt{s}}+\frac{B_{t}^2}{s}\right)\right).
\end{equation}
By induction on (\ref{eq:thm1eq8}), we obtain that
\begin{equation}
\label{eq:thm1eq9}
(\hat\rho_1 - \hat r_{t+1})\leq (1-\nu\eta)^{t+1}[\hat\rho_1 - \hat r_{0}]_{+}+I_4+I_5+I_6,
\end{equation}
where 
\begin{align*}
&I_4 = \rho_1 \mO\left(\frac{ks\ln p}{n}\sum_{0\leq \ell\leq t}(1-\nu\eta)^{\ell}\right)=\rho_1 \mO\left(\frac{ks\ln p}{n\nu\eta}\right) = \rho_1 \mO\left(\frac{ks\ln p}{n}\right),\\
&I_5 = \rho_1\mO\left(\sqrt{\frac{s\ln p}{n}}\sum_{0\leq \ell\leq t}\frac{B_\ell}{\sqrt{s}}(1-\nu\eta)^{t-\ell}\right),\\
& I_6 = \rho_1\mO\left(\sqrt{\frac{s\ln p}{n}}\sum_{0\leq \ell\leq t}\frac{B^2_\ell}{s}(1-\nu\eta)^{t-\ell}\right).
\end{align*}
We next upper bound $I_5$ and $I_6$.  Recall that  $B_{\ell+1}\geq \frac{4c_{B_2}(1+c_{B_1})}{\nu}\sqrt{\frac{s^2\ln p}{n}}$ and $B_{\ell}-B_{\ell+1}\leq c_{B_2}\eta (1+c_{B_1})\sqrt{\frac{s^2\ln p}{n}}$. Consequently, for all $\ell \leq t$:
\begin{align}
&\frac{B_{\ell+1}}{B_{\ell}}\geq \frac{\frac{4c_{B_2}(1+c_{B_1})}{\nu}}{\frac{4c_{B_2}(1+c_{B_1})}{\nu}+c_{B_2}\eta(1+c_{B_1})}\geq 1-\frac{\eta\nu}{4},\label{eq:thm1eq10}\\
&\frac{B^2_{\ell+1}}{B^2_{\ell}}\geq \frac{\left(\frac{4c_{B_2}(1+c_{B_1})}{\nu}\right)^2}{\left(\frac{4c_{B_2}(1+c_{B_1})}{\nu}+c_{B_2}\eta(1+c_{B_1})\right)^2}=1-\frac{(\frac{\nu\eta}{4})^2+\frac{\eta\nu}{2}}{(1+\frac{\nu\eta}{4})^2}\geq 1-\frac{\nu\eta}{2}.\label{eq:thm1eq11}
\end{align}
Consequently, we can upper bound $\frac{B_{\ell}}{\sqrt{s}}(1-\nu\eta)^{t-\ell}$ by $\frac{B_{t+1}}{(1-\frac{\eta\nu}{4})\sqrt{s}}(1-\frac{3\nu\eta}{4})^{t-\ell}$, and upper bound $\frac{B_{\ell}^2}{s}(1-\nu\eta)^{t-\ell}$ by $\frac{B^2_{t+1}}{(1-\frac{\eta\nu}{2})\sqrt{s}}(1-\frac{\nu\eta}{2})^{t-\ell}$. In turn, we can bound $I_5$ and $I_6$ as below
\begin{align}
&I_5 = \rho_1\mO(\sqrt{\frac{s\ln p}{n}}\frac{B_{t+1}}{(1-\frac{\nu\eta}{4})\sqrt{s}}\sum_{0\leq \ell\leq t}(1-\frac{3\nu\eta}{4})^\ell)=  \rho_1\mO(\sqrt{\frac{s\ln p}{n}}\frac{B_{t+1}}{\sqrt{s}}),\label{eq:thm1eq12}\\
& I_6 = \rho_1\mO(\sqrt{\frac{s\ln p}{n}}\frac{B^2_{t+1}}{(1-\frac{\nu\eta}{2})\sqrt{s}}\sum_{0\leq \ell\leq t}(1-\frac{\nu\eta}{2})^\ell)=  \rho_1\mO(\sqrt{\frac{s\ln p}{n}}\frac{B_{t+1}^2}{s}).\label{eq:thm1eq13}
\end{align}
Combine (\ref{eq:thm1eq9}) with the bounds on $I_4$, $I_5$ and $I_6$, we obtain
\begin{equation}
\label{eq:thm1eq14}
(\hat\rho_1 - \hat r_{t+1})\leq (1-\nu\eta)^{t+1}[\hat\rho_1 - \hat r_0]_++\rho_1\mO\left(\frac{ks\ln p}{n}+\sqrt{\frac{s\ln p}{n}}(\frac{B_{t+1}}{\sqrt{s}}+\frac{B_{t+1}^2}{s})\right).
\end{equation}
By Lemma \ref{lem:coef_bound} (\ref{eq:coef_main1}), and (\ref{eq:thm1eq_prop3a}) and (\ref{eq:thm1eq_prop2a}), we can upper bound $[\hat\rho_1 - \hat r_0]_+$:
\begin{equation}
\label{eq:thm1eq15}
[\hat\rho_1 - \hat r_0]_+\leq \frac{2(k+1)}{k}\frac{M}{\mu_1(1-\frac{c}{2})^2}\delta_0+\rho_1\mO(\frac{ks\ln p}{n}+\sqrt{\frac{s\ln p}{n}}(\frac{B_0}{\sqrt{s}}+\frac{B_0^2}{s})).
\end{equation}
Combine (\ref{eq:thm1eq10}), (\ref{eq:thm1eq11}) and (\ref{eq:thm1eq15}),  and plug them into (\ref{eq:thm1eq14}), we obtain the bound (\ref{eq:thm1eq1}):
\[
(\hat\rho_1 - \hat r_{t+1})\leq (1-\nu\eta)^{t+1}\frac{2(k+1)M}{k(1-c)\mu_1}\rho_1\delta_0+\rho_1\mO\left(\frac{ks\ln p}{n}+\sqrt{\frac{s\ln p}{n}}(\frac{B_{t+1}}{\sqrt{s}}+\frac{B_{t+1}^2}{s})\right).
\]
Combine the last display with Lemma \ref{lem:coef_bound} (\ref{eq:coef_main2}) and lower bound $m_{1t}$ with (\ref{eq:thm1eq_prop2a}), we obtain the bound in (\ref{eq:thm1eq2}):
\[
\delta_{t+1}\leq  \frac{\beta_t^\top\Lambda\beta_t}{\mu_1}\frac{2(k+1)M^2}{(k-1)(1-c)\gamma}\delta_0(1-\nu\eta)^{t+1}+\mO\left(\frac{ks\ln p}{n}+\sqrt{\frac{s\ln p}{n}}(\frac{B_{t+1}}{\sqrt{s}}+\frac{B_{t+1}^2}{s})\right).
\]
\subsubsection*{Proof of $\delta_{T+1}\leq \delta_{T+1}^{upper}(c,k)$.}
Since $\delta_0 \leq \psi_1\delta_0^{upper}(c,k) = \frac{\gamma(1-c)^2(k-1)}{2M^4(k+1)}\delta_0^{upper}(c,k)$, by (\ref{eq:thm1eq2}), we have
\begin{equation}
\label{eq:thm1eq16}
\delta_{T+1}\leq (1-c)\delta_0^{upper}(c,k)(1-\nu\eta)^{T+1}+\mO\left(\frac{ks\ln p}{n}+\sqrt{\frac{s\ln p}{n}}\left(\frac{B_t}{\sqrt{s}}+\frac{B_t^2}{s}\right)\right).
\end{equation}
Since $\delta^{upper}_{t}(c, k)=  \min\{\frac{c^2}{8M^2}, \frac{cB_{t}}{\|\xi_1\|_1}, \frac{c\gamma\eta}{M}(\frac{k-1}{k})^2(1-c)^2\}$, and by (\ref{eq:thm1eq10}), (\ref{eq:thm1eq11}), we know  that
\[
\frac{cB_{t}}{\|\xi_1\|_1}<\frac{cB_{0}(1-\nu\eta)^{t}}{\|\xi_1\|_1},\; \forall 1\leq t \leq T^*.
\]
Consequently, we have $(1-c)\delta_0^{upper}(c,k)(1-\nu\eta)^{T+1}\leq (1-c)\delta^{upper}_{T+1}(c, k)$. In other words, if we can show that
\[
\mO\left(\frac{ks\ln p}{n}+\sqrt{\frac{s\ln p}{n}}\left(\frac{B_{T+1}}{\sqrt{s}}+\frac{B_{T+1}^2}{s}\right)\right)\leq c\delta^{upper}_{T+1}(c, k),
\]
we can conclude that $\delta_{T+1}\leq \delta_{T+1}^{upper}(c,k)$. Since $\frac{c^2}{8M^2}$, $ \frac{c\gamma\eta}{M}(\frac{k-1}{k})^2(1-c)^2$ are constants, we can always let $\iota$ be a sufficiently small constant, such that  for all $T < T^*$:
\begin{equation}
\label{eq:thm1eq17}
\mO\left(\frac{ks\ln p}{n}+\sqrt{\frac{s\ln p}{n}}\left(\frac{B_{T+1}}{\sqrt{s}}+\frac{B_{T+1}^2}{s}\right)\right)\leq c\min\{\frac{c^2}{8M^2},  \frac{c\gamma\eta}{M}(\frac{k-1}{k})^2(1-c)^2\}.
\end{equation}
Hence, it remains to show that for a sufficiently small positive constant $\iota$:
\[
\mO\left(\frac{ks\ln p}{n}+\sqrt{\frac{s\ln p}{n}}\left(\frac{B_{t+1}}{\sqrt{s}}+\frac{B_{t+1}^2}{s}\right)\right)\leq \frac{c^2B_{t+1}}{\|\xi_1\|_1}.
\]
Because $B_{T+1}\leq B_0\leq (1+c_{B_1})\sqrt{s}$, $B_{T+1}\geq \frac{4c_{B_2}(1+c_{B_1})}{\nu}\sqrt{\frac{s^2\ln p}{n}}$ and $\|\xi_1\|_1\leq \sqrt{s}$, we have
\begin{align*}
&\frac{ks\ln p}{n}+\sqrt{\frac{s\ln p}{n}}\left(\frac{B_{t+1}}{\sqrt{s}}+\frac{B_{t+1}^2}{s}\right)\\
\leq & \left(\frac{k\nu}{4c_{B_2}(1+c_{B_1})}+1+(1+c_{B_1})\right)\sqrt{\frac{s\ln p}{n}}\frac{B_{t+1}}{\|\xi_1\|_1}.
\end{align*}
As a result, for all $T < T^*$, we have
{\small
\begin{equation}
\label{eq:thm1eq18}
\mO\left(\frac{ks\ln p}{n}+\sqrt{\frac{s\ln p}{n}}\left(\frac{B_{T+1}}{\sqrt{s}}+\frac{B_{T+1}^2}{s}\right)\right)\leq  \frac{B_{t+1}}{\|\xi_1\|_1}\mO((1+c_{B_1}+\frac{1}{c_{B_2}})\sqrt{\frac{s\ln p}{n}})\leq c^2\frac{B_{T+1}}{\|\xi_1\|_1},
\end{equation}
}
when $\iota$ is a sufficiently small positive constant such that  $\mO((1+c_{B_1}+\frac{1}{c_{B_2}})\sqrt{\frac{s\ln p}{n}})\leq c^2$. Combine (\ref{eq:thm1eq17}) and (\ref{eq:thm1eq18}), we have $\delta_{T+1}\leq \delta_{T+1}^{upper}(c,k)$ for a sufficiently small positive  constant $\iota$.

\subsection{Proof of Lemma \ref{lem:mainLemma}}
We know that $t\geq T^*$ if  
\begin{align*}
&(L_0-L_{\infty})(1-c_{B_2}\eta\sqrt{\frac{s\ln p}{n}})^{t}+L_{\infty}\geq \|\xi_1\|_1+\frac{4c_{B_2}(1+c_{B_1})}{\nu}\sqrt{\frac{s^2\ln p}{n}}.
\end{align*}
Since $\|\xi_1\|_1\leq \sqrt{s}$ and $L_0 = (1+c_{B_1})\sqrt{s}$,  $L_{\infty}>0$, the above requirement is satisfied if
\begin{align}
\label{eq:mainLemma1}
&c_{B_1}\sqrt{s}(1-c_{B_2}\eta\sqrt{\frac{s\ln p}{n}})^{t}\geq \sqrt{s}(1-(1-c_{B_2}\eta\sqrt{\frac{s\ln p}{n}})^t)+\frac{4c_{B_2}(1+c_{B_1})}{\nu}\sqrt{\frac{s^2\ln p}{n}}\notag\\
\Leftarrow& c_{B_1}(1-tc_{B_2}\eta\sqrt{\frac{s\ln p}{n}})\geq tc_{B_2}\eta\sqrt{\frac{s^2\ln p}{n}}+\frac{c_{B_2}(1+c_{B_1})}{\nu}\sqrt{\frac{s^2\ln p}{n}}\notag\\
\Leftrightarrow& t\leq \frac{c_{B_1}}{2\eta c_{B_2}}\sqrt{\frac{n}{s\ln p}} - \frac{2(1+c_{B_1})}{\eta \nu}.
\end{align}
Hence, we must have $T^* \geq \lfloor \frac{c_{B_1}}{2\eta c_{B_2}}\sqrt{\frac{n}{s\ln p}} - \frac{2(1+c_{B_1})}{\eta \nu} \rfloor$. As a result, when $\frac{s\ln p}{n}\rightarrow \infty$ as $n\rightarrow\infty$, we will have
\begin{equation}
\label{eq:mainLemma2}
\frac{\delta_0(1-\nu\eta)^{T^*}}{\frac{s\ln p}{n}}\rightarrow 0.
\end{equation}
Also, because $B_{T^*}  - B_{T^*+1}\leq L_0\times c_{B_2}\eta\sqrt{\frac{s\ln p}{n}}=(1+c_{B_1})c_{B_2}\eta \sqrt{\frac{s^2\ln p}{n}}$ and $ B_{T^*+1} < \frac{4c_{B_2}(1+c_{B_1})}{\nu}\sqrt{\frac{s^2\ln p}{n}}$, we must have 
\begin{equation}
\label{eq:mainLemma3}
B_{T^*} \leq \left(\frac{4c_{B_2}(1+c_{B_1})}{\nu}+(1+c_{B_1})c_{B_2}\right)\sqrt{\frac{s^2\ln p}{n}}.
\end{equation}
Combine (\ref{eq:mainLemma2}) and (\ref{eq:mainLemma3})  with Theorem \ref{thm:thm1}, we obtain that
\[
\lim_{n\rightarrow\infty} \bP(\delta_{T^*}\leq \psi_2\frac{s\ln p}{n})=1.
\]
for some sufficiently large constant $\psi_2$.
\subsection{Proof of Theorem \ref{thm:thm2}}

\begin{proof}
Combine (\ref{eq:mainLemma2}) and (\ref{eq:mainLemma3}) in Lemma \ref{lem:mainLemma} with (\ref{eq:thm1eq1}) to lower bound $\hat r_{T^*}$ and with (\ref{eq:thm1eq2}) to upper  bound $\delta_{T^*}$ , with probability approaching 1,
\begin{align}
\hat r_{T^*}\geq \hat\rho_1 - \rho_1\mO(\frac{s\ln p}{n}),\; \delta_{T^*}\leq \mO(\frac{s\ln p}{n})\label{eq:thm2eq1},
\end{align}
Hence, for any $t$ with $\underline{f}_{\tau}(\beta_t)\geq \underline{f}_{\tau}(\beta_{T^*})$, we have:
\begin{align}
\label{eq:thm2eq2}
\hat\rho_1 - \hat r_t&\leq \rho_1\mO(\frac{s\ln p}{n})-\tau\sqrt{\frac{\ln p}{n}} \rho_1 (\|\beta_t\|_1+ \frac{c_2 \|\beta_t\|_1^2}{L_0}-\|\beta_{T^*}\|_1-\frac{c_2 \|\beta_{T^*}\|_1^2}{L_0})\notag\\
&=\rho_1\mO(\frac{s\ln p}{n})-\tau\sqrt{\frac{\ln p}{n}} \rho_1 \left[(1+\frac{2c_2\|\xi_1\|_1}{L_0})(\omega_t - \omega_{T^*})+ \frac{c_2(\omega_t^2-\omega_{T^*}^2)}{L_0}\right]
\end{align}
For any $t\geq 0$, write it as $\beta_t= \xi_1+h_t$.  By construction:
\begin{align}
\label{eq:thm2eq3}
&\omega_{t}\leq B_{t}.
\end{align}
By Proposition \ref{prop:prop4}, we also have
\begin{align}
\label{eq:thm2eq4}
&\omega_{t}\geq \|h_{t}\|_1-2\sqrt{2s\delta_{t}}.
\end{align}
Combine (\ref{eq:thm2eq3}), (\ref{eq:thm2eq4}),  we have 
\begin{align}
\label{eq:thm2eq5}
-2\sqrt{2s\delta_{T^*}}\leq \omega_{T^*}\leq B_{T^*}.
\end{align}
Combine (\ref{eq:thm2eq5}) with (\ref{eq:mainLemma3}) and (\ref{eq:thm2eq1}), we can upper bound  $\omega_{T^*}$ and $\omega_{T^*}^2$ as
\[
\omega_{T^*}\leq B_{T^*}, \; \omega_{T^*}^2\leq \max\{B_{T^*}^2, 8s\delta_t\}.
\]
Combine the last display with (\ref{eq:thm2eq1}) and (\ref{eq:mainLemma3}), we can bound $\hat\rho_1 - \hat r_t$ as (in probability),
\begin{align}
\label{eq:thm2eq6}
\hat\rho_1 - \hat r_t&\leq \rho_1\mO(\frac{s\ln p}{n})-\tau\sqrt{\frac{\ln p}{n}} \rho_1 \left[(1+\frac{2c_2\|\xi_1\|_1}{L_0})\omega_t + \frac{c_2\omega_t^2}{L_0}\right],
\end{align}
Combine (\ref{eq:thm2eq6}) with Lemma \ref{lem:coef_bound} (\ref{eq:coef_main2}):
{\small
\begin{align}
\frac{k-1}{k}\frac{\gamma m_{1t}}{\beta_t^\top\Lambda\beta_t}\delta_t\leq \mO(\frac{s\ln p}{n})+\mO\left(\sqrt{\frac{\ln p}{n}}\left(\bomega_t+\frac{\bomega_t^2}{\sqrt{k's}}\right)\right)-\tau\sqrt{\frac{\ln p}{n}}\left[(1+\frac{2c_2\|\xi_1\|_1}{L_0})\omega_t + \frac{c_2\omega_t^2}{L_0}\right],\notag
\end{align}
}
where we have absorbed $k$ into $\mO(.)$ since $(c, k)$ are constants depending only on $(M, \gamma, \eta, c_0)$. Notice that $L_0  = (1+c_{B_1})\sqrt{s}$, and  we can always take $k'$ to be large enough such that 
\[
\mO\left(\sqrt{\frac{\ln p}{n}}\frac{\bomega_t^2}{\sqrt{k's}}\right)\leq   \tau\sqrt{\frac{\ln p}{n}}\frac{c_2\bomega_t^2}{L_0}\leq \tau\sqrt{\frac{\ln p}{n}}\frac{c_2\omega_t^2}{L_0}.
\]
Hence, taking $k'$ to be a sufficiently large constant, we can bound $\delta_t$ as
\begin{align}
\label{eq:thm2eq7}
\frac{k-1}{k}\frac{\gamma m_{1t}}{\beta_t^\top\Lambda\beta_t}\delta_t\leq \mO(\frac{s\ln p}{n})+\tau(1+\frac{2c_2\|\xi_1\|_1}{L_0})\sqrt{\frac{s\ln p}{n}\delta_t}.
\end{align}
We take $\tau$ to be a sufficiently large constant such that
\begin{align}
\label{eq:thm2eq8}
\mO\left(\sqrt{\frac{\ln p}{n}}\bomega_t\right)\leq  \tau\sqrt{\frac{\ln p}{n}}  (1+\frac{2c_2\|\xi_1\|_1}{L_0})\bomega_t.
\end{align}
Now we consider two cases for (\ref{eq:thm2eq7}):
\begin{itemize}
\item When $\omega_t\leq 0$: we have $\omega_t\geq -2\sqrt{2s\delta_t}$ by (\ref{eq:thm2eq4}), $\bomega_t=0$, and $\frac{\|\xi_1\|_1}{L_0}\leq \frac{1}{1+c_{B_1}}$, hence,
\begin{align*}
\delta_t&\leq  \mO(\frac{s\ln p}{n}) +\tau(1+\frac{2c_2}{c_{B_1}+1})\sqrt{\frac{\ln p}{n}\times 8s\delta_t}\Rightarrow \delta_t\leq \mO(\frac{s\ln p}{n}).
\end{align*}
\item When $\omega_t\geq 0$, $\omega_t = \bomega_t\geq 0$. Combine (\ref{eq:thm2eq7}) and (\ref{eq:thm2eq8}), we immediately have $\delta_t\leq  \mO(\frac{s\ln p}{n})$.
\end{itemize}
Combine them together, when $\tau$ is a sufficiently large constant,  we have
$\lim_{n\rightarrow\infty}\bP(\delta_{t^*}\leq \psi_2\frac{s\ln p}{n}) = 1$ 
for a sufficiently large constant $\psi_2$.
\end{proof}
\subsection{Proof of Lemma \ref{lem:init}}
We prove Lemma \ref{lem:init} using arguments for proving Theorem 4.3 in \cite{gao2021sparse} and Lemma 12 from \cite{yuan2013truncated}. To use arguments for  \cite{gao2021sparse} Theorem 4.3, we define several additional notations. We define $V$ as the scaled $\xi$ with $V_j = \frac{\xi_j}{\sqrt{\xi_j^\top\Lambda\xi_j}}$ for $j=1,\ldots, p$. For any given rank $k$, let $V_{[k]}  = (V_1,\ldots, V_{k})$ and $V_{\bar{[k]}} = (V_{k+1},\ldots, V_p)$; let $\Gamma_{[k]} = \diag\{\rho_1,\ldots, \rho_k\}$. Define $\tilde V_{[k]} = V_{[k]}(V_{[k]}^\top\hat \Lambda V_{[k]})^{-\frac{1}{2}}$ and $\tilde \Gamma_{[k]} = (V_{[k]}^\top\hat \Lambda V_{[k]})^{\frac{1}{2}}\Gamma_{[k]}(V_{[k]}^\top\hat \Lambda V_{[k]})^{\frac{1}{2}}$. Define $\|H\|_{\infty,\infty}=\max_{j,\ell}|H_{j,\ell}|$ and  events
\begin{align*}
\mA_4 = &\left\{\|\hat\Sigma - \Sigma\|_{\infty,\infty}+\|\Lambda V_{[k]}\Gamma_{[k]} V_{[k]}^\top \Lambda -\hat\Lambda V_{[k]}\Gamma_{[k]} V_{[k]}^\top\hat\Lambda  \|_{\infty,\infty}\right.\\
&\left. \rho_{k+1}\|\hat\Lambda - \Lambda\|_{\infty,\infty}+\rho_{k+1}\|\Lambda V_{[k]} V_{[k]}^\top \Lambda -\hat\Lambda V_{[k]} V_{[k]}^\top\hat\Lambda  \|_{\infty,\infty}\leq CM\rho_1\sqrt{\frac{\ln p}{n}}\right\},\\
\mA_5 = &\left\{\|\tilde\Gamma_{[k]} - \Gamma_{[k]}\|_F+\rho_{k+1}\|V_{[k]}^\top\hat\Lambda V_{[k]} - \Id\|_F+\|\Lambda^{\frac{1}{2}}(\tilde V_{[k]} - V_{[k]})\|_{F}\leq C\rho_1\sqrt{\frac{k(s+\ln p)}{n}}\right\}.
\end{align*}
According to Proposition \ref{prop:prop6}, $\mA_4\cap \mA_5$ happens with high probability. Proofs to Proposition \ref{prop:prop6} are deferred to Appendix \ref{app:proof_supportings}.
\begin{proposition}
\label{prop:prop6}
Under Assumptions \ref{ass:regularity} and \ref{ass:sparsity}, $\bP(\mA_4\cap \mA_5)\rightarrow 1$ as $n\rightarrow \infty$ for a sufficiently large universal constant $C$.
\end{proposition}
On the event $\mA_4\cap \mA_5$ and when  $\frac{s^2\ln p}{n}\rightarrow 0$ and $\tau =CM\rho_1\sqrt{\frac{\ln p}{n}}$ for a sufficiently large universal constant $C$, following exactly the same arguments for Theorem 4.3 in \cite{gao2021sparse}, we have
\[
\|V_{[k]}V_{[k]}^\top - \hat P \|_F\leq  \|\tilde V_{[k]} \tilde V_{[k]}^\top - \hat P \|_F+ \|\tilde V_{[k]} \tilde V_{[k]}^\top - V_{[k]}V_{[k]}^\top  \|_F\leq \mO\left(\frac{ \rho_1}{\rho_{k}-\rho_{k+1}}s\sqrt{\frac{\ln p}{n}}+\sqrt{\frac{k(s+\ln p)}{n}}\right),
\]
We have left out the proofs here because the arguments are identical except for replacing the tail bound events $\mathcal{B}_3$ and $\mathcal{B}_4$ in \cite{gao2021sparse} with $\mA_4$ and $\mA_5$ in this paper, to account for potentially infinite $\rho_1$. 

When $k = 1$ and $\rho_1 - \rho_2 \geq \gamma \rho_1$, we obtain that,
\begin{equation}
\label{eq:initlem_eq1}
\|V_{1}V_{1}^\top - \hat P \|_F\leq \mO\left(s\sqrt{\frac{\ln p}{n}}\right).
\end{equation}
Apply the the Davis-Kahan Theorem as described in \cite{vu2013minimax}:
\[
\frac{1}{2\sqrt{2}}\|\xi_1 - \hat\beta\|_2\leq \frac{\|V_1V_1^\top - \hat P\|_F}{V_1^\top V_1}=\xi_1^\top\Lambda\xi_1 \|V_1V_1^\top - \hat P\|_F\leq \mO\left(s\sqrt{\frac{\ln p}{n}}\right).
\]
Hence, we have
\begin{equation}
\label{eq:initlem_eq2}
\delta(\hat\beta) = 1-|\hat\beta^\top\xi_1|=\frac{1}{2}\|\xi_1-\hat\beta\|_2^2 \leq \mO(\frac{s^2\ln p}{n}).
\end{equation}
As a last step, we turn to the truncated version of $\hat\beta$ and evaluate its quality using Lemma 12 from \cite{yuan2013truncated}.
\begin{proposition}[Lemma 12 from \cite{yuan2013truncated}]
\label{prop:other1}
For any unit vector $\beta$, let $F$ be the set of indices $j$ with the $ks$ largest  $|\beta_j|$ and let $\tilde\beta$ be the truncated version of $\beta$ with entries outside of  $F$ zeroed out. Then,
\[
|\tilde\beta^\top\xi_1|\geq |\beta^\top\xi_1| -(\frac{1}{k})^{\frac{1}{2}}\min\{\sqrt{1-(\beta^\top\xi_1)^2}, (1+(\frac{1}{k})^{\frac{1}{2}})(1-(\beta^\top\xi_1)^2)\}.
\]
\end{proposition}
Since $1-(\hat\beta^\top\xi_1)^2 = 1-(1-\delta(\hat\beta))^2 \leq 2\delta(\hat\beta)$, and $|\beta_0^\top\xi_1|\geq |\tilde\beta^\top\xi_1|$, by Proposition \ref{prop:other1}, on the event that $\delta(\hat\beta)=\mO(\frac{s^2\ln p}{n})$, we have
\[
\delta_0 \leq 1-|\tilde\beta^\top\xi_1| \leq \delta(\hat\beta)+4(\frac{1}{k})^{\frac{1}{2}}\delta(\hat\beta)=\mO(\frac{s^2\ln p}{n}).
\]
On the other hand, we have 
\[
\|\beta_0\|_1\leq \|\xi_1\|_1+\|\beta_0-\xi_1\|_1\leq \|\xi_1\|_1+\sqrt{(k+1)s}\|\beta_0-\xi_1\|_2=\|\xi_1\|_1+\sqrt{2(k+1)s\delta_0}.
\] 
This indicates that
\[
\frac{\|\beta_0\|_1-\sqrt{s}}{\sqrt{s}}\leq\frac{\|\beta_0\|_1-\|\xi_1\|_1}{\sqrt{s}}\leq \sqrt{2(k+1)\delta_0}=\mO(\sqrt{\frac{s^2\ln p}{n}})\rightarrow 0.
\]
Consequently, when $n$ is sufficiently large, with probability approaching 1, we have
\[
\delta_0 \leq \psi_1,\;\|\beta_0\|_1\leq (1+c_{B_1})\sqrt{s},
\]
for any positive constants $\psi_1$ and $c_{B_1}$.

\subsection{Proof of Proposition \ref{prop:proximal_l1_update0}}
Let $j_* = \lceil L_{t+1}^2\rceil$.   We consider the two cases separately.
\subsubsection*{When $|\theta|_{(1)} > |\theta|_{(j_*)}$:} We are equivalently consideirng the min max  problem of its Lagrangian:
\[
 \min_{\beta}\max_{\nu_1\in \real, \nu_2\geq 0}\mathcal{L}(\beta, \nu_1, \nu_2) = \min_{\beta}\max_{\nu_1\in \real, \nu_2\geq 0}\left(\|\beta-\theta\|_2^2+\nu_1(\|\beta\|_2^2-1)+\nu_2(\|\beta\|_1-L_{t+1})\right).
\]
As a result, if we can find $\nu_1^*\in \real$ and $\nu_2^*\geq 0$, such that its  associated minimizer $\beta^*$ for $\mathcal{L}(\beta, \nu_1^*, \nu_2^*)$   satisfies $\|\beta^*\|_2^2 = 1$, $\nu_{2}^*(\|\beta^*\|_1  - L_{t+1}) = 0$, we must have
\[
\|\beta^*-\theta\|_2^2=\min_{\beta}\mathcal{L}(\beta, \nu_1^*, \nu_2^*)\leq \min_{\beta}\max_{\nu_1\in \real, \nu_2\geq 0}L(\beta, \nu_1, \nu_2)= \min_{\|\beta\|_2=1, \|\beta\|_1\leq L_{t+1}}\|\beta-\theta\|_2^2.
\]
As a result, $\beta^*$ is a minimizer to the original problem.  Set $\zeta(c) = \frac{\|[|\theta|-c]_+\|_1}{\|[|\theta|-c]_+\|_2}$. When $|\theta|_{(1)} > |\theta|_{(j_*)}$, we have $\zeta(|\theta|_{(j_*)})\leq \sqrt{\lceil L_{t+1}^2\rceil-1}\leq L_{t+1}$. 
\begin{enumerate}
\item  If $\zeta(0) \leq  L_{t+1}$, we can simply let $\nu_2^* = 0$ and $\nu_1^* = \|\theta\|_2-1$, and $\beta^* = \frac{\theta}{\|\theta\|_2}$.
\item If $\zeta(0) >  L_{t+1}$,  let $c$ be the smallest positive value such that $\zeta(c)\leq L_{t+1}$.  Then,  $\zeta(c) = L_{t+1}$ by the continuity of $\zeta(c)$ from Proposition \ref{prop:proximal_l1_update}. Let $\nu_2^* = 2c$, $\nu_1^* = \|[|\theta| - c]_+\|_2-1$, we have $\beta^* = \frac{\sign(\theta)\cdot [|\theta| - c]_+}{\|[|\theta| - c]_+\|_2}$ and the optimal conditions are satisfied.
\end{enumerate}
We hence proved part (1) of Proposition \ref{prop:proximal_l1_update0}.
\subsubsection*{When $|\theta|_{(1)} = |\theta|_{(j_*)}$:} 
Since the proximal problem can be equivalently expressed as $\max_{\|\beta\|_2^2\leq 1, \|\beta\|_1\leq L_{t+1}} \beta^\top\theta$, and 
\[
\max_{\|\beta\|_2^2\leq 1, \|\beta\|_1\leq L_{t+1}} \beta^\top\theta \leq |\theta|_{(1)}\max_{\|\beta\|_2^2\leq 1, \|\beta\|_1\leq L_{t+1}}\|\beta\|_1 \leq L_{t+1}c_0.
 \]
 Let $F = \{j: |\theta_j|\geq |\theta|_{(1)}\}$ be the index subset that contains all entries in $|\theta|$ with the largest magnitude, and let $\bar{F}$ be the complement of $F$. Then, $\beta_{t+1} = \sign(\theta)\cdot \tilde\beta$  achieves this optimal objective, with any $\tilde\beta$ such that $\tilde\beta\geq 0$, $\tilde\beta_{\bar{F}} = 0$ and $\|\tilde\beta_{F}\|_2^2=1$, $\|\tilde\beta_{F}\|_1 = L_{t+1}$. 

\subsection{Proof of Proposition \ref{prop:proximal_l1_update}}
We set $\zeta(c) = \frac{\|[|\theta|-c]_+\|_1}{\|[|\theta|-c]_+\|_2}$. When $\max_j |\theta_j| > \min_j |\theta_j|$ and $0\leq c < \max_j|\theta_j|$, the denominator of $\zeta(c)$ is positive. Without loss of generality, $u_{j} = |\theta_j|$ is arranged from large to small, e.g., $u_j = |\theta|_{(j)}$.  We can write out the sub-gradient of $\zeta(c)$ with respect to $c$. For any $u_j\geq 0$, we know that
\begin{equation}
\label{prox2_eq1}
\frac{\partial [u_j - c]_+}{\partial c} = \left\{\begin{array}{lll} -1 &\mbox{if }& u_j > c,\\
0&\mbox{if }& u_j < c,\\
q_j\in [-1, 0]&\mbox{if }& u_j =c.
\end{array}\right.
\end{equation}
Let $n_c^+$ be number of $u_j$ with $u_j > c$. Based on (\ref{prox2_eq1}), we can calculate the sub-gradient for $\zeta(c)$,
\begin{align}
\label{prox2_eq2}
\frac{\partial \zeta(c)}{\partial c} &= \frac{\|[u-c]_+\|_2\frac{\partial \sum_{j=1}^p [u_j-c]_+}{\partial c}-\frac{\sum_{j=1}^p [u_j-c]_+\frac{\partial [u_j-c]_+}{\partial c}}{\|[u-c]_+\|_2}\|[u-c]_+\|_1}{\|[u-c]_+\|_2^2}\notag\\
&= \frac{-\|[u-c]_+\|_2^2 n_{c}^++\|[u-c]_+\|_2^2 \sum_{j:u_j = c}q_j+\|[u-c]_+\|_1\sum_{j:u_j > c}[u_j-c]}{\|[u-c]_+\|_2^3}.
\end{align}
Hence, when $ c < u_1$, $\frac{\partial \delta(c)}{\partial c}$ exists and is bounded, hence, $\zeta(c)$ is continuous in $c$. On the other hand, since $q_j\leq 0$, we have
\begin{align}
\label{prox2_eq3}
\frac{\partial \zeta(c)}{\partial c} &\leq  \frac{-\|[u-c]_+\|_2^2 n_{c}^++(\sum_{j:u_j > c}[u_j-c])^2}{\|[u-c]_+\|_2^3}.
\end{align}
By the Cauchy–Schwarz inequality, we know
\begin{align}
\|[u-c]_+\|_2^2 n_{c}^+&=(\sum_{j:u_j > c} (u_j - c)^2)(\sum_{j: u_j > c} 1)\geq  (\sum_{j:u_j > c} (u_j - c))^2. \label{prox2_eq4}
\end{align}
Combine (\ref{prox2_eq3}) and (\ref{prox2_eq4}), we obtain that $\frac{\partial \zeta(c)}{\partial c} \leq  0$ and $\zeta(c)$ is non-increasing in $c$.
\subsection{Proof of Proposition \ref{prop:deflate1}}
\label{app:proof_deflate1}
We prove the statement by induction. Let $\hat\beta_1,\ldots,\hat\beta_p$ be solutions to (\ref{eq:deflate1}) for all $p$ directions. Set $\tilde \Sigma_{k} = \frac{\tilde X_{k}^\top\tilde X_{k}}{n}$. Suppose that for all $k\leq K$:

\noindent (a)  (\ref{eq:deflate1}) and (\ref{eq:deflate2}) are equivalent.

\noindent (b) $\tilde\Sigma_k \hat\beta_j= 0$ for all $j < k$.

\noindent (c) $\tilde\Sigma_k\hat\beta_j = \hat\Sigma \hat\beta_j$ for all $j \geq  k$.

Then,  we show that  relationships (a) - (c) hold  for  $k = K+1$. From Proposition \ref{prop:deflate2}, we have
\begin{equation}
\label{eq:prop_deflate1_eq1}
\tilde\Sigma_{K+1}\hat\beta_j = \tilde\Sigma_{K}\hat\beta_j -\tilde\Sigma_{K}\hat\beta_{K} \frac{\hat\beta_{K}^\top\tilde\Sigma_{K}\hat\beta_j}{\hat\beta_K^\top\tilde\Sigma_{K}\hat\beta_K}, \; \mbox{for all } j = 1,\ldots, p.
\end{equation}
\begin{itemize}
\item Combine (\ref{eq:prop_deflate1_eq1}) with relationship (b), we obtain $\hat\tilde\Sigma_{K+1}\beta_j = 0$ for all $j< K$. From (\ref{eq:prop_deflate1_eq1}), we also have
\[
\tilde\Sigma_{K+1} \hat\beta_K = \tilde\Sigma_{K}\hat\beta_K - \tilde\Sigma_{K}\hat\beta_{K} \frac{\hat\beta_{K}^\top\tilde\Sigma_{K}\hat\beta_K}{\hat\beta_K^\top\tilde\Sigma_{K}\hat\beta_K} = 0.
\]
 Hence, relationship (b) holds for $k = K+1$.
\item By relationship  (c) and the fact that $\hat\beta_{j_1}^\top \hat\Sigma\hat\beta_{j_2} = \hat\rho_{j_1}\1_{j_1=j_2}$ for all $1\leq j_1, j_2\leq p$,  we obtain that for all $j \geq K+1$:
\begin{align}
&\tilde\Sigma_K\hat\beta_j = \hat\Sigma\hat\beta_j,\;\tilde\Sigma_{K}\hat\beta_{K} \frac{\hat\beta_{K}^\top\tilde\Sigma_{K}\hat\beta_j}{\hat\beta_K^\top\tilde\Sigma_{K}\hat\beta_K}=\tilde\Sigma_{K}\hat\beta_{K} \frac{\hat\beta_{K}^\top\hat\Sigma\hat\beta_j}{\hat\beta_K^\top\tilde\Sigma_{K}\hat\beta_K} = 0 . \label{eq:prop_deflate1_eq2}
\end{align}
Thus, we have $\tilde\Sigma_{K+1}\hat\beta_j = \hat\Sigma\hat\beta_j$ for all $j\geq K+1$ and relationship (c) holds for $k= K+1$.
\item  From (\ref{eq:prop_deflate1_eq2}), we immediately obtain that $ \hat\beta_{j_1}^\top\tilde\Sigma_{K+1} \hat\beta_{j_2}= \hat\beta_{j_1}^\top\hat\Sigma \hat\beta_{j_2}=\hat\rho_{j_1}\1_{j_1=j_2}$ for all $j_1, j_2\geq K+1$. Hence, the generalized eigenvalue and eigenvector pairs for $\tilde\Sigma_{K+1}$ is $(\tilde\rho_j, \tilde\beta_j)$ (unordered) where $\tilde\beta_j = \hat\beta_j$ and $\tilde\rho_j = 0$ for $j\leq K$ and $\tilde \rho_j = \hat\rho_j$ for $j\geq K+1$. The leading eigenvector pair is $(\hat\rho_{K+1}, \hat\beta_{K+1})$ and the relationship (a) holds. 
\end{itemize}
By induction, relationships (a)-(c) hold for all $k\leq p$, and we have proved our statement.

\subsection{Proof of Proposition \ref{prop:deflate2}}
\label{app:proof_deflate2}
We prove the statement by induction. Suppose that for all $k\leq K$, we have $\frac{\tilde X_k^\top\tilde X_k}{n} = \tilde\Sigma_k$ (holds obviously for $K = 1$). Then, we show that  $\frac{\tilde X_{K+1}^\top\tilde X_{K+1}}{n} = \tilde\Sigma_{K+1}$:
\begin{align*}
\frac{\tilde X_{K+1}^\top\tilde X_{K+1}}{n} &= \frac{\tilde X_{K}^\top(\Id - \frac{\tilde Z_k \tilde Z_k^\top}{\|\tilde Z_j\|_2^2})(\Id - \frac{\tilde Z_k \tilde Z_k^\top}{\|\tilde Z_j\|_2^2})\tilde X_{K}}{n} \\
& = \tilde \Sigma_K - 2\frac{\tilde X_K^\top \tilde X_K\hat\beta_K\hat\beta_K^\top \tilde X_K^\top\tilde X_K}{n\hat\beta_K^\top\tilde X_K^\top \tilde X_K\hat\beta_K}+\frac{\tilde X_K^\top \tilde X_K\hat\beta_K\hat\beta_K^\top \tilde X_K^\top  \tilde X_K\hat\beta_K\hat\beta_K^\top \tilde X_K^\top \tilde X_K}{n(\hat\beta_K^\top\tilde X_K^\top \tilde X_K\hat\beta_K)^2}\\
& = \tilde \Sigma_K - \frac{\tilde\Sigma_K\hat\beta_K\hat\beta_K^\top \tilde\Sigma_K}{\hat\beta_K^\top\tilde\Sigma_K\hat\beta_K} = \tilde\Sigma_{K+1}.
\end{align*}
Hence, we have $\frac{\tilde X_k^\top\tilde X_k}{n}  =\tilde\Sigma_k$ for all $1\leq k\leq p$.
\section{Proofs of Technical Lemmas}
\label{app:proof_technical}
\subsection{Proof of Lemma \ref{lem:coef_bound}}
\begin{proof}[Proof of Lemma \ref{lem:coef_bound}]
In this proof, we drop the subscript $t$ in the decomposition $\beta_t = \sum_{j=1}^p\alpha_{jt}\xi_j$ and denote $\alpha_{jt}$ as $\alpha_j$ for convenience. Define    $O_H= \frac{\sum_{j\geq 2}\alpha_j^2\mu_j(\rho_1 - \rho_j)}{\beta_t^\top  H \beta_t}$  for $H\in \{\Lambda, \Sigma\}$.  By Lemma \ref{lem:basic_bound2} (\ref{basic2main2}):
{\small
\begin{align}
T_{\Sigma}\coloneqq&|\frac{\hat\rho_1 - \hat r_t}{\hat r_t}-O_{\Sigma}|= |\frac{\hat\rho_1\beta_t^\top\hat\Lambda\beta_t - \beta_t^\top\hat\Sigma\beta_t}{\beta_t^\top\hat\Sigma\beta_t}-O_{\Sigma}|=|\frac{\hat\rho_1 F_{\Lambda} - F_{\Sigma}+\hat\rho_1 W_{\Lambda}-W_{\Sigma}}{\beta_t^\top\hat\Sigma\beta_t}-O_{\Sigma}|.\label{eq:coef_bound1}\\
T_{\Lambda}\coloneqq&|(\hat\rho_1 - \hat r_t)-O_{\Lambda}|= |\frac{\hat\rho_1\beta_t^\top\hat\Lambda\beta_t - \beta_t^\top\hat\Sigma\beta_t}{\beta_t^\top\hat\Lambda\beta_t}-O_{\Lambda}|=|\frac{\hat\rho_1 F_{\Lambda} - F_{\Sigma}+\hat\rho_1 W_{\Lambda}-W_{\Sigma}}{\beta_t^\top\hat\Lambda\beta_t}-O_{\Lambda}|.\label{eq:coef_bound2}
\end{align}
}
where $|W_{\Lambda}|\leq CMW_0$, $|W_{\Sigma}|\leq CM\rho_1W_0$ and $W_0 = \sqrt{\frac{s\ln p}{n}}(\sqrt{\delta_t}+\sqrt{k'}\delta_t+\frac{\bomega_t}{\sqrt{s}}+\frac{\bomega_t^2}{\sqrt{k'}s})$. For $H\in \{\Lambda,\Sigma\}$, we can upper bound $T_{H}$:
\begin{align}
\label{eq:coef_bound3}
T_H& =|\frac{\sum_{j\geq 2}\alpha_j^2\mu_j(\hat\rho_1 - \rho_j)+\hat\rho_1 W_{\Lambda}-W_{\Sigma}}{\beta_t^\top\hat H\beta_t} - O_{H}|\leq W_{1}+W_{2}+W_{3},
\end{align}
where $W_{1}=\frac{(\hat\rho_1+\rho_1)CMW_0}{\beta_t^\top\hat H\beta_t}$, $W_{2} = \frac{\sum_{j\geq 2}\alpha_j^2\mu_j|\hat\rho_1 - \rho_1|}{\beta_t^\top\hat H\beta_t}$ and $W_{3} = |\frac{1}{\beta_t^\top\hat H\beta_t} -\frac{1}{\beta_t^\top H\beta_t} |\sum_{j\geq 2}\alpha_j^2\mu_j(\rho_1-\rho_j)$.  By Proposition \ref{prop:prop3}, we know that
\begin{equation}
\label{eq:coef_extra1}
\beta_t^\top \Sigma\beta_t=\rho_1 m_{2t}\geq \rho_1 \beta_t^\top\Lambda\beta_t(1-\frac{c^2}{2}),\;\mbox{when } \delta_t\leq \frac{c^2}{8M^2}
\end{equation}
From (\ref{eq:coef_extra1}), we have
\begin{align}
\beta_t^\top H\beta_t - C\lambda_{H}^{\max}(1+c_{B_1})^2\sqrt{\frac{s\ln p}{n}}&\leq \beta_t^\top H\beta_t\left[1-C\frac{\lambda_{H}^{\max}}{\beta_t^\top H\beta_t}(1+c_{B_1})^2\sqrt{\frac{s\ln p}{n}}\right]\notag\\
\leq& \left\{\begin{array}{ll} \beta_t^\top H\beta_t(1-CM^2(1+c_{B_1})^2\sqrt{\frac{s\ln p}{n}}),&\mbox{if} \;H = \Lambda,\notag\\
\beta_t^\top H\beta_t(1-\frac{CM^2}{1-\frac{c^2}{2}}(1+c_{B_1})^2\sqrt{\frac{s\ln p}{n}}),&\mbox{if} \;H = \Sigma, \;\delta_t\leq \frac{c^2}{8M^2}
\end{array}\right.\notag\\
\overset{(a_1)}{\leq}&\frac{1}{2}\beta_t^\top H\beta_t,\; \mbox{if} \;H = \Lambda\;\mbox{or}\; H = \Sigma, \;\delta_t\leq \frac{c^2}{8M^2}.\label{eq:coef_bound3_1}
\end{align}
Step $(a_1)$ holds when $\sqrt{\frac{s\ln p}{n}}\leq \frac{(1-\frac{c^2}{2})}{2CM^2}$.  Combine the last display with  Lemma \ref{lem:basic_bound3} (\ref{basic3main1}) -(\ref{basic3main2}), we obtain that
{\small
\begin{align}
W_{1}&\leq \frac{\rho_1(2+\frac{C}{M}\sqrt{\frac{\ln p}{n}})CMW_0}{\beta_t^\top H \beta_t - C\lambda_{H}^{\max}(1+c_{B_1})^2\sqrt{\frac{s\ln p}{n}}}=\frac{\rho_1}{\beta_t^\top H\beta_t}\mO(\sqrt{\frac{s\ln p}{n}}(\sqrt{\delta_t}+\sqrt{k'}\delta_t+\frac{\bomega_t}{\sqrt{s}}+\frac{\bomega_t^2}{\sqrt{k'}s})),\label{eq:coef_bound4}\\
W_{2}&\leq  \frac{\frac{C}{M}\rho_1\sqrt{\frac{\ln p}{n}}}{\beta_t^\top H\beta_t-C\lambda_{H}^{\max}(1+c_{B_1})^2\sqrt{\frac{s\ln p}{n}}}m_{1t}\delta_t\overset{(a_2)}{<} \frac{\gamma\rho_1m_{1t}}{6k(\beta_t^\top H\beta_t)}\delta_t,\label{eq:coef_bound5}\\
W_{3}&\leq\frac{ C\lambda_{H}^{\max}(1+c_{B_1})^2\sqrt{\frac{s\ln p}{n}}}{\beta_t^\top H\beta_t(\beta_t^\top H\beta_t- C\lambda_{H}^{\max}(1+c_{B_1})^2\sqrt{\frac{s\ln p}{n}})} \rho_1 m_{1t}\delta_t\overset{(a_3)}{<} \frac{\gamma\rho_1m_{1t}}{6k(\beta_t^\top H\beta_t)} \delta_t.\label{eq:coef_bound6}
\end{align}
}
We have required $\sqrt{\frac{s\ln p}{n}} \leq \frac{\gamma }{15kCM^2(1+c_{B_1})^2}$ at steps $(a_2), (a_3)$, and used (\ref{eq:coef_bound3_1}) and the facts that $m_{1t}\geq \frac{1}{M}$ from Proposition \ref{prop:prop2}.

Since $\sqrt{\frac{s\ln p}{n}\delta_t}\leq \frac{1}{4b}\delta_t+b\frac{s\ln p}{n}$ for any positive value $b$,  we can take $b = \mO(1)$ to be a sufficiently large constant   such that
\[
\mO(\sqrt{\frac{s\ln p}{n}}\delta_t)\leq \frac{\gamma m_{1t}}{3k}\delta_t + \mO(\frac{ks\ln p}{n}).
\]
Hence, we can upper bound $W_1$ as 
{\small
\begin{align}
\label{eq:coef_bound7}
W_1 &\leq \frac{\rho_1}{\beta_t^\top H\beta_t}\left[\left(\mO(\sqrt{\frac{k's\ln p}{n}})+\frac{\gamma m_{1t}\rho_1}{3k}\right)\delta_t+ \mO(\frac{ks\ln p}{n}+\sqrt{\frac{s\ln p}{n}}\left(\frac{\bomega_t}{\sqrt{s}}+\frac{\bomega_t^2}{\sqrt{k'}s}\right)\right].
\end{align}
}
Plug the upper bounds of $W_1$, $W_2$, $W_3$ from (\ref{eq:coef_bound7}),  (\ref{eq:coef_bound5}), (\ref{eq:coef_bound6}) into   (\ref{eq:coef_bound3}), we obtain that
\begin{align}
\label{eq:coef_bound8}
T_H &\leq \frac{\rho_1}{\beta_t^\top H\beta_t}\left[\left( \frac{2\gamma m_{1t}}{3k}+\mO\left(\sqrt{\frac{k's\ln p}{n}}\right)\right)\delta_t+\mO\left(\frac{ks\ln p}{n}+\sqrt{\frac{s\ln p}{n}}\left(\frac{\bomega_t}{\sqrt{s}}+\frac{\bomega_t^2}{\sqrt{k'}s}\right)\right)\right]\notag\\
&\overset{(a_4)}{\leq}   \frac{\rho_1}{\beta_t^\top H\beta_t}\left[ \frac{\gamma m_{1t}}{k}\delta_t+\mO\left(\frac{ks\ln p}{n}+\sqrt{\frac{s\ln p}{n}}\left(\frac{\bomega_t}{\sqrt{s}}+\frac{\bomega_t^2}{\sqrt{k'}s}\right)\right)\right].
\end{align}
Step $(a_4)$ holds when $\iota$ is sufficiently small such that  $\mO\left(\sqrt{\frac{k's\ln p}{n}}\right)\leq \frac{\gamma }{3kM}\leq \frac{\gamma m_{1t}}{3k}$ for all $t$. (recall that $m_{1t}\geq \frac{1}{M}$.)

We  bound $O_H$ using  $\gamma\rho_1\leq \rho_1 - \rho_j\leq \rho_1$ for $j\geq 2$:
\begin{equation}
\label{eq:coef_bound9}
\frac{\gamma\rho_1 m_{1t}\delta_t}{\beta_t^\top H\beta_t}\leq O_{H}\leq \frac{\rho_1 m_{1t}\delta_t}{\beta_t^\top H\beta_t}.
\end{equation}
Combine (\ref{eq:coef_bound8}) and (\ref{eq:coef_bound9}), we will reach the desired bounds: 
\begin{itemize}
\item  When $H = \Lambda$, we get an upper bound on $\frac{[\hat\rho_1 -\hat r_t]_+}{\rho_1}$ combining (\ref{eq:coef_bound8}) with the upper bound in  (\ref{eq:coef_bound8}) and  the non-negativity of $\delta_t$;   we get a lower bound on $\frac{\hat\rho_1 -\hat r_t}{\rho_1}$ using   (\ref{eq:coef_bound8}) and the lower bound in (\ref{eq:coef_bound9}):
\begin{align*}
&\frac{[\hat\rho_1 - \hat r_t]_+}{\rho_1}\leq  \frac{k+1}{k}\frac{ m_{1t}\delta_t}{\beta_t^\top\Lambda\beta_t}+\mO\left(\frac{ks\ln p}{n}+\sqrt{\frac{s\ln p}{n}}\left(\frac{\bomega_t}{\sqrt{s}}+\frac{\bomega_t^2}{\sqrt{k'}s}\right)\right),\\
&\frac{\hat\rho_1 - \hat r_t}{\rho_1}\geq  \frac{k-1}{k}\frac{\gamma m_{1t}\delta_t}{\beta_t^\top\Lambda\beta_t}-\mO\left(\frac{ks\ln p}{n}+\sqrt{\frac{s\ln p}{n}}\left(\frac{\bomega_t}{\sqrt{s}}+\frac{\bomega_t^2}{\sqrt{k'}s}\right)\right).
\end{align*}
We have proved (\ref{eq:coef_main1}) and (\ref{eq:coef_main2}).
\item When $H = \Sigma$ and $\delta_t\leq \frac{c^2}{8M^2}$, we have $\beta_t^\top\Sigma\beta_t = m_{2t}\rho_1$ and $m_{2t}\geq \frac{(1-\frac{c}{2})^2}{\mu_1}\geq \frac{1-c}{M}$ by Proposition \ref{prop:prop3}. We get an upper bound on $\frac{[\hat\rho_1 -\hat r_t]_+}{\hat r_t}$ combining (\ref{eq:coef_bound8}) with the upper bound in  (\ref{eq:coef_bound8}) and  the non-negativity of $\delta_t$;   we get a lower bound on $\frac{\hat\rho_1 -\hat r_t}{\hat r_t}$ using   (\ref{eq:coef_bound8}) and the lower bound in (\ref{eq:coef_bound9}):
\begin{align*}
&\frac{[\hat\rho_1 - \hat r_t]_+}{\hat r_t}\leq  \frac{k+1}{k}\frac{ m_{1t}\delta_t}{m_{2t}}+\mO\left(\frac{ks\ln p}{n}+\sqrt{\frac{s\ln p}{n}}\left(\frac{\bomega_t}{\sqrt{s}}+\frac{\bomega_t^2}{\sqrt{k'}s}\right)\right),\\
&\frac{\hat\rho_1 - \hat r_t}{\hat r_t}\geq  \frac{k-1}{k}\frac{\gamma m_{1t}\delta_t}{m_{2t}}-\mO\left(\frac{ks\ln p}{n}+\sqrt{\frac{s\ln p}{n}}\left(\frac{\bomega_t}{\sqrt{s}}+\frac{\bomega_t^2}{\sqrt{k'}s}\right)\right).
\end{align*}
We have proved (\ref{eq:coef_main3}) and (\ref{eq:coef_main4}). 
\end{itemize}
\end{proof}

\subsection{Proof of Lemma \ref{lem:proximal}}
\begin{proof}
We lower bound the proximal objective improvement for the optimal solution by finding a feasible solution $\tilde\beta$ that leads to sufficient improvement. We consider the following construction $\tilde\beta = (1-\epsilon)\beta_t+\epsilon_1\xi_1$. Both $\epsilon$ and $\epsilon_1$ are non-negative and $\epsilon_1$ is a function of $\epsilon$ such that $\|\tilde\beta\|_2^2=1$ for any given $\epsilon$. For a given $\epsilon$, we can upper and lower bound $\epsilon_1$ using polynomials of $\epsilon$ and $\delta_t$ according to Proposition \ref{prop:prop5}:
\begin{proposition}
\label{prop:prop5}
For any unit vector $\beta$ with $\delta(\beta) = 1-\xi_1^\top\beta$ and $\xi_1^\top\beta\geq 0$. Define $\tilde\beta = (1-\epsilon)\beta+\epsilon_1\xi_1$ for some $\epsilon, \epsilon_1\in [0,1]$ chosen such that $\|\tilde\beta\|_2^2=1$. Then,
\[
-\frac{\epsilon(1-\epsilon)}{2}\delta(\beta)^2\leq \epsilon_1 - \epsilon -\epsilon(1-\epsilon)\delta(\beta)\leq \epsilon\delta(\beta)^2+3\epsilon^2\delta(\beta)^2+\frac{1}{2}\delta(\beta)^3.
\]
\end{proposition}
The corresponding proximal improvement is measured by $\Delta_t(\tilde\beta) = \tilde\beta^\top \theta- \beta_t^\top\theta$ with $\theta$ defined in (\ref{eq:theta}). Notice that $\beta_t^\top(\hat\Sigma - \hat r_t\hat\Lambda)\beta_t = 0$ by definition, thus,
\begin{equation}
\label{eq:proximal_eq1}
\Delta_t(\tilde\beta) = -\epsilon+\epsilon_1(1-\delta_t)+\epsilon_1\frac{\eta}{\hat r_t}\underbrace{\xi_1^\top(\hat\Sigma-\hat r_t\hat\Lambda)\beta_t}_{I_1}.
\end{equation}
Drop the subscript $t$ in the decomposition of $\beta_t$ and let $\beta_t = \sum_{j}\alpha_j \xi_j$ for convenience. Then,
\begin{equation}
\label{eq:proximal_eq2}
I_1 = \alpha_1\xi_1^\top(\hat\Sigma - \hat r_t\hat\Lambda)\xi_1+\underbrace{\xi_1^\top(\hat\Sigma - \hat r_t \hat\Lambda)(\beta_t - \alpha_1\xi_1)}_{I_2}=\underbrace{\alpha_1(\hat\rho_1 - \hat r_t)\hat\mu_1}_{I_3}+I_2.
\end{equation}
By  Proposition \ref{prop:prop2}, we know 
\begin{equation}
\label{eq:propDelta_extra1}
\frac{2}{M}(1-\frac{c^2}{16M^2})\leq m_{1t}\leq 2M,\;\mbox{when } \delta_t\leq \frac{c^2}{8M}.
\end{equation}
Combine (\ref{eq:propDelta_extra1}) with Lemma \ref{lem:coef_bound} (\ref{eq:coef_main1}) and Lemma \ref{lem:basic_bound3},  we can lower bound $\hat r_t$ as 
{\small
\begin{align}
\label{eq:proximal_eq3}
\hat r_t&\geq \rho_1\left(1-\mO(\sqrt{\frac{\ln p}{n}})\right)-\rho_1\frac{k+1}{k}\frac{m_{1t}}{\beta_t^\top\Lambda\beta_t}\delta_t-\rho_1\mO\left(\frac{ks\ln p}{n}+\sqrt{\frac{s\ln p}{n}}\left(\frac{\bomega_t}{\sqrt{s}}+\frac{\bomega_t^2}{s}\right)\right)\notag\\
&\overset{(a_1)}{\geq} \rho_1\left(1-\frac{k+1}{k}\frac{m_{1t}}{\beta_t^\top\Lambda\beta_t}\delta_t-\mO\left(\sqrt{\frac{s\ln p}{n}}(1+c_{B_1})^2+\frac{ks\ln p}{n}\right)\right)\overset{(a_2)}{\geq} \frac{1}{2}\rho_1\Rightarrow \frac{\rho_1}{\hat r_t}\leq 2.
\end{align}
}
At step $(a_1)$, we have upper bounded $\bomega_t$ as $\bomega_t\leq B_0 \leq (1+c_{B_1})\sqrt{s}$. At step $(a_2)$, we have upper bounded $\delta_t$ by $\frac{c^2}{8M^2}$ and $m_{1t}$ by $2M$ from (\ref{eq:propDelta_extra1}) and $\frac{k+1}{k}$ by $\frac{3}{2}$. Then, $(a_2)$ holds when 
\[
\mO\left(\sqrt{\frac{s\ln p}{n}}(1+c_{B_1})^2+\frac{ks\ln p}{n}\right)\leq \frac{3}{8}.
\]
We next upper bound $|I_2|$.  Combine (\ref{eq:proximal_eq3}) with Lemma \ref{lem:basic_bound2} (\ref{basic2main1}), we have
\begin{equation}
\label{eq:proximal_eq4}
|I_2|\leq \hat r_t\mO\left(\sqrt{\frac{s\ln p}{n}}(\sqrt{\delta_t}+\frac{\bomega_t}{\sqrt{s}})\right).
\end{equation}
By Lemma \ref{lem:coef_bound} (\ref{eq:coef_main4}), we can lower bound $I_3$ as 
\begin{equation}
\label{eq:proximal_eq5}
I_3\geq \alpha_1\frac{k-1}{k}\frac{\gamma m_{1t}}{m_{2t}}\hat\mu_1\hat r_t\delta_t-\alpha_1\hat r_t\hat\mu_1\mO\left(\frac{ks\ln p}{n}+\sqrt{\frac{s\ln p}{n}}\left(\frac{\bomega_t}{\sqrt{s}}+\frac{\bomega_t^2}{s}\right)\right).
\end{equation}
Plug in the bounds on $I_2$, $I_3$ from (\ref{eq:proximal_eq4}) and (\ref{eq:proximal_eq5}) into the expression of $I_1$ in (\ref{eq:proximal_eq2}):
\begin{align}
\label{eq:proximal_eq6}
I_1&\geq \hat r_t\left[\alpha_1\hat\mu_1\frac{k-1}{k}\frac{\gamma m_{1t}}{m_{2t}}\delta_t-\mO(\sqrt{\frac{s\ln p}{n}\delta_t})-\mO\left((\alpha_1\hat\mu_1+1)\left(\frac{ks\ln p}{n}+\sqrt{\frac{s\ln p}{n}}\left(\frac{\bomega_t}{\sqrt{s}}+\frac{\bomega_t^2}{s}\right)\right)\right)\right]\notag\\
& \overset{(a_3)}{=} \hat r_t\left[(\frac{k-1}{k})^2(1-c)^2m_{1t}\gamma \delta_t-\mO\left(\frac{ks\ln p}{n}+\sqrt{\frac{s\ln p}{n}}\left(\frac{\bomega_t}{\sqrt{s}}+\frac{\bomega_t^2}{s}\right)\right)\right].
\end{align}
At step $(a_3)$, we have used the following results:
\begin{itemize}
\item By Proposition \ref{prop:prop2} (\ref{eq:prop2eq3}):
\begin{equation}
\label{eq:proximal_prop2a}
|\alpha_1-1|\leq \sqrt{2M^2\delta_t}\leq \frac{c}{2}.
\end{equation}
\item By Proposition \ref{prop:prop3}, we know $m_{2t}\geq \mu_1(1-\frac{c}{2})^2$; On $\mA_2$, we have  $|\hat\mu_1 - \mu_1|=\mO(\sqrt{\frac{\ln p}{n}})$. Thus,
\[
\hat\mu_1 = \mO(1), \; \frac{\hat\mu_1}{m_{2t}}\geq \frac{\mu_1-\mO(\sqrt{\frac{\ln p}{n}})}{\mu_1(1+\frac{c}{2})^2}\overset{(a_4)}{\geq} (1-\frac{c}{2})^2.
\]
Step $(a_4)$ holds when $\mO(\sqrt{\frac{\ln p}{n}})\leq \mu_1(1-(1-\frac{c}{2})^2(1+\frac{c}{2})^2)$.
\item Combine the previous two statements, we have
\[
\alpha_1\hat\mu_1\frac{k-1}{k}\frac{\gamma m_{1t}}{m_{2t}}\delta_t\geq \frac{k-1}{k}(1-\frac{c}{2})^3\geq \frac{k-1}{k}(1-c)^2\gamma m_{1t}\delta_t, \;(\alpha_1\hat\mu_1+1) = \mO(1).
\]
\item At the last step, we bound $\mO(\sqrt{\frac{s\ln p}{n}\delta_t})$. Since $\mO(\sqrt{\frac{s\ln p}{n}\delta_t})\leq \mO(b\frac{s\ln p}{n}+\frac{1}{4b}\delta_t)$ for any positive value $b$, we can take $b = \mO(\frac{k}{\gamma m_{1t}})$ to be a sufficiently large positive constant  such that $\mO(\sqrt{\frac{s\ln p}{n}\delta_t})\leq \mO(\frac{ks\ln p}{n})+\frac{1}{k}(\frac{k-1}{k})(1-c)^2\gamma m_{1t}\delta_t$. Hence, at step $(a_3)$ follows and  we have
\[
I_1 \geq \hat r_t\left[(\frac{k-1}{k})^2(1-c)^2m_{1t}\gamma \delta_t-\mO\left(\frac{ks\ln p}{n}+\sqrt{\frac{s\ln p}{n}}\left(\frac{\bomega_t}{\sqrt{s}}+\frac{\bomega_t^2}{s}\right)\right)\right].
\]
\end{itemize}
Combine (\ref{eq:proximal_eq6}) and (\ref{eq:proximal_eq1}), we obtain that
{\footnotesize
\begin{align}
\label{eq:proximal_eq7}
\Delta_t(\tilde\beta)\geq \underbrace{-\epsilon+\epsilon_1(1-\delta_t+\eta(\frac{k-1}{k})^2(1-c)^2\gamma m_{1t}\delta_t)}_{I_4}-\underbrace{\epsilon_1\eta \mO\left(\frac{ks\ln p}{n}+\sqrt{\frac{s\ln p}{n}}\left(\frac{\bomega_t}{\sqrt{s}}+\frac{\bomega_t^2}{s}\right)\right)}_{I_5}.
\end{align}
}
We set $R_1 = \eta(\frac{k-1}{k})^2(1-c)^2\gamma m_{1t}$.  Replace $\epsilon_1$ in $I_4$ by its lower bound in  Proposition \ref{prop:prop5}:
\begin{align}
\label{eq:proximal_eq8}
I_4&\geq -\epsilon+(\epsilon+\epsilon(1-\epsilon)\delta_t-\frac{\epsilon(1-\epsilon)}{2}\delta_t^2)(1-\delta_t+R_1\delta_t)\geq -\epsilon^2\delta_t+R_1\epsilon\delta_t-\frac{3\epsilon}{2}\delta_t^2.
\end{align}
Replace $\epsilon_1$ in $I_5$ by its upper bound in  Proposition \ref{prop:prop5}:
\begin{align}
\label{eq:proximal_eq9}
I_5&\leq \epsilon\eta(1+\frac{\delta_t^3}{\epsilon}) \mO\left(\frac{ks\ln p}{n}+\sqrt{\frac{s\ln p}{n}}\left(\frac{\bomega_t}{\sqrt{s}}+\frac{\bomega_t^2}{s}\right)\right).
\end{align}
We set $\epsilon = \frac{R_1}{2} = \frac{\eta\gamma (\frac{k-1}{k})^2(1-c)^2m_{1t}}{2}$, which is guaranteed to be less than $\frac{1}{8}$ as $\eta < \frac{1}{2M(M+3)}$ and $m_{1t}\leq 2M$ by (\ref{eq:propDelta_extra1}). Apply (\ref{eq:propDelta_extra1})  again:
\[
\delta_t\leq \delta_t^{upper}(c,k)\leq \frac{c\gamma\eta}{M}(\frac{k-1}{k})^2(1-c)^2\leq \frac{c\gamma\eta m_{1t}}{2(1-\frac{c^2}{16M^2})}(\frac{k-1}{k})^2(1-c)^2\leq \frac{c}{1-\frac{c^2}{16}}\epsilon\leq \frac{64c}{63}\epsilon.
\] 
Plug the upper bound of $\delta_t$ into $I_4$ and $I_5$ and plug in the expression for $\epsilon$, we obtain the following relationships:
{\small
\begin{align}
&I_4\geq  \epsilon^2\delta_t - \frac{3\epsilon}{2}\times \frac{64c\epsilon}{63}\delta_t\geq (1-\frac{96}{63}c)\epsilon^2\delta_t\geq (1-c)^2\epsilon^2\delta_t=\frac{\eta^2\gamma^2m_{1t}^2}{4}(\frac{k-1}{k})^4(1-c)^6\delta_t.\label{eq:proximal_eq10}\\
&I_5\leq \eta^2\mO\left(\frac{ks\ln p}{n}+\sqrt{\frac{s\ln p}{n}}\left(\frac{\bomega_t}{\sqrt{s}}+\frac{\bomega_t^2}{s}\right)\right)\label{eq:proximal_eq11}.
\end{align}
}
Combine (\ref{eq:proximal_eq10}) and (\ref{eq:proximal_eq11}) with (\ref{eq:proximal_eq7}), we acquire 
\begin{equation}
\label{eq:proximal_eq12}
\Delta_t(\tilde\beta)\geq \frac{\eta^2\gamma^2m_{1t}^2}{4}(\frac{k-1}{k})^4(1-c)^6\delta_t-\eta^2\mO\left(\frac{ks\ln p}{n}+\sqrt{\frac{s\ln p}{n}}\left(\frac{\bomega_t}{\sqrt{s}}+\frac{\bomega_t^2}{s}\right)\right).
\end{equation}
If we can show that $\tilde\beta$ is indeed a feasible solution, we can use $\Delta_t(\tilde\beta)$ as a lower bound of $\Delta_{t,t+1}$ and conclude Lemma \ref{lem:proximal}. We now show the feasibility of $\tilde\beta$.

$\tilde\beta$ is a feasible when $\|\tilde\beta\|_1\leq \|\xi_1\|_1+B_{t+1}$.  This is guaranteed if the inequality holds replacing $\epsilon_1$ and  $\|\beta_{t}\|_1\leq \|\xi_1\|_1+B_{t}$ by their upper bounds:
\begin{align}
&(1-\epsilon)(\|\xi_1\|_1+B_{t})+\epsilon_1\|\xi_1\|_1\leq \|\xi_1\|_1+B_{t+1}\notag\\
\Leftarrow&(1-\epsilon)(\|\xi_1\|_1+B_{t})+(\epsilon+\epsilon(1-\epsilon)\delta_t +\epsilon\delta_t^2+3\epsilon^2\delta_t^2+\frac{1}{2}\delta_t^3)\|\xi_1\|_1\leq \|\xi_1\|_1+B_{t+1}\notag\\
\Leftarrow& \epsilon\delta_t \|\xi\|_1+4\epsilon\delta_t^2\|\xi_1\|_1+\frac{1}{2}\delta_t^3\|\xi\|_1+(B_t - B_{t+1})\leq \epsilon B_{t}.\notag
\end{align}
Since $B_{t}\geq \frac{4c_{B_2}(1+c_{B_1})}{\nu}\sqrt{\frac{s^2\ln p}{n}}$ when $t\leq T^*$ and $B_t - B_{t+1}\leq c_{B_2}\eta(1+c_{B_1})\sqrt{\frac{s^2\ln p}{n}}$ by construction, we must have $B_t - B_{t+1}\leq \frac{\nu\eta}{4}B_{t}$. Combine it with the last display, we obtain the $\tilde\beta$ is a feasible solution when
\begin{align}
\label{eq:proximal_eq13}
\epsilon\delta_t \|\xi\|_1+4\epsilon\delta_t^2\|\xi_1\|_1+\frac{1}{2}\delta_t^3\|\xi\|_1\leq (\epsilon-\frac{\nu\eta}{4}) B_{t}.
\end{align}
Notice that $\delta_t\leq\delta_t^{upper}(c,k)$. As a result, we have $\delta_t\leq c\frac{B_t}{\|\xi_1\|_1}$, $\delta_t^2\leq c^3\frac{B_t}{8\|\xi_1\|_1}$ and $\delta_t^3\leq c^4\frac{\gamma B_t\eta}{8M^3\|\xi_1\|_1}$.  Consequently,
\begin{align*}
\epsilon\delta_t\|\xi_1\|_1\leq c \epsilon B_t, \; 4\epsilon\delta_t^2\|\xi_1\|_1\leq \frac{c^3\epsilon}{2},\; \frac{1}{2}\delta_t^3\|\xi\|_1\leq \frac{c^4\gamma\eta}{16M^3}B_t.
\end{align*}
Combine the last display with (\ref{eq:proximal_eq13}), $\tilde\beta$ is a feasible solution as long as 
\begin{align}
\label{eq:proximal_eq14}
\frac{c^4\gamma\eta}{16M^3}+\frac{\nu\eta}{4}\leq (1-c-\frac{c^3}{2})\epsilon.
\end{align}
Recall that $\epsilon  =  \frac{\eta\gamma (\frac{k-1}{k})^2(1-c)^2m_{1t}}{2}\geq \frac{\eta\gamma (\frac{k-1}{k})^2(1-c)^2}{M}(1-\frac{c^2}{8M^2})$ by (\ref{eq:propDelta_extra1}) and $\nu\leq \frac{\gamma^2 (1-c_0)}{M}$, plug then into (\ref{eq:proximal_eq14}), $\tilde\beta$ is a feasible solution as long as 
\begin{align*}
&\frac{c^4\gamma\eta}{16M^3}+\frac{\gamma^2\eta(1-c_0)}{4M}\leq (1-c-\frac{c^3}{2})\frac{\eta\gamma (\frac{k-1}{k})^2(1-c)^2}{M}(1-\frac{c^2}{8M^2})\\
\Leftrightarrow& \frac{c^4}{16M^3}+\frac{\gamma(1-c_0)}{4M}\leq (1-c-\frac{c^3}{2})\frac{(\frac{k-1}{k})^2(1-c)^2}{M}(1-\frac{c^2}{8M^2})\\
\Leftarrow& (\frac{k-1}{k})^2(1-c)^4\geq \frac{1}{3},
\end{align*}
and the last step is guaranteed in (\ref{condI}). Hence, $\tilde\beta$ is a feasible solution. This concludes our proof.
\end{proof}
\subsection{Proof of Lemma \ref{lem:obj_lowerbound}}
\begin{proof}
We aim to show that for a sufficiently large $C$ and large $n$:
\[
\beta_t^\top\hat\Lambda\beta_t\left(\hat r_{t+1} - \hat r_t - \frac{2\hat r_t}{\eta \beta_t^\top\hat\Lambda\beta_t}\Delta_{t,t+1}\right)\geq -\mO\left(\hat r_t\sqrt{\frac{k's\ln p}{n}}(\delta_t+\delta_{t+1}+\frac{\bomega_t^2}{k's})\right)
\] 
Set $h \coloneqq \beta_{t+1} - \beta_t$ and $\tilde g_t\coloneqq (\hat\Sigma - \hat r_t\hat\Lambda)\beta_t$ be the scaled gradient at $\beta_t$, then,
\begin{align}
\beta_t^\top\hat\Lambda\beta_t\left(\hat r_{t+1} - \hat r_t - \frac{2\hat r_t}{\eta \beta_t^\top\hat\Lambda\beta_t}\Delta_{t,t+1}\right)& \overset{(a_1)}{=} \left(\underbrace{\beta_t^\top\hat\Lambda\beta_t(\hat r_{t+1} - \hat r_t )-2\tilde g_t^\top h}_{I_1}+\frac{\hat r_t}{\eta}\|h\|_2^2\right),\label{eq:obj_boundeq1}
\end{align}
At step $(a_1)$, we have re-expressed $\Delta_{t,t+1}$:
{\small
\begin{align*}
2\frac{\hat r_t}{\eta}\Delta_{t,t+1} = 2\frac{\hat r_t}{\eta}\theta^\top (\beta_{t+1}-\beta_t) = 2\frac{\hat r_t}{\eta}(\beta_t+\frac{\eta}{\hat r_t}\tilde g_t)^\top(\beta_{t+1}-\beta_t)=2\tilde g_t^\top h - 2\frac{\hat r_t}{\eta}(1-\beta_{t+1}^\top\beta_t )=2\tilde g_t^\top h -\frac{\hat r_t}{\eta}\|h\|_2^2.
\end{align*}
}
 Notice that
\begin{equation}
\label{eq:obj_boundeq2}
I_1= \frac{1}{\beta_{t+1}^\top\hat\Lambda\beta_{t+1}}\left(\beta_{t+1}^\top\hat\Sigma \beta_{t+1}\times \beta_{t}^\top\hat\Lambda\beta_t - \beta_{t}^\top\hat\Sigma \beta_{t}\times \beta_{t+1}^\top\hat\Lambda\beta_{t+1} -2\beta_{t+1}^\top\hat\Lambda\beta_{t+1}\times \tilde g_{t}^\top h\right).
\end{equation}
For the convenience of notation,  we define $\Lambda_{t,t}\coloneqq \beta_t^\top\hat\Lambda\beta_t$, $\Lambda_{t+1,t+1}\coloneqq \beta_{t+1}^\top\hat\Lambda\beta_{t+1}$,  $\Lambda_{h,t}\coloneqq h^\top\hat\Lambda\beta_{t}$, $\Lambda_{h,t+1}\coloneqq h^\top\hat\Lambda\beta_{t+1}$ and $\Lambda_{h,h}\coloneqq h^\top\hat\Lambda\beta_{h}$. Similarly, we define $\Sigma_{t,t}$, $\Sigma_{t+1,t+1}$, $\Sigma_{h,t}$,$\Sigma_{h,t+1}$ and $\Sigma_{h,h}$ replacing $\hat\Lambda$ by $\hat\Sigma$ in the previous definitions, e.g.,  $\Sigma_{t,t}\coloneqq \beta_t^\top\hat\Sigma\beta_t$. These set of definitions are only used in this proof. Of course, these are not independent definitions and can be re-expressed using each others:
\begin{align*}
&\Lambda_{t,t} = \beta_{t+1}^\top\hat\Lambda\beta_{t+1}+h^\top\hat\Lambda h-2h^\top\hat\Lambda \beta_{t+1} =  \Lambda_{t+1,t+1}+\Lambda_{h,h}-2\Lambda_{h,t+1},\\
&\Sigma_{t,t}=  \beta_{t+1}^\top\hat\Sigma\beta_{t+1}+h^\top\hat\Sigma h-2h^\top\hat\Sigma \beta_{t+1} =  \Sigma_{t+1,t+1}+\Sigma_{h,h}-2\Lambda_{h,t+1}.
\end{align*}
Similarly,
\begin{align*}
&\Lambda_{t+1,t+1}=  \Lambda_{t,t}+\Lambda_{h,h}+2\Lambda_{h,t},\;\Sigma_{t+1,t+1}  =    \Sigma_{t,t}+\Sigma_{h,h}+2\Lambda_{h,t},\\ 
& \Lambda_{t+1,h} = \Lambda_{t,h}+\Lambda_{h,h},\; \Sigma_{t+1,h} = \Sigma_{t,h}+\Sigma_{h,h}. 
\end{align*}
As a result, we have
\begin{align}
\Lambda_{t+1,t+1}I_1 =&\Sigma_{t+1,t+1}\left(\cancel{\Lambda_{t+1,t+1}}+\Lambda_{h,h}-2\Lambda_{t+1,h}\right) - \Lambda_{t+1,t+1}\left(\cancel{\Sigma_{t+1,t+1}}+\Sigma_{h,h}-\cancel{2\Sigma_{t+1,h}}\right)\notag\\
& -2\Lambda_{t+1,t+1}(\cancel{\Sigma_{t+1,h}}-\Sigma_{h,h}-\hat r_t(\Lambda_{t+1,h}-\Lambda_{h,h}))\notag\\
\overset{(a_2)}{=}&\Lambda_{t+1,t+1}\left[\hat r_{t+1}\Lambda_{h,h}-2\hat r_{t+1}\Lambda_{t+1,h}-\cancel{\Sigma_{h,h}}+\cancel{2}\Sigma_{h,h}+2\hat r_t\Lambda_{t+1,h}-2\hat r_t\Lambda_{h,h}\right]\notag\\
\overset{(a_3)}{\geq}&\Lambda_{t+1,t+1}\left[\hat r_{t+1}\Lambda_{h,h}-2\hat r_{t+1}\Lambda_{t+1,h}+2\hat r_t\Lambda_{t+1,h}-2\hat r_t\Lambda_{h,h}\right]\notag,
\end{align}
where we have replaced $\Sigma_{t+1,t+1}$ by $\hat r_{t+1}\Lambda_{t+1,t+1}$ at step $(a_2)$ and drop a term $\Lambda_{t+1, t+1}\Sigma_{h, h}\geq 0$ at step $(a_3)$. Rearrange the remaining terms, we obtain that 
\begin{align}
 \label{eq:obj_boundeq7}
\frac{I_1}{\hat r_t} \geq &-2\Lambda_{h,h}-2\frac{\hat r_{t+1}-\hat r_t}{\hat r_t}\Lambda_{t+1,h}.
\end{align}
By definition:
\begin{align}
|\frac{\hat r_t - \hat r_{t+1}}{\hat r_t}|=&|\frac{\Sigma_{t,t}\Lambda_{t+1,t+1}-\Sigma_{t+1, t+1}\Lambda_{t,t}}{\hat r_t\Lambda_{t+1,t+1}\Lambda_{t,t}}|\notag\\
=&|\frac{\Lambda_{t+1,t+1}(\cancel{\Sigma_{t+1,t+1}}+\Sigma_{h,h}-2\Sigma_{h,t+1})-\Sigma_{t+1, t+1}(\cancel{\Lambda_{t+1,t+1}}+\Lambda_{h,h}-2\Lambda_{h,t+1})}{\hat r_t\Lambda_{t+1,t+1}\Lambda_{t,t}}|\notag\\
\overset{(a_4)}{=}&|\frac{(\Sigma_{h,h}-\hat r_{t+1}\Lambda_{h,h})-2(\Sigma_{h,t+1}-\hat r_{t+1}\Lambda_{h,t+1})}{\hat r_t\Lambda_{t,t}}|\notag\\
\overset{(a_5)}{\leq}&\frac{\max\{\hat r_{t+1}\Lambda_{hh},\Sigma_{hh}\}}{\hat r_t \Lambda_{tt}}+2(\frac{\hat r_{t+1}}{\hat r_t})^{\frac{1}{2}}\frac{\Sigma_{hh}^{\frac{1}{2}}\Lambda_{t+1,t+1}^{\frac{1}{2}}}{\hat r_t^{\frac{1}{2}}\Lambda_{tt}}+2(\frac{\hat r_{t+1}}{\hat r_t})\frac{\Lambda_{hh}^{\frac{1}{2}}\Lambda_{t+1,t+1}^{\frac{1}{2}}}{\Lambda_{tt}},\label{eq:obj_bound8}
\end{align}
where we have replaced $\Sigma_{t+1,t+1}$ by $\hat r_{t+1}\Lambda_{t+1, t+1}$ at step $(a_4)$, and upper bound $|\Lambda_{h,t+1}|$ by $\Lambda_{h,h}^{\frac{1}{2}}\Lambda_{t+1,t+1}^{\frac{1}{2}}$ and $|\Sigma_{h,t+1}|$ by $(\Sigma_{h,h}\hat r_{t+1}\Lambda_{t+1,t+1})^{\frac{1}{2}}$ at step $(a_5)$. Similarly,  we have
\begin{align}
|\frac{\hat r_t - \hat r_{t+1}}{\hat r_t}|=&|\frac{\Sigma_{t,t}(\cancel{\Lambda_{t,t}}+\Lambda_{h,h}+2\Lambda_{h,t})-\Lambda_{t,t}(\cancel{\Sigma_{t,t}}+\Sigma_{h,h}+2\Sigma_{h,t})}{\hat r_t\Lambda_{t+1,t+1}\Lambda_{t,t}}|\notag\\
\leq & \frac{\max\{\hat r_{t}\Lambda_{hh},\Sigma_{hh}\}}{\hat r_t \Lambda_{t+1,t+1}}+2\frac{\Lambda_{h,h}^{\frac{1}{2}}\Lambda_{t,t}^{\frac{1}{2}}}{\Lambda_{t+1,t+1}}+2\frac{\Sigma_{h,h}^{\frac{1}{2}}\Lambda_{t,t}^{\frac{1}{2}}}{\hat r_t^{\frac{1}{2}}\Lambda_{t+1,t+1}}. \label{eq:obj_bound9}
\end{align}
By (\ref{eq:obj_bound8}), (\ref{eq:obj_bound9}), $|\Lambda_{t+1,h}|\leq \Lambda_{t+1,t+1}^{\frac{1}{2}}\Lambda_{h,h}^{\frac{1}{2}}$ , and set $R_1 = \max\{\frac{\rho_1}{\hat r_t}, \frac{\hat r_{t+1}}{\hat r_{t}},1\}$, we obtain
{\small
\begin{align}
&|\frac{r_t - r_{t+1}}{r_t}\Lambda_{t+1,h}|\leq\left\{\begin{array}{l}  R_1\left[2\frac{\Lambda_{t+1,t+1}}{\Lambda_{tt}}(\Lambda_{hh}+(\frac{\Sigma_{hh}}{\rho_1})^{\frac{1}{2}}\Lambda_{hh}^{\frac{1}{2}})+\frac{\max\{\Lambda_{hh}^{\frac{3}{2}},\frac{\Sigma_{hh}}{\rho_1}\Lambda_{hh}^{\frac{1}{2}}\}}{\Lambda_{tt}^{\frac{1}{2}}}(\frac{\Lambda_{t+1,t+1}}{\Lambda_{tt}})^{\frac{1}{2}}\right],\\
R_1\left[2(\frac{\Lambda_{t,t}}{\Lambda_{t+1,t+1}})^{\frac{1}{2}}(\Lambda_{hh}+(\frac{\Sigma_{hh}}{\rho_1})^{\frac{1}{2}}\Lambda_{hh}^{\frac{1}{2}})+\frac{\max\{\Lambda_{hh}^{\frac{3}{2}},\frac{\Sigma_{hh}}{\rho_1}\Lambda_{hh}^{\frac{1}{2}}\}}{\Lambda_{t+1,t+1}^{\frac{1}{2}}}\right].
\end{array}\right. \label{eq:obj_bound10}
\end{align}
}
When $\Lambda_{t+1,t+1}\geq \Lambda_{t,t}$, we use the second bound from (\ref{eq:obj_bound10}):
\[
|\frac{r_t - r_{t+1}}{r_t}\Lambda_{t+1,h}|\leq R_1\left[2(\Lambda_{hh}+(\frac{\Sigma_{hh}}{\rho_1})^{\frac{1}{2}}\Lambda_{hh}^{\frac{1}{2}})+\frac{\max\{\Lambda_{hh}^{\frac{3}{2}},\frac{\Sigma_{hh}}{\rho_1}\Lambda_{hh}^{\frac{1}{2}}\}}{\Lambda_{t,t}^{\frac{1}{2}}}\right].
\]
When $\Lambda_{t+1,t+1}< \Lambda_{t,t}$, we use the the first bound from (\ref{eq:obj_bound10}):
\[
|\frac{r_t - r_{t+1}}{r_t}\Lambda_{t+1,h}|\leq R_1\left[2(\Lambda_{hh}+(\frac{\Sigma_{hh}}{\rho_1})^{\frac{1}{2}}\Lambda_{hh}^{\frac{1}{2}})+\frac{\max\{\Lambda_{hh}^{\frac{3}{2}},\frac{\Sigma_{hh}}{\rho_1}\Lambda_{hh}^{\frac{1}{2}}\}}{\Lambda_{t,t}^{\frac{1}{2}}}\right].
\]
Consequently, we always have
\begin{align}
|\frac{r_t - r_{t+1}}{r_t}\Lambda_{t+1,h}|\leq R_1\left[2(\Lambda_{hh}+(\frac{\Sigma_{hh}}{\rho_1})^{\frac{1}{2}}\Lambda_{hh}^{\frac{1}{2}})+\frac{\max\{\Lambda_{hh}^{\frac{3}{2}},\frac{\Sigma_{hh}}{\rho_1}\Lambda_{hh}^{\frac{1}{2}}\}}{\Lambda_{t,t}^{\frac{1}{2}}}\right].
 \label{eq:obj_bound11}
\end{align}
Combine (\ref{eq:obj_bound11}) into ( \ref{eq:obj_boundeq7}):
{\small
\begin{align}
\frac{I_1}{\hat r_{t}}&\geq -2\Lambda_{h,h}-R_1\left[2(\Lambda_{hh}+(\frac{\Sigma_{hh}}{\rho_1})^{\frac{1}{2}}\Lambda_{hh}^{\frac{1}{2}})+\frac{\max\{\Lambda_{hh}^{\frac{3}{2}},\frac{\Sigma_{hh}}{\rho_1}\Lambda_{hh}^{\frac{1}{2}}\}}{\Lambda_{t,t}^{\frac{1}{2}}}\right].
 \label{eq:obj_bound12}
\end{align}
}
Let $c'$ be a small positive constant that we will specify later. By Lemma \ref{lem:basic_bound1} (\ref{basic1main1}):
\begin{align}
\Lambda_{h,h} &\leq  M \|h\|_2^2+CM\sqrt{\frac{k's\ln p}{n}}(\|h\|_2^2+\frac{\|h\|_1^2}{k's})\notag\\
&\overset{(a_6)}{\leq} M \|h\|_2^2+CM\sqrt{\frac{k's\ln p}{n}}\|h\|_2^2+I_2\notag\\
&\overset{(a_7)}{\leq} M(1+c')\|h\|_2^2+I_2
\label{eq:obj_bound13},
\end{align}
for some $I_2\leq \mO\left(\sqrt{\frac{s\ln p}{k'n}}\left(\delta_t+\delta_{t+1}+\frac{\bomega_t^2}{s}\right)\right)$. At step $(a_6)$,  we upper bound $\|h\|_1^2$ by Proposition \ref{prop:prop4}:
\begin{align*}
\|h\|_1\leq \|\beta_t-\xi_1\|_1+\|\beta_{t+1}-\xi_1\|_1\leq 2\sqrt{2s\delta_t}+2\sqrt{2s\delta_{t+1}}+2\omega_t\leq 2\sqrt{2s\delta_t}+2\sqrt{2s\delta_{t+1}}+2\bomega_t
\end{align*}
and also $(a+b+c)^2\leq 3(a^2+b^2+c^2)$. Step  holds $(a_7)$ when $\sqrt{\frac{k's\ln p}{n}}\leq \frac{c'}{CM}$. Similarly, we have
\begin{align}
\Sigma_{h,h} &\leq  M\rho_1 \|h\|_2^2+CM\rho_1\sqrt{\frac{k's\ln p}{n}}(\|h\|_2^2+\frac{\|h\|_1^2}{k's})\leq  M\rho_1(1+c')\|h\|_2^2+\rho_1 I_2. \label{eq:obj_bound14}
\end{align}
We now use  (\ref{eq:obj_bound13}) and (\ref{eq:obj_bound14}) to bound $\Lambda_{hh}^{\frac{3}{2}}$, $\Sigma_{hh}\Lambda_{hh}^{\frac{1}{2}}$ and $\Lambda_{hh}^{\frac{1}{2}}\Sigma_{hh}^{\frac{1}{2}}$ , which also appeared  in (\ref{eq:obj_bound12}) to upper bound $I_1$. Consider two cases: (1) when $(1+c')M\|h\|_2^2\geq \frac{1}{c'} I_2$, we have $\Lambda_{hh}\leq M(1+c')^2\|h\|_2^2$ and $\Sigma_{hh}\leq M\rho_1(1+c')^2\|h\|_2^2$; (2) when $(1+c')M\|h\|_2^2< \frac{1}{c'} I_2$, we have $\Lambda_{hh}\leq (1+\frac{1}{c'})I_2$ and $\Sigma_{hh}\leq (1+\frac{1}{c'})\rho_1 I_2$. Combine these two cases together, we have
\begin{align}
&\Lambda_{hh}^{\frac{3}{2}}\leq M^{\frac{3}{2}}(1+c')^3\|h\|_2^3+(1+\frac{1}{c'})^{\frac{3}{2}}I_2^{\frac{3}{2}},\label{eq:obj_bound15}\\
& \Sigma_{hh}\Lambda_{hh}^{\frac{1}{2}}\leq \rho_1M^{\frac{3}{2}}(1+c')^{3}\|h\|_2^3+(1+\frac{1}{c'})^{\frac{3}{2}}\rho_1I_2^{\frac{3}{2}}\label{eq:obj_bound16},\\
&\Sigma_{hh}^{\frac{1}{2}}\Lambda_{hh}^{\frac{1}{2}}\leq \rho_1^{\frac{1}{2}}M(1+c')\|h\|_2^2+\rho_1^{\frac{1}{2}}I_2\label{eq:obj_bound17}.
\end{align}
By Lemma \ref{lem:basic_bound3} (\ref{basic3main2}), we have 
\begin{equation}
\label{eq:obj_bound15}
\Lambda_{tt}\geq \frac{1}{M}(1-CM^2(1+c_{B_1})^2\sqrt{\frac{s\ln p}{n}})\geq \frac{1-c'}{M}.
\end{equation}
The last step holds when $\sqrt{\frac{s\ln p}{n}}\leq \frac{c'}{CM^2(1+c_{B_1})^2}$. Plug (\ref{eq:obj_bound13}) and (\ref{eq:obj_bound15}) - (\ref{eq:obj_bound18}) back to   (\ref{eq:obj_bound12}), we obtain
{\small
\begin{align}
\label{eq:obj_bound16}
\frac{I_1}{\hat r_t}&\geq -2M(1+c')\|h\|_2^2-2I_2 - R_1\left[4M(1+c')\|h\|_2^2+4I_2+\left(M^{\frac{3}{2}}(1+c')^3\|h\|_2^3+(1+\frac{1}{c'})^{\frac{3}{2}}I_2^{\frac{3}{2}}\right)(\frac{M}{(1-c')})^{\frac{1}{2}}\right]\notag\\
& \geq -6R_1(1+c')M\|h\|_2^2-R_1 M^2\frac{(1+c')^{3}}{\sqrt{1-c'}}\|h\|_2^3- I_2R_1(6(1+c')+I_2^{\frac{1}{2}}(\frac{(1+\frac{1}{c'})^3M}{(1-c')})^{\frac{1}{2}})
\end{align}
}
We next upper bound $R_1$.  From Lemma \ref{lem:coef_bound} (\ref{eq:coef_main1}) and Lemma \ref{lem:basic_bound3} (\ref{basic3main1}),  when $\delta_t\leq \frac{c^2}{8M^2}$,  we have
\begin{align}
\label{eq:obj_bound17}
\hat r_{t}&\geq \rho_1(1-\frac{m_{1t}}{\beta_{t}^\top\Lambda\beta_t}\frac{c^2}{8M^2})- \rho_1\mO(\frac{ks\ln p}{n}+\sqrt{\frac{s\ln p}{n}}\left(\frac{\bomega_t}{\sqrt{s}}+\frac{\bomega_t^2}{\sqrt{k'}s}\right))\notag\\
&\geq  \rho_1(1-\frac{c^2}{4})-\rho_1 \mO(\frac{ks\ln p}{n}+\sqrt{\frac{s\ln p}{n}}\left(\frac{\bomega_t}{\sqrt{s}}+\frac{\bomega_t^2}{\sqrt{k'}s}\right)).
\end{align}
The last step of (\ref{eq:obj_bound17}) uses Proposition \ref{prop:prop2}:
\begin{align}
\label{eq:obj_prop2a}
\frac{2}{M}(1-\frac{\delta_t}{2})\leq m_{1t}\leq 2M.
\end{align}
By Lemma \ref{lem:coef_bound} (\ref{eq:coef_main2}),  Lemma \ref{lem:basic_bound3} (\ref{basic3main1}) and (\ref{eq:obj_prop2a}), we have
\begin{align}
\label{eq:obj_bound18}
\hat r_{t+1}&\leq \rho_1+\rho_1 \mO(\frac{ks\ln p}{n}+\sqrt{\frac{s\ln p}{n}}\left(\frac{\bomega_{t+1}}{\sqrt{s}}+\frac{\bomega_{t+1}^2}{\sqrt{k'}s}\right))
\end{align}
Hence, we have
\begin{align}
\label{eq:obj_bound17}
R_1& \leq \frac{\rho_1+\rho_1 \mO\left(\frac{ks\ln p}{n}+\sqrt{\frac{s\ln p}{n}}\left(\frac{\bomega_t}{\sqrt{s}}+\frac{\bomega_t^2}{\sqrt{k'}s}\right)\right)}{\rho_1(1-\frac{c^2}{4})-\rho_1 \mO\left(\frac{ks\ln p}{n}+\sqrt{\frac{s\ln p}{n}}\left(\frac{\bomega_t}{\sqrt{s}}+\frac{\bomega_t^2}{\sqrt{k'}s}\right)\right)}\notag\\
& \leq \frac{1}{1-\frac{c^2}{4}}  +\mO\left(\frac{ks\ln p}{n}+\sqrt{\frac{s\ln p}{n}}\left(\frac{\bomega_t}{\sqrt{s}}+\frac{\bomega_t^2}{\sqrt{k'}s}\right)\right)\notag\\
&\leq   \frac{1}{1-\frac{c^2}{4}}  +\mO\left(\frac{ks\ln p}{n}+\sqrt{\frac{s\ln p}{n}}(1+c_{B_1})^2\right),
\end{align}
where the last step has used the facts that  $\bomega_t\leq (1+c_{B_1})\sqrt{s}$ and $\bomega_t^2\leq (1+c_{B_1})^2s$. Hence, when $\iota$ is sufficiently small such that
\[
\mO\left(\frac{ks\ln p}{n}+\sqrt{\frac{s\ln p}{n}}(1+c_{B_1})^2\right)\leq 1+c - \frac{1}{1-\frac{c^2}{4}},
\]
we have
\begin{align}
\label{eq:obj_bound18}
R_1& \leq 1+c.
\end{align}
We now take $c'$ such that $\frac{(1+c')^4}{\sqrt{1-c'}} = 1+c$. We can numerically check that $\frac{c}{6}<c' < c$. Combine (\ref{eq:obj_bound18}), $\|h\|_2\leq \|\beta_t\|_2+\|\beta_{t+1}\|_2\leq 2$,  and our choice of $c'$ into (\ref{eq:obj_bound18}), we obtain that
\begin{align}
\label{eq:obj_bound19}
\frac{I_1}{\hat r_t}&\geq -2M(M+3) (1+c)^2\|h\|_2^2-\left(\mO\left(\sqrt{\frac{1}{c^3}I_2}\right)+1\right)I_2.
\end{align}
Since $I_2 \leq \mO\left(\sqrt{\frac{s\ln p}{k'n}}\left(\delta_t+\delta_{t+1}+\frac{\bomega_t^2}{s}\right)\right)$, and $\delta_t\leq 1, \delta_{t+1}\leq 1$, $\frac{\bomega_t^2}{s}\leq (1+c_{B_1})^2$, we have
\[
\mO\left(\sqrt{\frac{1}{c^3}I_2}\right) \leq  \mO\left(\sqrt{\frac{1}{c^3}\sqrt{\frac{s\ln p}{n}}\left(1+c_{B_1}\right)^2}\right) \leq \mO(1).
\]
Combine the last display with (\ref{eq:obj_bound19}) and plug in the expression for $I_2$, we have
\begin{align}
\label{eq:obj_bound19}
\frac{I_1}{\hat r_t}&\geq -2M(M+3) (1+c)^2\|h\|_2^2-\mO\left(\sqrt{\frac{s\ln p}{k'n}}\left(\delta_t+\delta_{t+1}+\frac{\bomega_t^2}{s}\right)\right).
\end{align}
Plug the bound on $I_1$ in (\ref{eq:obj_bound19}) back into (\ref{eq:obj_boundeq1}):
{\footnotesize
\begin{align*}
\beta_t^\top\hat\Lambda\beta_t\left(\hat r_{t+1} - \hat r_t - \frac{2\hat r_t}{\eta \beta_t^\top\hat\Lambda\beta_t}\Delta_{t,t+1}\right)&\geq -2M(M+3) (1+c)^2\|h\|_2^2\hat r_t+\frac{\hat r_t}{\eta}\|h\|_2^2-\hat r_t\mO\left(\sqrt{\frac{s\ln p}{k'n}}\left(\delta_t+\delta_{t+1}+\frac{\bomega_t^2}{s}\right)\right).
\end{align*}
}
Hence, when $\eta\leq \frac{1}{(1+c)^22M(M+3)}$, we obtain the desired bound
\begin{align*}
\beta_t^\top\hat\Lambda\beta_t\left(\hat r_{t+1} - \hat r_t - \frac{2\hat r_t}{\eta \beta_t^\top\hat\Lambda\beta_t}\Delta_{t,t+1}\right)&\geq -\hat r_t\mO\left(\sqrt{\frac{s\ln p}{k'n}}\left(\delta_t+\delta_{t+1}+\frac{\bomega_t^2}{s}\right)\right).
\end{align*}
\end{proof}
\section{Proofs of supporting Propositions and Lemmas}
\label{app:proof_supportings}
In this section, we provide proofs to supporting propositions and Lemmas appearing in Section \ref{app:supporting} as well as Proposition \ref{prop:prop5} and \ref{prop:prop6} in the proofs of Lemmas. We first present proofs of different Propositions, with those for Proposition \ref{prop:prop1} and Proposition \ref{prop:prop6} combined since they are both about tail bounds. We then give proofs to supporting Lemmas.
\subsection{Proof of Proposition \ref{prop:prop1} and \ref{prop:prop6}}
\subsubsection*{Event $\mA_1(k')$}
Let $J_0 = J_1\cap J_2$, $J_{11} = J_1\setminus J_0$ and $J_{21} = J_2\setminus J_0$. For $H = \Sigma$, we have
\begin{align*}
\|\hat\Sigma_{J_1 J_2} - \Sigma_{J_1 J_2}\|_{}op& = \|\frac{1}{n}\sum_{i=1}^n (X_{iJ_1}X_{iJ_2}^{\top}-\bE[X_{iJ_1} X_{iJ_2}^\top])\|_{op}
\end{align*}
We can decompose $X_{iJ_2}$ as $X_{iJ_2} = W_{i J_2}+\Sigma_{J_2 J_1}\Sigma_{J_1 J_1}^{-1}X_{iJ_1}$ where $W_{iJ_2}\sim N(0,  \Sigma_{J_2 J_2} - \Sigma_{J_2J_1}\Sigma_{J_1J_1}^{-1}\Sigma_{J_1 J_2})$ is independent of $X_{iJ_1}$. Hence,
\begin{align*}
\|\hat\Sigma_{J_1 J_2} - \Sigma_{J_1 J_2}\|_{}op& \leq I_1 + I_2,
\end{align*}
where $I_1 = \|\frac{1}{n}\sum_{i=1}^n (X_{iJ_1}X_{iJ_1}^\top \Sigma_{J_1 J_1}^{-1}\Sigma_{J_1 J_2}-\bE[X_{iJ_1}X_{iJ_1}^\top \Sigma_{J_1 J_1}^{-1}\Sigma_{J_1 J_2}])\|_{op}$ and $I_2 = \|\frac{1}{n}\sum_{i=1}^n (X_{iJ_1}W_{iJ_2}^\top -\bE[X_{iJ_1}W_{iJ_2}^\top ])\|_{op}$. Set $Z_{iJ_1} = \Sigma_{J_1 J_1}^{-\frac{1}{2}}X_{iJ_1}$, $Z_{i J_2} = ( \Sigma_{J_2 J_2} - \Sigma_{J_2J_1}\Sigma_{J_1J_1}^{-1}\Sigma_{J_1 J_2})^{-\frac{1}{2}}W_{i J_2}$. Then, $Z_{iJ_1}\sim N(0, \Id)$, $Z_{iJ_2}\sim N(0, \Id)$ and
\begin{align}
I_1&\leq \|\Sigma_{J_1 J_1}^{\frac{1}{2}}\frac{1}{n}\sum_{i=1}^n\left(Z_{iJ_1}Z_{iJ_1}^\top - \bE[Z_{i J_1} Z_{iJ_1}]\right)\Sigma_{J_1 J_1}^{-\frac{1}{2}}\Sigma_{J_1 J_2}\|_{op}\notag\\
&\overset{(a_1)}{\leq}  \|\Sigma_{J_1 J_1}^{\frac{1}{2}}\frac{1}{n}\sum_{i=1}^n\left(Z_{iJ_1}Z_{iJ_1}^\top - \bE[Z_{i J_1} Z_{iJ_1}]\right)\Sigma_{J_2 J_2}^{\frac{1}{2}}\|_{op}\notag\\
& \leq   \|\frac{1}{n}\sum_{i=1}^n\left(Z_{iJ_1}Z_{iJ_1}^\top - \bE[Z_{i J_1} Z_{iJ_1}]\right)\|_{op}\|\Sigma_{J_2 J_2}^{\frac{1}{2}}\|_{op}\|\Sigma_{J_1 J_1}^{\frac{1}{2}}\|_{op}\notag\\
&\leq \|\Sigma\|_{op}\|\frac{1}{n}\sum_{i=1}^n\left(Z_{iJ_1}Z_{iJ_1}^\top - \bE[Z_{i J_1} Z_{iJ_1}]\right)\|_{op} \label{eq:prop1eq1}
\end{align}
At step $(a_1)$, we have used the relationship that $\Sigma_{J_2 J_2} - \Sigma_{J_2 J_1}\Sigma_{J_1 J_1}^{-1}\Sigma_{J_1 J_2}$ is positive semi-definite, for all unit vector $v\in \real^{|J_2|}$, we have 
\[
 \|\Sigma^{\frac{1}{2}}_{J_2 J_2}\|^2_{op}= \|\Sigma_{J_2 J_2}\|_{op}\geq v^\top \Sigma_{J_2 J_2} v\geq  \|\Sigma_{J_1 J_1}^{-\frac{1}{2}}\Sigma_{J_1 J_2} v\|_2^2.
 \]
 Apply Proposition (D.1) from  \cite{ma2013sparse} to (\ref{eq:prop1eq1}), we obtain that that $\bP(I_1\leq C\lambda_{\Sigma}^{\max}\sqrt{\frac{k's\ln p}{n}})\rightarrow 1$ for a sufficiently large universal constant $C$.  We can also bound $I_2$ similarly:
 \begin{align*}
  I_2&\leq \|\Sigma_{J_1 J_1}^{\frac{1}{2}}\frac{1}{n}\sum_{i=1}^n\left(Z_{iJ_1}Z_{iJ_2}^\top - \bE[Z_{i J_1} Z_{iJ_2}]\right) ( \Sigma_{J_2 J_2} - \Sigma_{J_2J_1}\Sigma_{J_1J_1}^{-1}\Sigma_{J_1 J_2})^{\frac{1}{2}}\|_{op}\\
  & \leq   \|\Sigma_{J_1 J_1}^{\frac{1}{2}}\frac{1}{n}\sum_{i=1}^n\left(Z_{iJ_1}Z_{iJ_2}^\top - \bE[Z_{i J_1} Z_{iJ_2}]\right) \Sigma_{J_2 J_2}^{\frac{1}{2}}\|_{op}\leq  \|\Sigma\|_{op}\|\frac{1}{n}\sum_{i=1}^n\left(Z_{iJ_1}Z_{iJ_2}^\top - \bE[Z_{i J_1} Z_{iJ_2}]\right)\|_{op}.
 \end{align*}
Apply Proposition (D.2) from  \cite{ma2013sparse} to (\ref{eq:prop1eq1}), we obtain that that $\bP(I_2\leq C\lambda_{\Sigma}^{\max}\sqrt{\frac{k's\ln p}{n}})\rightarrow 1$ for a sufficiently large universal constant $C$.   Combine the probabilistic bounds for $I_1$ and $I_2$, we acquire the desired bound for $\mA_1(k')$ at $H = \Sigma$.

When $H = \Lambda$,  since $\Lambda$ is block diagonal, we have
\[
\|(\hat\Lambda - \Lambda)_{J_1, J_2}\|_{op} = \max_{d=1}^D \|(\hat\Lambda - \Lambda)_{J_1\cap [d], J_2\cap [d]}\|_{op}
\]
Apply the previous arguments for each $\|(\hat\Lambda - \Lambda)_{J_1\cap [d], J_2\cap [d]}\|_{op}$, $d = 1,\ldots, D$ and take a union bound across different $d$, we get the desired probabilistic bound at $H = \Lambda$. Consequently, $\mA_1(k')$ holds with probability approaching 1 for a sufficiently large universal constant $C$ for any given positive constant $k'$.
\subsubsection*{Event $\mA_2$}
When $H = \Sigma$, let $Z_i = (\xi_1^\top\Sigma\xi_1)^{-\frac{1}{2}}\xi_1^\top X_{i}$, we have $Z_i\sim N(0,1)$ and
\begin{align}
|\xi_1^\top(\hat\Sigma - \Sigma)|& = |\frac{1}{n}\sum_{i=1}^n \left((\xi_1^\top X_{i})^2-\bE[(\xi_1^\top X_{i})^2]\right)|\notag\\
&  =(\xi_1^\top\Sigma\xi_1)|\frac{1}{n}\sum_{i=1}^n \left(Z_i^2-\bE[Z_i^2]\right)|\leq \lambda_{\Sigma}^{\max}|\frac{1}{n}\sum_{i=1}^n \left(Z_i^2-\bE[Z_i^2]\right)|\notag
\end{align}
By Bernstein's inequality, we obtain the desired probabilistic bound at $H = \Sigma$. When $H = \Lambda$, let $Z_{id} = (\xi_{1[d]}^\top\Lambda_{[d]}\xi_{1[d]})^{-\frac{1}{2}}\xi_{1[d]}^\top X_{i[d]}$, we have $Z_{id}\sim N(0,1)$ and $Z_{id}$ can be correlated with $Z_{id'}$ for $d, d' = 1,\ldots, D$. Then,
\begin{align*}
&|\xi_1^\top(\hat\Lambda - \Lambda)|   =|\frac{1}{n}\sum_{i=1}^n \sum_{d=1}^D(\xi_{1[d]}^\top\Lambda_{[d]}\xi_{1[d]})\left(Z_{id}^2-\bE[Z_{id}^2]\right)|\\
& \bP\left(|\xi_1^\top(\hat\Lambda - \Lambda)|>t\right)\leq  \bP\left((\xi_{1}^\top\Lambda \xi_{1})|\frac{1}{n}\sum_{i=1}^n \left(Z_{i}^2-\bE[Z_{i}^2]\right)|>t\right)\leq  \bP\left(\lambda_{\Lambda}^{\max}|\frac{1}{n}\sum_{i=1}^n \left(Z_{i}^2-\bE[Z_{i}^2]\right)|>t\right)
\end{align*}
By Bernstein's inequality, we obtain the desired probabilistic bound at $H = \Lambda$. Consequently, $\mA_2$ holds with probability approaching 1 for a sufficiently large universal constant $C$.
\subsubsection*{Event $\mA_3$}
When $H = \Sigma$, set $w_{ij} = x_i^\top\xi_1\sim N(0, \xi_1^\top\Sigma \xi_1)$ and $\bE[w_{ij}x_{ij}] = \Sigma_{j}\xi_1$. We can write $w_{ij}=\frac{\Sigma_{j,}\xi_1}{\Sigma_{jj}}x_{ij}+(\xi_1^\top\Sigma \xi_1-\frac{(\Sigma_j \xi_1)^2}{\Sigma_{jj}})^{\frac{1}{2}}z_{ij}$ with $z_{ij}=(\xi_1^\top\Sigma \xi_1-\frac{(\Sigma_j \xi_1)^2}{\Sigma_{jj}})^{-\frac{1}{2}}(w_{ij}-(\xi_1^\top\Sigma \xi_1-\frac{(\Sigma_j \xi_1)^2}{\Sigma_{jj}})z_{ij})\sim N(0,1)$ and independent of $x_{ij}$.  Set $t_i = \Sigma_{jj}^{-\frac{1}{2}}x_{ij}\sim N(0,1)$. Then,
\begin{align*}
\|\xi_1^\top(\hat\Sigma - \Sigma)\|_{\infty}& = \max_{j=1}^p |\frac{1}{n}\sum_{i=1}^n\left(x_{ij}x_i^\top\xi_1-\bE[x_{ij}x_i^\top\xi_1]\right)|\\
& =  \max_{j=1}^p |\frac{1}{n}\sum_{i=1}^n\left(w_{ij}x_{ij}-\bE[w_{ij}x_{ij}]\right)| \\
& =  \max_{j=1}^p |\frac{1}{n}\sum_{i=1}^n\left((\Sigma_j\xi_1)(t_{ij}^2-\bE[t_{ij}^2])+(\xi_1^\top\Sigma \xi_1\Sigma_{jj}-(\Sigma_j \xi_1)^2)^{\frac{1}{2}} (z_{ij}t_{ij}-\bE[z_{ij}t_{ij}]\right)|\\
& \leq \lambda_{\Sigma}^{\max} \max_{j=1}^p |\frac{1}{n}\sum_{i=1}^n(t_{ij}^2-\bE[t_{ij}^2])|+\sqrt{\lambda_{\Sigma}^{\max}} \max_{j=1}^p |\frac{1}{n}\sum_{i=1}^n(t_{ij}z_{ij}-\bE[t_{ij}z_{ij}])|
\end{align*}
Apply Bernstein's inequality to each of the average terms in the last display and use a union bound, we obtain the desired probabilistic  bound at $H = \Sigma$. When $H = \Lambda$, we have
\begin{align*}
\|\xi_1^\top(\hat H- H)\|_{\infty}& =\max_{d=1}^D\|\xi_{1[d]}^\top(\hat\Lambda_{[d]}-\Lambda_{[d]})\|_{\infty}.
\end{align*}
Apply the arguments used for proving $H = \Sigma$ to each $\|\xi_{1[d]}^\top(\hat\Lambda_{[d]}-\Lambda_{[d]})\|_{\infty}$ and use a union bound,  we obtain the desired probabilistic  bound at $H = \Lambda$.

\subsubsection*{Event $\mA_4$}
Same as before, we can bound the terms $\|\hat\Sigma - \Sigma\|_{\infty, \infty}$ and  $\rho_{k+1}\|\hat\Lambda - \Lambda\|_{\infty, \infty}$  using  Bernstein's inequality: for a sufficiently large universal constant $C$, with probability approaching 1:
\begin{align*}
&\|\hat\Sigma - \Sigma\|_{\infty, \infty}\leq C\lambda_{\Sigma}^{\max}\sqrt{\frac{\ln p}{n}}\leq  CM\rho_1\sqrt{\frac{\ln p}{n}},\\
&\rho_{k+1}\|\hat\Lambda - \Lambda\|_{\infty, \infty}\leq C\rho_{k+1}\lambda_{\Lambda}^{\max}\sqrt{\frac{\ln p}{n}}\leq (1-\gamma)CM\rho_1\sqrt{\frac{\ln p}{n}}.
\end{align*}
Next, we turn to the expression $\|\Lambda V_{[k]} G V_{[k]}^\top \Lambda-\hat\Lambda V_{[k]} G V_{[k]}^\top \hat\Lambda\|_{\infty,\infty}$ for any positive semi-definite diagonal matrix $G\in \real^{k\times k}$.
\begin{align*}
\zeta(G)=&\|\Lambda V_{[k]} G V_{[k]}^\top \Lambda-\hat\Lambda V_{[k]} GV_{[k]}^\top \hat\Lambda\|_{\infty,\infty}\\
\leq &\underbrace{\|(\Lambda -\hat\Lambda)V_{[k]}G V_{[k]}^\top \Lambda\|_{\infty,\infty}}_{I_{11}}+\underbrace{\|\Lambda V_{[k]}G V_{[k]}^\top (\Lambda -\hat\Lambda)^\top\|_{\infty,\infty}}_{I_{12}}+\underbrace{\|(\Lambda -\hat\Lambda)V_{[k]}G V_{[k]}^\top(\Lambda -\hat\Lambda)\|_{\infty,\infty}}_{I_2}.
\end{align*}  
The first term $I_{11}$
\begin{align*}
I_{11}=&\max_{d_1, d_2}\|(\Lambda_{[d_1]} -\hat\Lambda_{[d_1]})V_{[d_1][k]}G V_{[d_2][k]}^\top \Lambda_{[d_2]}\|_{\infty,\infty}\\
=&\max_{j\in [d_1], l\in [d_2]}|\frac{1}{n}\sum_{i}\left(x_{ij}(\bx_{i[d_1]}^\top V_{[d_1][k]}G V_{[d_2][k]}^\top\Lambda_{[d_2]})_{\ell}-E\left[x_{ij}(\bx_{i[d_1]}^\top V_{[d_1][k]}GV_{[d_2][k]}^\top\Lambda_{[d_2]})_{\ell}\right]\right)
\end{align*}
We set $u_{\ell} = (V_{[d_1][k]}GV_{[d_2][k]}^\top\Lambda_{[d_2]})_{\ell}$. Set $w_{ij\ell} = x_{i[d_1]}^\top u_{\ell}\sim N(0, u_{\ell}^\top\Lambda_{[d_1]}u_{\ell})$ and $\bE[w_{ij\ell}x_{ij}] = \Lambda_{j[d_1]}u_{\ell}$.  We write $x_{ij} = \Lambda_{jj}^{\frac{1}{2}}t_{ij}$, $w_{ij\ell}$ as $w_{ij\ell} = \frac{\Lambda_{j[d_1]}u_{\ell}}{\Lambda_{jj}}x_{ij}+\left(u_{\ell}^\top\Lambda_{[d_1]}u_{\ell}-\frac{(\Lambda_{j[d_1]}u_{\ell})^2}{\Lambda_{jj}}\right)^{\frac{1}{2}}z_{ij}$ with  $z_{ij\ell} = \left(u_{\ell}^\top\Lambda_{[d_1]}u_{\ell}-\frac{(\Lambda_{j[d_1]}u_{\ell})^2}{\Lambda_{jj}}\right)^{-\frac{1}{2}}(w_{ij\ell}- \frac{\Lambda_{j[d_1]}u_{\ell}}{\Lambda_{jj}}x_{ij})\sim N(0,1)$. Then, following the same arguments for analyzing $\mA_3$, with $\xi_1$ replaced by $u$, we have
\begin{align*}
I_{11} \leq  &\max_{d_1, d_2}\max_{j\in [d_1]}\max_{\ell\in [d_2]} \left(|\Lambda_{j[d_1]}u_{\ell}|\frac{1}{n}\sum_{i=1}^n(t_{ij}^2-\bE[t_{ij}^2])|\right)\\
&+\max_{d_1, d_2}\max_{j\in [d_1]}\max_{\ell\in [d_2]} \left(\left(u_{\ell}^\top\Lambda_{[d_1][d_1]}u_{\ell}\right)^{\frac{1}{2}}|\frac{1}{n}\sum_{i=1}^n(t_{ij}z_{ij\ell}-\bE[t_{ij}z_{ij\ell}])|\right)
\end{align*}
Let $e_j\in \real^p$ be a vector with $1$ at location $j$ and $0$ at other entries. Notice that
\begin{align*}
|\Lambda_{j[d_1]}u_{\ell}| &\leq \lambda_{G}^{\max}|\Lambda_{j[d_1]}V_{[d_1][k]}V_{[d_2][k]}^\top \Lambda_{[d_2]\ell}|\\
&\leq \lambda_{G}^{\max}\|\Lambda_{j[d_1]}V_{[d_1][k]}\|_2\| \Lambda_{\ell[d_2]} V_{[d_2][k]}\|_2\\
& \leq   \lambda_{G}^{\max} \max_j\|e_j\Lambda V_{[k]}\|_2^2\\
& = \lambda_{G}^{\max}Tr(V_{[k]}^\top \Lambda e_j e_j^\top\Lambda V_{[k]} )\\
& \leq  \lambda_{G}^{\max}(e_j^\top \Lambda e_j) \|V_{[k]}V_{[k]}^{\top}\Lambda\|_{op}=\lambda_{G}^{\max}\|V_{[k]}V_{[k]}^{\top}\Lambda\|_{op}
\end{align*}
For any vector $u = \sum_{j=1}^p \alpha_j v_j$, we have
\[
\|u\|_2^2\geq M  \|u\|_\Lambda^2 \geq \|\sum_{j\leq k}\alpha_j v_j\|_{\Lambda}^2 = \|V_{[k]}V_{[k]}^{\top}\Lambda u\|_{M}^2\geq \frac{1}{M} \|V_{[k]}V_{[k]}^{\top}\Lambda u\|^2.
\]
Hence, we have $|\Lambda_{j[d_1]}u_{\ell}|\leq \lambda_G^{\max} M$. Similarly, we have
\begin{align*}
(u_{\ell}^\top\Lambda_{[d_1][d_1]}u_{\ell})^{\frac{1}{2}}&\leq \lambda_{G}^{\max}|\Lambda_{\ell [d_2]}V_{[d_2][k]}V_{[d_1][k]}^\top\Lambda_{[d_1][d_1]}V_{[d_1][k]}V_{[d_2][k]}^\top \Lambda_{[d_2]\ell}|^{\frac{1}{2}}\\
&\leq   \lambda_{G}^{\max}\sqrt{\|V_{[d_1][k]}^\top\Lambda_{[d_1][d_1]}V_{[d_1][k]}\|_{op}\|V_{[d_2][k]}^\top\Lambda_{[d_2]\ell} \Lambda_{\ell [d_2]}V_{[d_2][k]}\|_{op}}\\
& \leq   \lambda_{G}^{\max}\max_d\|V_{[d][k]}^\top\Lambda_{[d]\ell} \Lambda_{\ell [d]}V_{[d][k]}\|_{op}\\
& \leq  \lambda_{G}^{\max}\max_{d}\|\Lambda_{\ell [d]}V_{[d][k]}\|_{2}^2 \leq \lambda_{G}^{\max}\|\Lambda_{\ell}V_{[k]}\|_{2}^2\leq \lambda_{G}^{\max}M.
\end{align*}
We thus obtain the desired probabilistic bound for $I_{11}$ by applying the Bernstein's inequality to bound the terms on $t_{ij}^2$ and $z_{ij}t_{ij}$ and using a union bound: for a large universal constant $C$, with probability approaching 1, we have $I_{11}\leq C\lambda_G^{\max} M\sqrt{\frac{\ln p}{n}}$.  Following exactly the same arguments, we have $I_{12}\leq  C\lambda_G^{\max} M\sqrt{\frac{\ln p}{n}}$.

For the last term $I_2$, let $\lambda_{G,r}$ be the $r^{th}$ largest diagonal element of $G$. Then,
{\footnotesize
\begin{align*}
I_2 \leq \sum_{r=1}^k\max_{d_1, d_2}\max_{j\in [d_1],\ell\in [d_2]} \lambda_{G,r}\|\underbrace{\left(\frac{1}{n}\sum_{i=1}^n\left(x_{ij}x_{i[d_1]}^\top V_{[d_1]r}-\bE[x_{ij}x_{i[d_1]}^\top V_{[d_1]r}]\right) \right)}_{I_3}\underbrace{\left(\frac{1}{n}\sum_{i=1}^n\left(x_{i\ell}x_{i[d_2]}^\top V_{[d_2]r}-\bE[x_{i\ell}x_{i[d_2]}^\top V_{[d_2]r}]\right) \right)}_{I_4}|
\end{align*}
}
Following the same arguments for bounding $I_{11}, I_{12}$, we obtain that for a sufficiently large universal constant $C$,  with probability approaching 1, we have:
\[
|I_3|\vee |I_4|\leq CM\sqrt{\frac{\ln p}{n}}
\]
Since $k$ is a constant and $\lambda_{G,r}\sqrt{\frac{\ln p}{n}}\rightarrow 0$ when $G = \Id$ or $G = \Gamma_{[k]}$. We acquire that $\zeta(G)\leq C\lambda_{G}^{\max}M\sqrt{\frac{\ln p}{n}}$ for a sufficiently large universal constant $C$ with probability approaching 1. Set  $G = \Id$ or $G = \Gamma_{[k]}$, we can bound $\|\Lambda V_{[k]}\Gamma_{[k]} V_{[k]}^\top \Lambda -\hat\Lambda V_{[k]}\Gamma_{[k]} V_{[k]}^\top\hat\Lambda  \|_{\infty,\infty}$ and  $\rho_{k+1}\|\Lambda V_{[k]} V_{[k]}^\top \Lambda -\hat\Lambda V_{[k]} V_{[k]}^\top\hat\Lambda  \|_{\infty,\infty}$ with the desired rate. We have now successfully bounded each of the for term in $\mA_4$, thus, $\mA_4$ happens with probability approaching 1.

\subsubsection*{Event $\mA_5$} 
Let $S = \{j:\|\xi_{j[k]}\|_2\neq 0\}$. According to the proofs of Lemma B.4 in \cite{gao2021sparse}:
\begin{align}
&\|\Lambda^{\frac{1}{2}}(\tilde V_{[k]}-V_{[k]})\|_{op}\leq \|\Lambda^{\frac{1}{2}}V_{[k]}\|_{op}\|(V^\top_{[k]}\hat\Lambda V_{[k]})^{\frac{1}{2}}-\Id\|_{op}\|(V^\top_{[k]}\hat\Lambda V_{[k]})^{-\frac{1}{2}}\|_{op}\label{eq:A5e1}\\
&\|\tilde \Gamma_{[k]} - \Gamma_{[k]}\|_{op}\leq \|(V^\top_{[k]}\hat\Lambda V_{[k]})^{\frac{1}{2}}-\Id\|_{op}\|\Gamma_{[k]}(V^\top_{[k]}\hat\Lambda V_{[k]})^{\frac{1}{2}}\|_{op}+\|\Lambda_{[k]}\|_{op} \|(V^\top_{[k]}\hat\Lambda V_{[k]})^{\frac{1}{2}}-\Id\|_{op}\label{eq:A5e2}\\
&\|V_{[k]}^\top\hat\Lambda V_{[k]}-\Id\|_{op}\leq \|\Lambda_{SS}^{-\frac{1}{2}}\hat\Lambda_{SS}\Lambda^{-\frac{1}{2}}_{SS}-\Id\|_{op}\label{eq:A5e3}.
\end{align}
We now continue with bounding (\ref{eq:A5e3}). Set $S_d = S\cap [d]$, then, by Proposition (D.1) from \cite{ma2013sparse}, for a sufficiently large  positive constant $C$, with probability approaching 1, we have
\begin{align}
\|V_{[k]}^\top\hat\Lambda V_{[k]}-\Id\|_{op}\leq \|\Lambda_{SS}^{-\frac{1}{2}}\hat\Lambda_{SS}\Lambda^{-\frac{1}{2}}_{SS}-\Id\|_{op}=\max_{d=1}^D \|\Lambda_{S_dS_d}^{-\frac{1}{2}}\hat\Lambda_{S_dS_d}\Lambda^{-\frac{1}{2}}_{S_dS_d}-\Id\|_{op}\leq C\sqrt{\frac{s+\ln p}{n}}.
\end{align}
Notice that for large $n$, with probability approaching 1, we have
\begin{align*}
\|(V_{[k]}^\top\hat\Lambda V_{[k]})^{\frac{1}{2}}-\Id\|_{op}&\leq \frac{\|(V_{[k]}^\top\hat\Lambda V_{[k]})-\Id\|_{op}}{\sqrt{\lambda^{\min}_{(V_{[k]}^\top\hat\Lambda V_{[k]})^{\frac{1}{2}}+\Id}}}\\
&\leq \frac{\|(V_{[k]}^\top\hat\Lambda V_{[k]})-\Id\|_{op}}{\sqrt{1+\sqrt{1-\|(V_{[k]}^\top\hat\Lambda V_{[k]})-\Id\|_{op}}}}\leq C\sqrt{\frac{s+\ln p}{n}}.
\end{align*}
Combine the bound on $\|(V_{[k]}^\top\hat\Lambda V_{[k]})-\Id\|_{op}$ and $\|(V_{[k]}^\top\hat\Lambda V_{[k]})^{\frac{1}{2}}-\Id\|_{op}$ with (\ref{eq:A5e1}) and (\ref{eq:A5e2}), we obtain the desired probabilistic bounds: for a sufficiently large universal constant $C$, when $n\rightarrow\infty$, with probability approaching 1, we have
\begin{align*}
&\|\Lambda^{\frac{1}{2}}(\tilde V_{[k]}-V_{[k]})\|_{F}\leq \sqrt{k}\|\Lambda^{\frac{1}{2}}(\tilde V_{[k]}-V_{[k]})\|_{op}\leq C\sqrt{\frac{\sqrt{k}(s+\ln p)}{n}}\\
&\|\tilde \Gamma_{[k]} - \Gamma_{[k]}\|_{F}\leq \sqrt{k}\|\tilde \Gamma_{[k]} - \Gamma_{[k]}\|_{op}\leq \rho_1C\sqrt{\frac{k(s+\ln p)}{n}}\\
&\rho_{k+1}\|V_{[k]}^\top\hat\Lambda V_{[k]}-\Id\|_{op}\leq   (1-\gamma)\rho_1C\sqrt{\frac{k(s+\ln p)}{n}}.
\end{align*}
Hence, we $\mA_5$ happens with probability approaching 1 for large $C$.

\subsection{Proof of Proposition \ref{prop:prop2}}
\begin{proof}
We denote $\delta(\beta)$ as $\delta$ for convenience. Define $\|\beta\|_{\Lambda}^2 = \beta^\top\Lambda\beta$. By definition,  $\sum_{j=1}^p\alpha_j^2=\|\beta\|_{\Lambda}^2\geq \frac{1}{M}$. Further,
\begin{align*}
&M\|\beta-\alpha_1\xi_1\|_{\Lambda}^2\geq \|\beta - \alpha_1\xi_1\|_2^2\geq \|\beta-(1-\delta)\xi_1\|_2^2=2\delta-\delta^2\overset{(a)}{\Rightarrow} \frac{\sum_{j\geq 2}\alpha_j^2\mu_j}{\delta} \geq \frac{2}{M}(1-\frac{\delta}{2}).\\
&2M\delta\geq \|\beta-\xi_1\|_{\Lambda}^2=\|(\alpha_1-1)\xi_1+\sum_{j\geq 2}\alpha_j\xi_j\|_{\Lambda}^2\overset{(b)}{\Rightarrow} (1-\alpha_1)^2\mu_1+\sum_{j\geq 2}\alpha_j^2\mu_j\leq 2M\delta.
\end{align*}
At steps (a) and (b), we have used the fact that $\xi_i^\top\Lambda \xi_j = 0$ for all $i\neq j$. As a result:
\begin{align*}
&\|\beta-\alpha_{1}\xi_1\|_{\Lambda}^2=\sum_{j\geq 2}\alpha_{j}^2\mu_j,\\
&\|(\alpha_1-1)\xi_1+\sum_{j\geq 2}\alpha_j\xi_i\|_{\Lambda}^2=(1-\alpha_1)^2\mu_1+\sum_{j\geq 2}\alpha_j^2\mu_j.
\end{align*}
\end{proof}
\subsection{Proof of Proposition \ref{prop:prop3}}
\begin{proof}
We drop the dependence on $\beta$ and denote $m_1(\beta)$, $\delta(\beta)$ and $m_2(\beta)$ as  $m_1$, $\delta$ and $m_2$ for convenience.  Since $\rho_1$ is the largest mCCA correlation, by definition,
\[
\frac{\beta^\top\Sigma\beta}{\beta^\top\Lambda\beta}\leq \rho_1\Rightarrow m_2\leq \beta^\top\Lambda\beta.
\]
Now, we consider the case when $\delta\leq \frac{c^2}{8M^2}$.
Decompose $\beta$ as $\beta = \sum_{j=1}^p\alpha_j \xi_j$, we have $m_{2} = \alpha_{1}^2\mu_1+\frac{\sum_{j\geq 2}\rho_j\alpha_{j}^2\mu_j}{\rho_1}$.  Then,
\[
\frac{m_{2}}{\beta^\top\Lambda\beta}=1-\frac{\sum_{j\geq 2}(1-\frac{\rho_j}{\rho_1})\alpha_j^2\mu_j}{\alpha_1^2\mu_1+\sum_{j\geq 2}\alpha_j^2\mu_j}\geq 1-\frac{\sum_{j\geq 2}\alpha_j^2\mu_j}{\alpha_1^2\mu_1+\sum_{j\geq 2}\alpha_j^2\mu_j}.
\]
Set $z = \sum_{j\geq 2}\alpha_j^2\mu_j\geq 0$. By Proposition \ref{prop:prop2} (\ref{eq:prop2eq3}), we know
\begin{equation}
\label{eq:prop3e1}
(1-\alpha_1)^2\mu_1 + z \leq 2M\delta\Rightarrow \alpha_1\in [1-\sqrt{\frac{2M\delta-z}{\mu_1}},1+\sqrt{\frac{2M\delta-z}{\mu_1}}], \;z\in [0, 2M\delta].
\end{equation}
Set $\zeta(z) = \frac{\sum_{j\geq 2}\alpha_j^2\mu_j}{\alpha_1^2\mu_1+\sum_{j\geq 2}\alpha_j^2\mu_j}=\frac{z}{\alpha_1^2\mu_1+z}$. Recall that $c\leq \frac{1}{2}$. By (\ref{eq:prop3e1}), we know
\begin{align}
\label{eq:prop3e2}
\zeta(z)&\leq \frac{z}{(1+\sqrt{\frac{2M\delta-z}{\mu_1}})^2\mu_1+z}\notag\\
&\leq \frac{z}{\mu_1+2M\delta+2\sqrt{\mu_1}\sqrt{2M\delta}}\notag\\
&\leq \frac{z}{\mu_1(1+\sqrt{\frac{2M\delta}{\mu_1}})^2}\overset{(a_1)}{\leq} \frac{\frac{c^2}{4}}{(1+\frac{c}{2})^2}\leq \frac{c^2}{2}.
\end{align}
We have used $\mu_1\geq \frac{1}{M}$ at step $(a_1)$. Hence, we obtain $m_2\geq (1-\frac{c^2}{2})\beta^\top\Lambda\beta$. Also, by (\ref{eq:prop3e1}), we  have
\begin{align*}
\min(\frac{m_2}{\mu_1},\frac{\beta^\top\Lambda\beta}{\mu_1})\geq \alpha_1^2 \geq (1-\frac{c}{2})^2,
\end{align*}
as well as
\begin{align*}
\frac{\beta^\top\Lambda\beta}{\mu_1}&=\alpha_1^2+\frac{\sum_{j\geq 2}\alpha_j^2\mu_j}{\mu_1}\\
 &\leq (1+\sqrt{\frac{2M\delta-z}{\mu_1}})^2+\frac{z}{\mu_1}\\
& \leq 1+\frac{2M\delta}{\mu_1}+2\sqrt{\frac{2M\delta}{\mu_1}}=(\sqrt{\frac{2M\delta}{\mu_1}}+1)^2\leq (1+\frac{c}{2})^2.
\end{align*}
Hence, when $\delta\leq \frac{c^2}{8M^2}$, we obtain
\[
(1-\frac{c^2}{2})\beta^\top\Lambda\beta\leq m_2\leq \beta^\top\Lambda\beta\leq (1+\frac{c}{2})^2\mu_1,\; \min(\beta^\top\Lambda\beta, m_2)\geq (1-\frac{c}{2})^2\mu_1.
\]
\end{proof}
\subsection{Proof of Proposition \ref{prop:prop4}}
\begin{proof}
We are interested in upper bound $\|h(\beta)\|_1$. By construction:
\begin{equation}
\label{eq:prop4eq1}
\|\xi_1\|_1+\omega(\beta)=\|\beta\|_1 = \|\xi_1+h(\beta)\|_1= \|(\xi_1+h(\beta))_{[s]}\|_1 + \|(\xi_1+h(\beta))_{\bar{[s]}}\|_1.
\end{equation}
By Assumption \ref{ass:sparsity}:
\begin{equation}
\label{eq:prop4eq2}
\|(\xi_1+h(\beta))_{[s]}\|_1 \geq \|\xi_1\|_1 - \|h_{[s]}(\beta)\|_1, \;\|(\xi_1+h(\beta))_{\bar{[s]}}\|_1 = \|h(\beta)_{\bar{[s]}}\|_1
\end{equation}
Also, notice that 
\begin{equation}
\label{eq:prop4eq3}
\|h(\beta)\|_2^2=2 - 2\xi_1^\top\beta_t = 2\delta(\beta),\; \|h(\beta)_{[s]}\|_1\leq \sqrt{s}\|h(\beta)\|_2 = \sqrt{2s\delta(\beta)}.
\end{equation}
Combine (\ref{eq:prop4eq1}) -(\ref{eq:prop4eq3}), we have
\[
\|h(\beta)\|_1\leq 2\|h(\beta)_{[s]}\|_1+\omega(\beta) \leq 2\sqrt{2s\delta(\beta)}+\omega(\beta).
\]
\end{proof}
\subsection{Proof of Proposition \ref{prop:prop5}}
\begin{proof}
For convenience, we leave out the $\beta$ argument in $\delta(\beta)$ and denote $\delta(\beta)$ as $\delta$. Since $\beta = (1-\delta)\xi_1+\sqrt{1- (1-\delta)^2}v$ for some  $v^\top \xi_1 = 0$ and $\|v\|_2=1$, we have $1=\|\tilde\beta\|_2^2=[(1-\delta)(1-\epsilon)+\epsilon_1]^2+[1- (1-\delta)^2](1-\epsilon)^2$, and consequently,
\begin{equation}
\label{eq:prop3_eq1}
\epsilon_1 = \sqrt{1-[1- (1-\delta)^2](1-\epsilon)^2}-(1-\delta)(1-\epsilon).
\end{equation}
Set $\zeta(\epsilon) = \epsilon_1 - \epsilon-\epsilon(1-\epsilon)\delta +\frac{\epsilon(1-\epsilon)}{2}\delta^2$, then $\zeta(\epsilon)$ is a concave function with non-positive second derivative:
 \begin{align*}
\zeta'(\epsilon)&=\frac{(1-\epsilon)[1- (1-\delta)^2]}{ \sqrt{1-[1- (1-\delta)^2](1-\epsilon)^2}}+(1-\delta)-1-(1-2\epsilon)\delta+(1-2\epsilon)\frac{\delta^2}{2},\\
\zeta''(\epsilon)&=2\delta - \delta^2-\frac{[1- (1-\delta)^2]}{ \sqrt{1-[1- (1-\delta)^2](1-\epsilon)^2}}-\frac{[1- (1-\delta)^2]^2(1-\epsilon)^2}{ (1-[1- (1-\delta)^2](1-\epsilon)^2)^{\frac{3}{2}}}.
\end{align*}
Because $\sqrt{1-[1- (1-\delta)^2](1-\epsilon)^2}\leq 1$ and $[1- (1-\delta)^2]=2\delta-\delta^2$, we must have $\zeta''(\epsilon)\leq 0$ and $\zeta(\epsilon)$ is concave. Since $\zeta(0) = \zeta(1) = 0$, this produce the desired lower bound on $\epsilon_1$ with $\epsilon_1\geq \epsilon+\epsilon(1-\epsilon)\delta -\frac{\epsilon(1-\epsilon)}{2}\delta^2$, as stated in Proposition \ref{prop:prop5}. For the upper bound, set $z = [1-(1-\delta)^2](1-\epsilon)^2$, then,
\[
 \sqrt{1-z}\leq \sqrt{(1-\frac{z}{2}-\frac{z^2}{8})^2} = \sqrt{1+\frac{z^2}{4}+\frac{z^4}{64}-z-\frac{z^2}{4}+\frac{z^3}{8}}=\sqrt{1-z+\frac{z^3}{8}+\frac{z^4}{64}}.
\]
Thus,
\begin{align*}
\epsilon_1 &\leq (1-\frac{z}{2}-\frac{z^2}{8}) - (1-\delta)(1-\epsilon)\\
& = 1-\frac{2\delta-\delta^2}{2}(1-\epsilon)^2-\frac{(2\delta-\delta^2)^2(1-\epsilon)^4}{8} - (1-\delta)(1-\epsilon)\\
& = \epsilon+\frac{1}{2}\epsilon(1-\epsilon)(2\delta-\delta^2)+\frac{1}{2}\delta^2(1-\epsilon)(1-(1-\epsilon)^3)+\frac{1}{2}(1-\epsilon)^4\delta^3-\frac{1}{8}(1-\epsilon)^4\delta^4\\
&\leq \epsilon+\frac{1}{2}\epsilon(1-\epsilon)(2\delta-\delta^2)+\frac{1}{2}\delta^2\epsilon(3+3\epsilon+\epsilon^2)+\frac{1}{2}\delta^3.
\end{align*}
This produces the upper bound for $\epsilon_1$.
\end{proof}
\subsection{Proof of Lemma \ref{lem:basic_bound1}}
\begin{proof}
We set $W = \hat H - H$. On the event $\mA(1)\cap \mA(k')$, we obtain
\begin{equation}
\label{eq:basic1eq1}
|\xi_1^\top W\zeta |\leq \|W\xi_1\|_{\infty}\|\zeta\|_{1}\leq C\lambda_H^{\max}\sqrt{\frac{\ln p }{n}}\|\zeta\|_1.
\end{equation}
We now turn to $|\zeta^\top W \zeta|$. We divide the index set $\{1,\ldots, p\}$ into $\cup_{\ell=1}^{\lceil\frac{p}{s}\rceil} J_\ell$ where $J_1$ contains the indexes $j$ for $(k's)$ entries with the largest $|\zeta_j|$, and $J_{\ell}$ contains $(k's)$ the largest-entry indexes excluding sets before it with the last index set $J_{\lceil\frac{p}{k's}\rceil}$ possibility having less than $(k's)$ elements. Then, on the event $\mA(1)\cap\mA(k')$, we have
\begin{align}
|\zeta^\top W\zeta|=&|\sum_{t_1=1}^{\lceil\frac{p}{k's}\rceil}\sum_{t_2}^{\lceil\frac{p}{k's}\rceil} \|\zeta_{J_{t_1}}\|_2\|\zeta_{J_{t_2}}\|_2\|W_{J_{t_1}J_{t_2}}\|_{op}|\notag\\
\leq&(\max_{t_1, t_2}\|W_{J_{t_1}, J_{t_2}}\|_{op})\left(\sum_{t_1=1}^{\lceil\frac{p}{k's}\rceil} \|\zeta_{J_{t_1}}\|_2\right)\left(\sum_{t_1=1}^{\lceil\frac{p}{k's}\rceil} \|\zeta_{J_{t_2}}\|_2\right)\notag\\
\leq& (\max_{t_1, t_2}\|W_{J_{t_1}, J_{t_2}}\|_{op})\left(\|\zeta_{J_{1}}\|_2+\sum_{t=2}^{\lceil\frac{p}{k's}\rceil} (\frac{\|\zeta_{J_{t_1-1}}\|_1}{\sqrt{k's}})\right)^2\notag\\
\leq&  (\max_{t_1, t_2}\|W_{J_{t_1}, J_{t_2}}\|_{op})\left(\|\zeta_{J_{1}}\|_2+\frac{\|\zeta\|_1}{\sqrt{k's}}\right)^2\leq C\lambda_H^{\max}\sqrt{\frac{k's\ln p}{n}}(\|\zeta\|_2^2+\frac{\|\zeta\|_1^2}{k's})\label{eq:basic1eq2}.
\end{align}
When $\zeta = h(\beta)$, by Proposition \ref{prop:prop4}, we have $\|h(\beta)\|_1\leq 2\sqrt{2s\delta(\beta)}+\omega(\beta) \leq 2\sqrt{2}(\sqrt{s\delta(\beta)}+\bomega(\beta))$ and $\|h(\beta)\|_2^2 = 2\delta(\beta)$. Combine them with (\ref{eq:basic1eq1}) and (\ref{eq:basic1eq2}), we obtain:
\begin{align*}
|\xi_1^\top W h |&\leq 2\sqrt{2}C\lambda^{\max}_H\sqrt{\frac{s\ln p}{n}}(\sqrt{\delta(\beta)}+\frac{\bomega(\beta)}{\sqrt{s}}),\\
|h^\top W h|&\leq C\lambda_H^{\max}\sqrt{\frac{k's\ln p}{n}}(2\delta(\beta)+\frac{\left(2\sqrt{2}(\sqrt{s\delta(\beta)}+\bomega(\beta)\right)^2}{k's})\\
&\leq C'\lambda_H^{\max}\sqrt{\frac{k's\ln p}{n}}(\delta(\beta)+\frac{\bomega(\beta)^2}{k's}).
\end{align*}
where $C'$ is a sufficiently large universal constant. Hence, (\ref{basic1main1}) - (\ref{basic1main4}) hold.

\subsection{Proof of Lemma \ref{lem:basic_bound2}}
Notice that
\begin{align*}
|\xi_1^\top \hat H(\beta - \alpha_1\xi_1)|\overset{(a_1)}{=}&|\xi_1^\top (\hat H-H)(\beta - \alpha_1\xi_1)| +\cancel{ |\xi_1^\top H(\beta - \alpha_1\xi_1)|}\\
\leq &  |\xi_1^\top (\hat H-H)h|+|1-\alpha_1||\xi_1^\top (\hat H-H)\xi_1|,
\end{align*}
where we have used the fact that $\xi_1^\top H(\beta-\alpha_1\xi_1) = \sum_{j\geq 2}\alpha_j \xi_1^\top H\xi_j =0$ for $H\in \{\Lambda,\Sigma\}$ at step $(a_1)$. We can bound the second term in the above equation using Proposition \ref{prop:prop2} (\ref{eq:prop2eq3}) and on the event $\mA(k')$, we have:
\begin{equation}
\label{eq:basic1eq02}
|1-\alpha_1||\xi_1^\top (\hat H-H)\xi_1|\leq \sqrt{2M\delta(\beta)}\frac{\lambda_H^{\max}}{M^3}\sqrt{\frac{\ln p}{n}}.
\end{equation}
Combine it with (\ref{basic1main4}), we have $|\xi_1^\top \hat H(\beta - \alpha_1\xi_1)|\leq (2\sqrt{2}C+\frac{1}{M^2})\sqrt{\frac{\ln p}{n}}\lambda_{H}^{\max}(\sqrt{\delta(\beta)}+\frac{\bomega(\beta)}{\sqrt{s}})$ and proved (\ref{basic2main1}). We now bound $|\beta^\top\hat H\beta - F_H|$. We look at $\beta^\top\hat H\beta$ for $H\in \{\Lambda, \Sigma\}$. Since $\beta= \sum_{j\geq 1}\alpha_j \xi_j$,  we have
\begin{align}
\label{eq:basic1eq2}
&\beta^\top\hat H\beta\notag\\
 =&\alpha_1^2\xi_1^\top\hat H\xi_1+2\alpha_1\xi_1^\top\hat H(\beta-\alpha_1\xi_1)+(\beta - \alpha_1\xi_1)^\top\hat H(\beta-\alpha_1\xi_1)\notag\\
 \overset{(a_2)}{=}&\alpha_1^2\xi_1^\top\hat H\xi_1+\sum_{j\geq 2}\alpha_j^2\xi_j^\top H\xi_j+W_H =  F_H +W_H,
\end{align}
where $W_H = 2\alpha_1\xi_1^\top \hat H(\beta-\alpha_1\xi_1)+(\beta - \alpha_1\xi_1)^\top(\hat H- H)(\beta-\alpha_1\xi_1)$, and at step $(a_2)$, we have used the relationship that for $H\in \{\Lambda, \Sigma\}$:
\[
(\beta - \alpha_1\xi_1)^\top\hat H(\beta-\alpha_1\xi_1)=\sum_{j\geq 2}\alpha_j^2\xi_j^\top H\xi_j+(\beta - \alpha_1\xi_1)^\top(\hat H-H)(\beta-\alpha_1\xi_1).
\]
Rearrange terms in $W_H$, we have
\begin{equation}
\label{eq:basic1eq3}
W_H = (1-\alpha_1^2)\xi_1^\top(\hat H - H)\xi_1+2\xi_1^\top(\hat H - H)h+h^\top(\hat H - H)h.
\end{equation}
By Proposition \ref{prop:prop2} (\ref{eq:prop2eq3}), we have 
\begin{equation}
\label{eq:basic1eq4}
|1-\alpha_1^2|\leq (1+\sqrt{2M\delta(\beta)})^2-1<5M\sqrt{\delta(\beta)}.
\end{equation}
Combine (\ref{eq:basic1eq3}) with (\ref{eq:basic1eq4}) and our bounds Lemma \ref{lem:basic_bound1} (\ref{basic1main3}) - (\ref{basic1main4}):
\begin{align}
\label{eq:basic1eq5}
|W_H|& \leq  5M\sqrt{\delta(\beta)}\frac{\lambda_H^{\max}}{M^3}\sqrt{\frac{\ln p}{n}}+C\lambda_H^{\max}\sqrt{\frac{s\ln p}{n}}(2\sqrt{\delta(\beta)}+\frac{2\bomega(\beta)}{\sqrt{s}}+\sqrt{k'}\delta(\beta)+\frac{\bomega(\beta)^2}{\sqrt{k'}s})
\end{align}
Combine (\ref{eq:basic1eq2}) with (\ref{eq:basic1eq5}), there is a sufficiently large universal constant $C$ to make (\ref{basic2main2}) hold.
\end{proof}

\subsection{Proof of Lemma \ref{lem:basic_bound3}}
On the event $\mA_2$, we have $|\xi_1^\top(\hat H -H)\xi_1|\leq \frac{\lambda_H^{\max}}{M^3}\sqrt{\frac{\ln p}{n}}, \; H\in \{\Lambda, \Sigma\}$, with $\lambda_{\Lambda}^{\max}\leq M$ and $\lambda_{\Sigma}^H \leq \rho_1 M$. Hence, when $\frac{\ln p}{n}\leq c^2$, we have
\begin{align}
|\hat\rho_1 - \rho_1| &= |\frac{\xi_1^\top\Sigma\xi_1\hat\mu_1 -\xi_1^\top\hat\Sigma\xi_1\mu_1 }{\mu_1\hat\mu_1}|\notag\\
& \leq\frac{\xi_1^\top\Sigma\xi_1|\hat\mu_1-\mu_1|}{\mu_1\hat\mu_1}+ \frac{|\xi_1^\top\hat\Sigma\xi_1-\xi_1^\top\Sigma\xi_1|\mu_1 }{\mu_1\hat\mu_1}\notag\\
&\leq \frac{\rho_1\sqrt{\frac{\ln p}{n}}}{M^2(\mu_1-\frac{1}{M^2}\sqrt{\frac{\ln p}{n}})}+  \frac{C\rho_1\sqrt{\frac{\ln p}{n}}}{M^2(\mu_1-\frac{1}{M^2}\sqrt{\frac{\ln p}{n}})}\notag\\
&\leq \frac{\rho_1\sqrt{\frac{\ln p}{n}}}{M(1-c)}+  \frac{\rho_1\sqrt{\frac{\ln p}{n}}}{M(1-c)}\leq \frac{4}{M}\rho_1\sqrt{\frac{\ln p}{n}}
.\label{eq:basic_bound3eq2}
\end{align}
From Lemma \ref{lem:basic_bound1} (\ref{basic1main1}) and $\frac{\|\beta\|_1}{\sqrt{s}}\leq (1+c_{B_1})$, we obtain
\begin{align*}
&|\beta^\top(\hat H - H)\beta|\leq C\lambda_{H}^{\max}\sqrt{\frac{s\ln p}{n}}(\|\beta\|_2^2+
\frac{\|\beta\|_1^2}{s})\leq 2\lambda_{H}^{\max}C\sqrt{\frac{s\ln p}{n}}(1+c_{B_1})^2.
\end{align*}
Thus, there exists sufficiently large constant $C$ such that (\ref{basic3main1}) and (\ref{basic3main2}) hold. 

\section{Initialization used in empirical studies}
\label{app:algorithm}
Initialization methods based on convex relaxation as described in Section \ref{subsec:init} have better theoretical guarantees, they are computationally expensive when $p$ is large. In our empirical studies, we will use the following  initialization approach that scale well with the data dimensions:
\begin{itemize}
\item We soft-threshold the empirical covariance matrix $\hat\Sigma$ and keep only $m^2$ non-zero entries in $\hat\Sigma$ with $m = \lceil\frac{n}{\ln p}\rceil$. Let $\tilde\Sigma$ be the resulting matrix.
\item For each block $d$, $d=1,\ldots,D$, we keep $\lceil \frac{n}{KD}\rceil$ non-zero entries with largest  $\|\tilde\Sigma_{j,\bar{[d]}}\|_2$ for $j\in [d]$, with $\bar{[d]}$ being the complement of $[d]$ and $K = 4$ by default.
\item Let $S$ to the selected feature index from each block, and our initialization is $\beta_{j} = 0$ for $j\in S^c$ and for features in $S$, we let $\beta_{S} = \tilde\beta$ where $\tilde\beta$ is the estimated mCCA direction using $(\hat\Sigma_{SS}, \hat\Lambda_{SS})$  and regularized estimation:
\[
\max_{\beta}\frac{\beta^\top\hat\Sigma_{SS}\beta}{\beta^\top\tilde\Lambda_{SS}\beta},\; \tilde\Lambda_{SS} = (1-\tau)\hat\Lambda_{SS}+\tau \diag\{\hat\sigma_j^2, j\in S\}
\]
and for some  optimal $\tau$ estimated according to \cite{schafer2005shrinkage}:
\[
\tau = \frac{\sum_{j,\ell\in S,j\neq \ell}\widehat{Var}(\hat\Sigma_{j\ell})}{\sum_{j\neq \ell}\hat\Sigma_{j\ell}^2},
\]
where $\widehat{Var}(\hat\Sigma_{j\ell})$ is the empirical estimation of the estimation variance in $\hat\Sigma_{j\ell}$ as described in \cite{schafer2005shrinkage}.
\end{itemize}
\end{document}